\documentclass[a4paper,unpublished,onecolumn,11pt]{quantumarticle}
\pdfoutput=1

\usepackage{amsmath}
\usepackage{amsthm}
\usepackage{thmtools}
\usepackage{amssymb}
\usepackage{amsfonts}
\usepackage{mathtools}
\usepackage{physics}    
\usepackage{mathdots}   
\usepackage{mathrsfs}   
\usepackage{stmaryrd}   

\usepackage{tikz}
\usetikzlibrary{decorations.markings}
\usepackage{tikzit}
\usetikzlibrary{cd}

\usepackage[
  style=alphabetic,
  sorting=nyt,
  backend=biber,
  maxnames=6,
  maxcitenames=2,
  url=false,
  doi=true,
  isbn=false,
  eprint=false
]{biblatex}

\usepackage{hyperref}
\hypersetup{
  colorlinks=true,
  linkcolor=cyan,
  citecolor=gray
}
\newtheorem{theorem}{Theorem}
\newtheorem{definition}[theorem]{Definition}
\newtheorem{proposition}[theorem]{Proposition}
\newtheorem{lemma}[theorem]{Lemma}
\newtheorem{corollary}[theorem]{Corollary}

\declaretheoremstyle[
headfont=\normalfont\bfseries\color{dark-gray},
bodyfont=\normalfont,
notefont=\normalfont\bfseries\color{dark-gray},
notebraces={}{},
headpunct={.},
qed=\qedsymbol,
mdframed={
  linewidth=1.5,
  linecolor=gray,
  hidealllines=true,
  leftline=true,
  skipabove=0,
  innerleftmargin=2mm,
  innerrightmargin=0,
  innertopmargin=0,
  innerbottommargin=.7mm
}
]{line_proof}

\declaretheorem[name=Proof,style=line_proof,numbered=no]{lproof}

\newcommand{\N}{\mathbb{N}}
\newcommand{\Z}{\mathbb{Z}}

\newcommand{\R}{\mathbb{R}}
\newcommand{\C}{\mathbb{C}}

\DeclareMathSymbol{\minu}{\mathbin}{AMSa}{"39}

\newcommand{\interp}[1]{\left\llbracket #1 \right\rrbracket}
\newcommand{\ZXp}{\mathsf{ZX}_p^{\mathrm{Stab}}}
\newcommand{\Stab}{\mathsf{Stab}_p}

\newcommand{\AffCoIsoRel}{\mathsf{AffCoIsoRel}_{\Z_p}}

\newcommand{\ZXeq}{\text{\textsc{\textnormal{zx}}}_p}

\tikzstyle{gn}=[font={\scriptsize\boldmath}, inner sep=1mm, outer sep=-1.8mm, scale=0.8, tikzit shape=circle, draw=black, fill={zx_green}, tikzit draw=black, tikzit fill=green, tikzit category=ZX, shape=circle]
\tikzstyle{rn}=[font={\scriptsize\boldmath}, inner sep=1mm, outer sep=-1.8mm, scale=0.8, tikzit shape=circle, draw=black, fill={zx_red}, tikzit fill=red, tikzit draw=black, shape=circle, tikzit category=ZX]
\tikzstyle{had}=[fill=yellow, draw=black, shape=rectangle, tikzit category=ZX, tikzit fill=yellow, tikzit draw=black, inner sep=2.5pt, font={\scriptsize\boldmath}]
\tikzstyle{gphase}=[rounded rectangle, rounded rectangle arc length=120, fill={zx_green}, inner sep=2pt, font={\tiny\boldmath}, label distance=1mm, fill opacity=.8, text opacity=1, tikzit category=ZX]
\tikzstyle{rphase}=[rounded rectangle, rounded rectangle arc length=120, fill={zx_red}, inner sep=2pt, font={\tiny\boldmath}, label distance=1mm, fill opacity=.6, text opacity=1, tikzit category=ZX]
\tikzstyle{mphase}=[rounded rectangle, rounded rectangle arc length=120, fill=gray, inner sep=2pt, font={\tiny\boldmath}, label distance=1mm, fill opacity=.6, text opacity=1, tikzit category=ZX]
\tikzstyle{rmat}=[draw, signal, fill={zx_grey}, signal to=east, signal from=west, inner sep=1pt, minimum height=6pt, font={\scriptsize\boldmath}, tikzit category=GLA]
\tikzstyle{lmat}=[draw, signal, fill={zx_grey}, signal to=west, signal from=east, inner sep=1pt, minimum height=6pt, font={\scriptsize\boldmath}, tikzit category=GLA]
\tikzstyle{umat}=[draw, signal, fill={zx_grey}, signal to=north, signal from=south, inner sep=1pt, minimum width=6pt, font={\scriptsize\boldmath}, tikzit category=GLA]
\tikzstyle{dmat}=[draw, signal, fill={zx_grey}, signal to=south, signal from=north, inner sep=1pt, minimum width=6pt, font={\scriptsize\boldmath}, tikzit category=GLA]
\tikzstyle{graph_vertex}=[fill=black, draw=black, shape=circle, tikzit category=mbqc, minimum size=2.4mm, inner sep=.8mm]
\tikzstyle{graph_weight}=[fill=white, draw=none, shape=rectangle, tikzit category=mbqc, inner sep=2pt, scale=.8]
\tikzstyle{graph_state}=[fill=yellow, draw=none, shape=rectangle, tikzit category=ZX, rounded corners=1.3mm, minimum height=1.7cm, minimum width=1.3cm, opacity=.7, text opacity=1]
\tikzstyle{box}=[fill=white, draw=black, shape=rectangle, inner sep=2.5pt, font={\scriptsize\boldmath}]
\tikzstyle{new style 0}=[fill={rgb,255: red,191; green,191; blue,191}, draw={rgb,255: red,191; green,191; blue,191}, shape=circle, inner sep=0.5mm]

\tikzstyle{dash_edge}=[-, dashed]
\tikzstyle{new edge style 0}=[-, draw=red, fill=red]
\tikzstyle{new edge style 1}=[-, draw={rgb,255: red,191; green,191; blue,191}]

\definecolor{zx_grey}{RGB}{211,211,211}
\definecolor{zx_red}{RGB}{232,165,165}
\definecolor{zx_green}{RGB}{216,248,216}
\definecolor{dark-gray}{gray}{0.40}

\usetikzlibrary{circuits.ee.IEC}
\usetikzlibrary{shapes.gates.logic.US,shapes.gates.logic.IEC}

\newcommand{\ground}{%
	\begin{tikzpicture}[circuit ee IEC,yscale=1.0,xscale=1.0]
		\draw[solid,arrows=-] (0,1ex) to (0,0) node[anchor=center,ground,rotate=-90,xshift=.66ex] {};
	\end{tikzpicture}
}
\newcommand{\bvdots}{ \tikz[baseline, every node/.style={inner sep=0}]{ \node at (0,0){.}; \node at (0,-6pt){.}; \node at (0,6pt){.}; } }

\title{Complete ZX-calculi for the stabiliser fragment in odd prime
  dimensions}

\author{Robert I. Booth}
\affiliation{
  University of Edinburgh, United Kingdom
}
\affiliation{
  Sorbonne Universit\'e, CNRS, LIP6,
  4 place Jussieu, \mbox{F-75005} Paris, France
}
\affiliation{
  LORIA CNRS, Inria Mocqua, Universit\'e de Lorraine, \mbox{F-54000} Nancy,
  France
}
\author{Titouan Carette}
\affiliation{
  Université Paris-Saclay, Inria, CNRS, LMF, 91190, Gif-sur-Yvette, France
}

\allowdisplaybreaks

\begin{document}

\maketitle

\begin{abstract}
  We introduce a family of ZX-calculi which axiomatise the stabiliser fragment of
  quantum theory in odd prime dimensions. These calculi recover many of the nice
  features of the qubit ZX-calculus which were lost in previous proposals for
  higher-dimensional systems. We then prove that these calculi are complete, i.e.
  provide a set of rewrite rules which can be used to prove any equality of
  stabiliser quantum operations. Adding a discard construction, we obtain a
  calculus complete for mixed state stabiliser quantum mechanics in odd prime
  dimensions, and this furthermore gives a complete axiomatisation for the
  related diagrammatic language for affine co-isotropic relations. 
\end{abstract}

The ZX-calculus is a powerful yet intuitive graphical language for reasoning
about quantum computing, or, more generally, about operations on quantum
systems \cite{coecke_picturing_2017, van_de_wetering_zx-calculus_2020}. It
allows one to represent such quantum operations pictorially, and comes equipped
with a set of rules which, in principle, make it possible to derive any equality
between those pictures \cite{backens_zx-calculus_2014, jeandel_complete_2017,
  ng_universal_2017}. It has now had several applications in quantum information
processing, from MBQC \cite{hutchison_rewriting_2010, backens_there_2021},
through quantum error correction codes \cite{duncan_verifying_2014,
  de_beaudrap_zx_2017, chancellor_graphical_2018, garvie_verifying_2018}. More
recently, it has been used to obtain state-of-the-art optimisation techniques
for quantum circuits \cite{kissinger_reducing_2020, de_beaudrap_fast_2020,
  duncan_graph-theoretic_2020} and faster classical simulation algorithms for
general quantum computations \cite{kissinger_classical_2022}.

Despite its origins in categorical quantum mechanics and the diagrammatic
language for finite-dimensional linear spaces \cite{abramsky_categorical_2008,
  coecke_interacting_2011, coecke_picturing_2017} the literature on the
ZX-calculus has been concerned almost exclusively with small-dimensional quantum
systems, and even then mostly with the case of two-dimensional quantum systems,
or qubits. The qubit ZX-calculus is remarkable in its simple treatment of
stabiliser quantum mechanics, along with the fact that any diagram can be
treated purely graph-theoretically, without concern to its overall layout, and
without losing its quantum-mechanical interpretation. Those proposals that go
beyond qubits lose many of these nice features, and are significantly more
complicated than the qubit case \cite{ranchin_depicting_2014, wang_qutrit_2014,
  bian_graphical_2015, wang_qufinite_2021, wang_non-anyonic_2021}. In
particular, they eschew the prised ``Only Connectivity Matters'' (OCM)
meta-rule, often cited as one of the key features in the qubit case. In these
calculi, which can represent any linear map between the corresponding Hilbert
spaces, it is also not necessarily obvious (at least, to us) how to pick out and
work with the stabiliser fragment. 

Stabiliser quantum mechanics is a simple yet particularly important fragment of
quantum theory. While much less powerful than the full fragment---it can be
efficiently classically simulated, even in odd prime dimensions
\cite{de_beaudrap_linearized_2012}---it has seen significant study
\cite{gheorghiu_standard_2014, hausmann_consolidating_2021} and forms the basis
for a number of key methods in quantum information theory
\cite{gottesman_fault-tolerant_1999}. Operationally, it can be described as the
fragment of quantum mechanics which is obtained if one allows only state
preparation in the computational basis, and unitary operations from the
\emph{Clifford groups} \cite{gottesman_fault-tolerant_1999}. In the qubit case,
the stabiliser fragment of the ZX-calculus was proved complete in
\cite{backens_zx-calculus_2014} while ignoring global scalars, and extended to
include scalars in \cite{backens_making_2015}. A simplified calculus based on
these results but further reducing the set of axioms of the calculus was
presented in \cite{backens_simplified_2017}.

In this article, we present a simple ZX-calculus for stabiliser quantum
mechanics in odd prime dimensions, and which recovers as many of the nice
features of the qubit calculus as possible. In odd prime dimensions, stabiliser
quantum mechanics can be given a particularly nice graphical presentation, owing
to the group-theory underlying the corresponding Clifford groups
\cite{nenhauser_explicit_2002, appleby_properties_2009,
  de_beaudrap_linearized_2012}. We then give this calculus a set of rewrite
rules that is complete, i.e. rich enough to derive any equality of stabiliser
quantum operations. In particular, it is a design priority to recover OCM, and
to make explicit the stabiliser fragment and its group-theoretical
underpinnings. Adding a discard construction
\cite{DBLP:conf/icalp/CaretteJPV19}, we obtain a universal and complete calculus
for mixed state stabiliser quantum mechanics in odd prime dimensions. By
previous work \cite{comfort_graphical_2021}, this gives a complete
axiomatisation for the related diagrammatic language for affine co-isotropic
relations, while still maintaining OCM.

Although we do not do so here, these calculi can naturally be extended to
represent much larger fragments of quantum theory, up to the entire theory in
odd prime dimensions \cite{wang_qutrit_2014}. However, finding a complete
axiomatisation for such calculi will presumably be a much more complicated task,
and we leave this for future work.

\paragraph{Acknowledgements}
We would like to thank Cole Comfort and Simon Perdrix for enlightening
discussions. We also thank John van de Wetering, Boldizsar Poor, and Lia Yeh as
well as our reviewers for their numerous suggestions for improvement. RIB was
funded by the ANR VanQuTe project (ANR-17-CE24-0035), as well as the Cisco
University Research Program Fund.

\section{Stabiliser quantum mechanics in odd prime dimensions}
\label{sec:preliminaries}
Throughout this paper, \(p\) denotes an arbitrary odd prime, and \(\Z_p =
\Z/p\Z\) the ring of integers with arithmetic modulo \(p\). We also put \(\omega
\coloneqq e^{i\frac{2\pi}{p}}\), and let \(\Z_p^*\) be the group of units of
\(\Z_p\). Since \(p\) is prime, \(\Z_p\) is a field and \(\Z_p^* = \Z_p
\setminus \{0\}\) as a set. We also have need of the following definition:
\begin{equation}
  \label{eq:legendre_characteristic}
  \chi_p(x) =
  \begin{cases}
    1 \qif \text{there is no } y \in \Z_p \text{ s.t. } x = y^2; \\
    0 \quad \text{otherwise};
  \end{cases}
\end{equation}
which is just the characteristic function of the complement of the set of
squares in \(\Z_p\).

The Hilbert space of a qupit \cite{gottesman_fault-tolerant_1999,
  wang_qudits_2020} is \(\mathcal{H} = \operatorname{span}\{\ket{m} \mid m \in
\Z_p\} \cong \C^p\), and we write \(U(\mathcal{H})\) the group of unitary
operators acting on \(\mathcal{H}\). We have the following standard operators on
\(\mathcal{H}\), also known as the clock and shift operators:
\begin{equation}
  Z \ket{m} \coloneqq \omega^{m} \ket{m} \qand
  X \ket{m} \coloneqq \ket{m+1}
  \qq{for any} m \in \Z_p.
  \label{eq:pauli} 
\end{equation}
In particular, note that \(ZX = \omega XZ\).
We call any operator of the form \(\omega^k X^a Z^b\) for \(k,a,b \in \Z_p\) a
\emph{Pauli operator}. We say a Pauli operator is \emph{trivial} if it is
proportional to the identity. The collection of all Pauli operators is denoted
\(\mathscr{P}_1\) and called the \emph{Pauli group}. For \(n \in \N^*\), the
\emph{generalised Pauli group} is \(\mathscr{P}_n \coloneqq \bigotimes_{k=1}^n
\mathscr{P}_1\).

Of particular importance to us are the \emph{(generalised) Clifford groups}.
These are defined, for each \(n \in \N^*\), as the normaliser of
\(\mathscr{P}_n\) in \(U(\mathcal{H}^{\otimes n})\): \(C\) is a Clifford
operator if for any \(P \in P_n\), \(CPC^\dagger \in \mathscr{P}_n\). It is clear that
that every Pauli operator is Clifford, but there are non-Pauli Clifford operators.
Some important examples are the \emph{Hadamard} gate:
\begin{equation} 
  H\ket{m} = \frac{1}{\sqrt{d}} \sum_{n \in \Z_p} \omega^{mn} \ket{n}
  \qq{s.t.} HXH^\dagger = Z \qand HZH^\dagger = X^{-1},
  \label{eq:hadamard}
\end{equation}
and the \emph{phase} gate:
\begin{equation} 
  S\ket{m} = \omega^{2^{-1}m(m+1)} \ket{m}
  \qq{s.t.} SXS^\dagger = \omega XZ \qand SZS^\dagger = Z.
  \label{eq:phase}
\end{equation}
Yet another important example is the \emph{controlled-phase} gate, which acts on
\(\mathcal{H} \otimes \mathcal{H}\),
\begin{equation}
  E \ket{m} \ket{n} \coloneqq \omega^{mn} \ket{m}\ket{n}.
\end{equation}

It is important to emphasise a key difference between the qupit and the qubit
case: when \(p \neq 2\), none of these operators are self-inverse.
In fact, if \(Q\) is a Pauli and \(I\) the identity operator on \(\mathcal{H}\),
we have:
\begin{equation}
  Q^p = I, \quad E^p = I \otimes I \qand H^4 = I.
\end{equation}

As a side note, equations \eqref{eq:pauli}, \eqref{eq:hadamard} and
\eqref{eq:phase} imply that both \(X\) and \(Z\), and in fact every Pauli, have
spectrum \(\{\omega^k \mid k \in \Z_p\}\). As a result, we denote \(\ket{k :
  Q}\) the eigenvector of a given Pauli \(Q\) associated with eigenvalue
\(\omega^k\), and furthermore use the notation
\begin{equation}
  \ket{\underbrace{k,\dots, k}_{n\text{ times}} : Q } = \bigotimes_{k=1}^n \ket{k:Q}.
\end{equation}
It follows from equation~\eqref{eq:pauli} that we can identify \(\ket{k:Z} =
\ket{k}\).

Now, for any \(\alpha \in [0,2\pi)\), the operator \(e^{i\alpha} I\) is Clifford.
However, we want to construct calculi with a finite axiomatisation. As a result,
the diagrams in the calculus are countable and this makes it impossible for us
to represent all such phases \(e^{i\alpha}\). Unfortunately, finding a group of
phases that behaves well diagrammatically is somewhat inconvenient, and involves
some elementary number theory. For an odd prime \(p\), we consider the group of
phases given by:
\begin{itemize}
\item if \(p \equiv 1 \mod 4\),
  \begin{equation}
    \label{eq:Pp_1}
    \mathbb{P}_p \coloneqq \left\{ (-1)^s \omega^t \mid s,t \in \Z \right\};
  \end{equation}
\item if \(p \equiv 3 \mod 4\),
  \begin{equation}
    \label{eq:Pp_3}
    \mathbb{P}_p \coloneqq \left\{ i^s \omega^t \mid s,t \in \Z \right\}.
  \end{equation}
\end{itemize}
This of course covers all cases since \(p\) is odd. Then, we restrict our
attention to the \emph{reduced} Clifford group,
\begin{equation}
  \mathscr{C}_n = \left\{ \lambda U \mid \lambda \in \mathbb{P}_p, U \text{ is Clifford and special unitary} \right\}.
\end{equation}
We call \(\mathscr{C}_1^{\otimes n}\) the \emph{local Clifford group} on \(n\)
qupits. It is clear from these examples that \(\mathscr{C}_n\) is strictly
larger than \(\mathscr{C}_1^{\otimes n}\), but it turns out to not be that much
larger:
\begin{proposition}[\cite{clark_valence_2006, nenhauser_explicit_2002}]
  \label{proposition:Cn_generators}
  The reduced Clifford group \(\mathscr{C}_n\) is generated by the gate-set
  \(\{H_j, S_j, E_{j,k} \mid j,k = 1, \dots, n\}\).
\end{proposition}

\emph{Stabiliser quantum mechanics} can be operationally described as the
fragment of quantum mechanics in which the only operations allowed are
initialisations and measurements in the eigenbases of Pauli operators, and
unitary operations from the generalised Clifford groups. As before, we restrict
our attention to the fragment of stabiliser quantum mechanics where only unitary
operations from the reduced Clifford groups are allowed. Scalars are then taken
from the monoid \(\mathbb{G}_p \coloneqq \left\{0, \sqrt{p^r} \lambda \mid r \in
  \Z, \lambda \in \mathbb{P}_p \right\}\). Little is lost for the description of
quantum algorithms, since we can always simplify by a global phase to make the
Clifford generators special unitary. Thus, we can embed any stabiliser circuit
into the calculus, and then calculate the relative phases of different branches
of a computation without restriction.

\begin{definition}
  The symmetric monoidal category \(\mathsf{Stab}_p\) has as objects \(\C^{pn}\)
  for each \(n \in \N\), and morphisms generated by:
  \begin{itemize}
  \item \(\C \to \C^p : \lambda \mapsto \lambda\ket{0}\);
  \item \(\C^{pn} \to \C^{pn} : \ket{\psi} \mapsto U\ket{\psi}\) for any \(U \in
    \mathscr{C}_n\);
  \item \(\C^p \to \C : \ket{\psi} \mapsto \ip{0}{\psi}\).
  \end{itemize}
  The monoidal product is given by the usual tensor product of linear spaces.
\end{definition}

It is clear that \(\mathsf{Stab}_p\) is a subcategory of the category
\(\mathsf{FLin}\) of finite dimensional \(\C\)-linear spaces; it is also a PROP.

\section{A ZX-calculus for odd prime dimensions}
\label{sec:calculus}
In this section, we present our family of ZX-calculi, with one for each odd
prime. Relying on some of the group theoretical properties of the qupit Clifford
groups, we can give a relatively simple presentation of the calculi, which
avoids the need to explicitly consider rotations in \(p\)-dimensional space,
significantly simplifing the presentation compared to previous work
\cite{ranchin_depicting_2014}. These calculi are also constructed in order to
satisfy the property of \emph{flexsymmetry}, proposed in
\cite{carette_when_2021, carette_wielding_2021}, and which allows one to recover
the OCM meta-rule. OCM is an intuitively desirable feature for the design of a
graphical language; anecdotally, it greatly simplifies the human manipulation of
diagrams, including in the proofs of this paper. More formally, it means that
the equational theory can be formalised in terms of double pushout rewriting
over graphs rather than over hypergraphs as is necessary in the more general
theory \cite{bonchi_string_2022, bonchi_string_2021, bonchi_string_2022-1}.

Another key concern is the issue of completeness, which we begin to address in
this article. Outside of qubits \cite{jeandel_complete_2017, ng_universal_2017,
  vilmart_near-optimal_2018}, there has so far been a complete axiomatisation
only for the stabiliser fragment in dimension \(p=3\) \cite{wang_qutrit_2018}.
We present an axiomatisation which is complete for the stabiliser fragment for
any odd prime \(p\), and leave the general case for future work.

\subsection{Generators}

For any odd prime \(p\), consider the symmetric monoidal category \(\ZXp\) with
objects \(\N\) and morphisms (or, diagrams) generated by:
\begin{align*}
  \begin{tikzpicture}
	\begin{pgfonlayer}{nodelayer}
		\node [style=none] (0) at (-1, 0) {};
		\node [style=none] (1) at (1, 0) {};
	\end{pgfonlayer}
	\begin{pgfonlayer}{edgelayer}
		\draw (0.center) to (1.center);
	\end{pgfonlayer}
\end{tikzpicture}
 &: 1 \to 1 & \quad\quad\quad
  \begin{tikzpicture}
	\begin{pgfonlayer}{nodelayer}
		\node [style=gn] (0) at (0, 0) {};
		\node [style=none] (1) at (-1, 0) {};
		\node [style=none] (2) at (1, 0) {};
	\end{pgfonlayer}
	\begin{pgfonlayer}{edgelayer}
		\draw (0) to (1.center);
		\draw (0) to (2.center);
	\end{pgfonlayer}
\end{tikzpicture}
 &: 1 \to 1 & \quad\quad\quad
  \begin{tikzpicture}
	\begin{pgfonlayer}{nodelayer}
		\node [style=rn] (0) at (0, 0) {};
		\node [style=none] (1) at (-1, 0) {};
		\node [style=none] (2) at (1, 0) {};
	\end{pgfonlayer}
	\begin{pgfonlayer}{edgelayer}
		\draw (0) to (1.center);
		\draw (0) to (2.center);
	\end{pgfonlayer}
\end{tikzpicture}
 &: 1 \to 1 & \\
  \begin{tikzpicture}
	\begin{pgfonlayer}{nodelayer}
		\node [style=had] (0) at (0, 0) {};
		\node [style=none] (1) at (-1, 0) {};
		\node [style=none] (2) at (1, 0) {};
	\end{pgfonlayer}
	\begin{pgfonlayer}{edgelayer}
		\draw (0) to (1.center);
		\draw (0) to (2.center);
	\end{pgfonlayer}
\end{tikzpicture}
 &: 1 \to 1 & \quad
  \begin{tikzpicture}
	\begin{pgfonlayer}{nodelayer}
		\node [style=gn, label={[gphase]below:$x,y$}] (0) at (-0.5, 0) {};
		\node [style=none] (1) at (0.5, 0) {};
	\end{pgfonlayer}
	\begin{pgfonlayer}{edgelayer}
		\draw (0) to (1.center);
	\end{pgfonlayer}
\end{tikzpicture}
 &: 0 \to 1 & \quad
  \begin{tikzpicture}
	\begin{pgfonlayer}{nodelayer}
		\node [style=rn, label={[rphase]below:$x,y$}] (0) at (-0.5, 0) {};
		\node [style=none] (1) at (0.5, 0) {};
	\end{pgfonlayer}
	\begin{pgfonlayer}{edgelayer}
		\draw (0) to (1.center);
	\end{pgfonlayer}
\end{tikzpicture}
 &: 0 \to 1 & \\
  \input{figures/generators/cup.tikz} &: 0 \to 2 & \quad
  \begin{tikzpicture}
	\begin{pgfonlayer}{nodelayer}
		\node [style=gn, label={[gphase]below:$x,y$}] (0) at (0, 0) {};
		\node [style=none] (1) at (-1, 0) {};
		\node [style=none] (2) at (1, 0) {};
	\end{pgfonlayer}
	\begin{pgfonlayer}{edgelayer}
		\draw (0) to (1.center);
	\end{pgfonlayer}
\end{tikzpicture}
 &: 1 \to 0 & \quad
  \begin{tikzpicture}
	\begin{pgfonlayer}{nodelayer}
		\node [style=rn, label={[rphase]below:$x,y$}] (0) at (0, 0) {};
		\node [style=none] (1) at (-1, 0) {};
		\node [style=none] (2) at (1, 0) {};
	\end{pgfonlayer}
	\begin{pgfonlayer}{edgelayer}
		\draw (0) to (1.center);
	\end{pgfonlayer}
\end{tikzpicture}
 &: 1 \to 0 & \\
  \input{figures/generators/cap.tikz} &: 2 \to 0 & \quad
  \input{figures/generators/g-mult.tikz} &: 2 \to 1 & \quad
  \input{figures/generators/r-mult.tikz} &: 2 \to 1 & \\
  \input{figures/generators/braid.tikz} &: 2 \to 2 & \quad
  \input{figures/generators/g-comult.tikz} &: 1 \to 2 & \quad
  \input{figures/generators/r-comult.tikz} &: 1 \to 2 &
\end{align*}
where \(x,y \in \Z_p\). We also introduce a generator \(\star : 0 \to 0\) to
simplify the calculus; it will correspond to a scalar whose representation in
terms of the other generators depends non-trivially on the dimension \(p\).
Finally, the empty diagram \(\input{figures/generators/empty_small.tikz} : 0 \to
0\) is also a morphism in the language. Morphisms are composed by connecting
output wires to input wires, and the monoidal product is given on objects by \(n
\otimes m = n+m\) and on morphisms by vertical juxtaposition of diagrams.

We extend this elementary notation with a first piece of syntactic sugar, which
is standard for the ZX-calculus family: \emph{green spiders} are defined
inductively, for any \(m,n \in \N\), by
\begin{equation}
  \input{figures/generators/g-spider-def.tikz}
\end{equation}
It is clear that these diagrams have types \(m+1 \to 1\), \(1 \to n+1\) and \(m
\to n\) respectively. This construction can be justified by remarking that, under
the equational theory to be presented in section~\ref{ssec:axiomatisation}, the
unlabelled green fragment of the language forms a special commutative Frobenius
algebra, which therefore admits a canonical form in terms of such ``spiders''
\cite{coecke_picturing_2017}.

\emph{Red spiders} are defined analogously, with a small twist:
\begin{equation}
  \input{figures/generators/r-spider-def.tikz}
\end{equation}
The addition of these \(1 \to 1\) red vertices (called \emph{antipodes}) when
compared to green spiders mimicks a construction of the ZH-calculus
\cite{backens_zh_2019, backens_completeness_2021}, where they are sometimes
dubbed ``harvestmen''. They are crucial in order for the red spiders to be
flexsymmetric in our equational theory and thus to recover the OCM meta-rule
mentioned above.

We can then define labelled spiders by passing the label of a \(0 \to 1\)
diagram along an input wire:
\begin{equation}
  \input{figures/generators/spiders_labelled.tikz}
\end{equation}

In order to avoid clutter, we also use the following shorthand:
\begin{equation}
  \input{figures/equations/hadamard_inverse.tikz}
\end{equation}
which we shall see represents the compositional inverse of
\(\begin{tikzpicture}
	\begin{pgfonlayer}{nodelayer}
		\node [style=had] (0) at (0, 0) {};
		\node [style=none] (1) at (-1, 0) {};
		\node [style=none] (2) at (1, 0) {};
	\end{pgfonlayer}
	\begin{pgfonlayer}{edgelayer}
		\draw (0) to (1.center);
		\draw (0) to (2.center);
	\end{pgfonlayer}
\end{tikzpicture}
\).

\subsection{Standard interpretation and universality}

The standard interpretation of a \(\ZXp\)-diagram is a symmetric
monoidal functor \(\interp{-} : \ZXp \to \mathsf{FLin}\) (the category
of finite-dimensional \(\C\)-linear spaces). It is defined on objects as
\(\interp{m} \coloneqq \C^{p \times m}\), and on the generators
of the morphisms as:
\begin{align*}
    \interp{\begin{tikzpicture}
	\begin{pgfonlayer}{nodelayer}
		\node [style=gn, label={[gphase]below:$x,y$}] (0) at (-0.5, 0) {};
		\node [style=none] (1) at (0.5, 0) {};
	\end{pgfonlayer}
	\begin{pgfonlayer}{edgelayer}
		\draw (0) to (1.center);
	\end{pgfonlayer}
\end{tikzpicture}
} & = \sum_{k\in\Z_p} \omega^{2^{-1} (xk + yk^2)} \ket{k:Z} & \quad
    \interp{\begin{tikzpicture}
	\begin{pgfonlayer}{nodelayer}
		\node [style=gn, label={[gphase]below:$x,y$}] (0) at (0, 0) {};
		\node [style=none] (1) at (-1, 0) {};
		\node [style=none] (2) at (1, 0) {};
	\end{pgfonlayer}
	\begin{pgfonlayer}{edgelayer}
		\draw (0) to (1.center);
	\end{pgfonlayer}
\end{tikzpicture}
} & = \sum_{k\in\Z_p} \omega^{2^{-1} (xk + yk^2)} \bra{k:Z}  \\
    \interp{\input{figures/generators/g-comult.tikz}} & = \sum_{k\in\Z_p} \dyad{k:Z}{k,k:Z}&
    \interp{\input{figures/generators/g-mult.tikz}} &= \sum_{k\in\Z_p} \dyad{k,k:Z}{k:Z} & \\
    \interp{\begin{tikzpicture}
	\begin{pgfonlayer}{nodelayer}
		\node [style=gn] (0) at (0, 0) {};
		\node [style=none] (1) at (-1, 0) {};
		\node [style=none] (2) at (1, 0) {};
	\end{pgfonlayer}
	\begin{pgfonlayer}{edgelayer}
		\draw (0) to (1.center);
		\draw (0) to (2.center);
	\end{pgfonlayer}
\end{tikzpicture}
} &= \sum_{k\in\Z_p} \dyad{k:Z}{k:Z} & \quad
    \interp{\begin{tikzpicture}
	\begin{pgfonlayer}{nodelayer}
		\node [style=rn] (0) at (0, 0) {};
		\node [style=none] (1) at (-1, 0) {};
		\node [style=none] (2) at (1, 0) {};
	\end{pgfonlayer}
	\begin{pgfonlayer}{edgelayer}
		\draw (0) to (1.center);
		\draw (0) to (2.center);
	\end{pgfonlayer}
\end{tikzpicture}
} &= \sum_{k\in\Z_p} \dyad{-k:X}{k:X} & \quad \\
    \interp{\begin{tikzpicture}
	\begin{pgfonlayer}{nodelayer}
		\node [style=rn, label={[rphase]below:$x,y$}] (0) at (-0.5, 0) {};
		\node [style=none] (1) at (0.5, 0) {};
	\end{pgfonlayer}
	\begin{pgfonlayer}{edgelayer}
		\draw (0) to (1.center);
	\end{pgfonlayer}
\end{tikzpicture}
} &= \sum_{k\in\Z_p} \omega^{2^{-1} (xk + yk^2)} \ket{-k:X} & \quad
    \interp{\begin{tikzpicture}
	\begin{pgfonlayer}{nodelayer}
		\node [style=rn, label={[rphase]below:$x,y$}] (0) at (0, 0) {};
		\node [style=none] (1) at (-1, 0) {};
		\node [style=none] (2) at (1, 0) {};
	\end{pgfonlayer}
	\begin{pgfonlayer}{edgelayer}
		\draw (0) to (1.center);
	\end{pgfonlayer}
\end{tikzpicture}
} &= \sum_{k\in\Z_p} \omega^{2^{-1} (xk + yk^2)}\bra{k:X} & \\
    \interp{\input{figures/generators/r-comult.tikz}} &= \sum_{k\in\Z_p} \dyad{-k,-k:X}{k:X} &
    \interp{\input{figures/generators/r-mult.tikz}} &= \sum_{k\in\Z_p} \dyad{-k:X}{k,k:X} & \\
    \interp{\begin{tikzpicture}
	\begin{pgfonlayer}{nodelayer}
		\node [style=none] (0) at (-1, 0) {};
		\node [style=none] (1) at (1, 0) {};
	\end{pgfonlayer}
	\begin{pgfonlayer}{edgelayer}
		\draw (0.center) to (1.center);
	\end{pgfonlayer}
\end{tikzpicture}
} &= \sum_{k\in\Z_p} \dyad{k:Z}{k:Z} & \quad
    \interp{\begin{tikzpicture}
	\begin{pgfonlayer}{nodelayer}
		\node [style=had] (0) at (0, 0) {};
		\node [style=none] (1) at (-1, 0) {};
		\node [style=none] (2) at (1, 0) {};
	\end{pgfonlayer}
	\begin{pgfonlayer}{edgelayer}
		\draw (0) to (1.center);
		\draw (0) to (2.center);
	\end{pgfonlayer}
\end{tikzpicture}
} &= \sum_{k\in\Z_p} \dyad{k:X}{k:Z} & \\
    \interp{\input{figures/generators/cup.tikz}} &= \sum_{k\in\Z_p} \ket{kk:Z} & \quad
    \interp{\input{figures/generators/cap.tikz}} &= \sum_{k\in\Z_p} \bra{kk:Z} & \\
    \interp{\input{figures/generators/braid.tikz}} &= \sum_{k,\ell\in\Z_p} \dyad{k,\ell:Z}{\ell,k:Z} & \quad
    \interp{\scalebox{1.4}{$\star$}} &= -1 &
\end{align*}
By the functoriality of the standard interpretation, we then deduce that
\begin{equation}
  \interp{\input{figures/generators/g-spider.tikz}} = \sum_{k\in\Z_p} \omega^{2^{-1} (xk + yk^2)} \ket{k:Z}^{\otimes n} \bra{k:Z}^{\otimes m},
\end{equation}
and
\begin{equation}
  \interp{\input{figures/generators/r-spider.tikz}} = \sum_{k\in\Z_p} \omega^{2^{-1} (xk + yk^2)} \ket{-k:X}^{\otimes n} \bra{k:X}^{\otimes m}.
\end{equation}

There are a couple of peculiarities in this standard interpretation that must be
remarked upon. Firstly, note that the red \(1 \to 1\) spider is not the
identity, but rather maps the label of any input \(X\)-eigenstate to its
additive inverse. This feature is repeated on every red spider, and is another
of the side-effects of imposing that the OCM rule must be sound for this
interpretation.

Secondly, and much more subtly, we have made an unconventional choice for the
interpretation of the label of spiders as a phase: there is an additional factor
\(2^{-1}\) which one might not expect. This might seem like an inconsequential
choice which ought to be removed, but it has an important consequence for the
equational theory. There is a specific rule in the axiomatisation,
\textsc{(Gauss)}, which is sound for this choice of phases but not for the
simpler one which omits the factor \(2^{-1}\). 

\begin{theorem}[Universality] \label{thm:universality}
  The standard interpretation \(\interp{-}\) is universal for the qupit
  stabiliser fragment, i.e. for any stabiliser operation \(C :
  \C^{p m} \to \C^{p n}\) there is a diagram \(D \in \ZXp\) such that
  \(\interp{D} = C\). Put formally, the co-restriction of \(\interp{-}\) to
  \(\mathsf{Stab}_p\) is full.
\end{theorem}

\subsection{Axiomatisation}
\label{ssec:axiomatisation}

\begin{figure}
  \centering
  \input{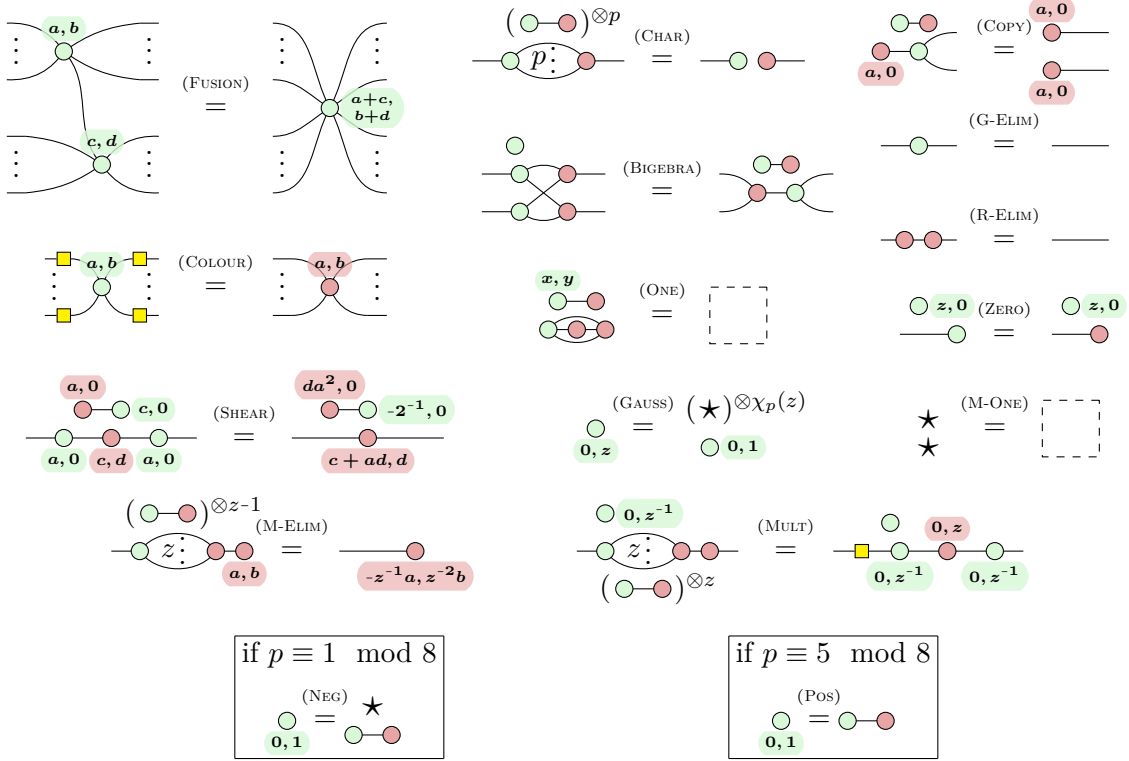}
  \caption{A presentation of the equational theory \(\ZXeq\), which is
  sound and complete for the stabiliser fragment. The equations hold for any
  \(a,b,c,d \in \Z_p\) and \(z \in \Z_p^*\). \(\chi_p\) is the characteristic
  function of the complement of the set of squares in \(\Z_p\), defined in
  equation~\eqref{eq:legendre_characteristic}}
  \label{fig:axioms}
\end{figure}

We now begin to introduce rewrite rules with which to perform purely
diagrammatic reasoning. By doing so we are in fact describing a PROP by
generators and relations \cite{baez_props_2018}, thus the swap is required to
satisfy the following properties:
\begin{equation}
 \begin{tikzpicture}
	\begin{pgfonlayer}{nodelayer}
		\node [style=none] (32) at (15, 0) {$=$};
		\node [style=none] (58) at (8, 1.5) {};
		\node [style=none] (59) at (10, 1.5) {};
		\node [style=none] (60) at (8, -0.5) {};
		\node [style=none] (61) at (10, -0.5) {};
		\node [style=none] (62) at (9, 0.5) {$D$};
		\node [style=none] (67) at (10, 1) {};
		\node [style=none] (68) at (11, 1) {};
		\node [style=none] (69) at (11, 0) {};
		\node [style=none] (70) at (10, 0) {};
		\node [style=none] (86) at (13, -1) {};
		\node [style=none] (87) at (14, -1) {};
		\node [style=none] (88) at (14, 0) {};
		\node [style=none] (89) at (13, 0) {};
		\node [style=none] (90) at (13, 1) {};
		\node [style=none] (92) at (14, 1) {};
		\node [style=none] (98) at (10.5, 0.5) {$\bvdots$};
		\node [style=none] (100) at (13.5, -0.5) {$\bvdots$};
		\node [style=none] (107) at (7, 1) {};
		\node [style=none] (108) at (8, 1) {};
		\node [style=none] (109) at (8, 0) {};
		\node [style=none] (110) at (7, 0) {};
		\node [style=none] (111) at (7.5, 0.5) {$\bvdots$};
		\node [style=none] (112) at (7, -1) {};
		\node [style=none] (113) at (8, -1) {};
		\node [style=none] (117) at (10, -1) {};
		\node [style=none] (118) at (11, -1) {};
		\node [style=none] (119) at (22, -1.5) {};
		\node [style=none] (120) at (20, -1.5) {};
		\node [style=none] (121) at (22, 0.5) {};
		\node [style=none] (122) at (20, 0.5) {};
		\node [style=none] (123) at (21, -0.5) {$D$};
		\node [style=none] (124) at (20, -1) {};
		\node [style=none] (125) at (19, -1) {};
		\node [style=none] (126) at (19, 0) {};
		\node [style=none] (127) at (20, 0) {};
		\node [style=none] (128) at (17, 1) {};
		\node [style=none] (129) at (16, 1) {};
		\node [style=none] (130) at (16, 0) {};
		\node [style=none] (131) at (17, 0) {};
		\node [style=none] (132) at (17, -1) {};
		\node [style=none] (133) at (16, -1) {};
		\node [style=none] (134) at (19.5, -0.5) {$\bvdots$};
		\node [style=none] (135) at (16.5, 0.5) {$\bvdots$};
		\node [style=none] (136) at (23, -1) {};
		\node [style=none] (137) at (22, -1) {};
		\node [style=none] (138) at (22, 0) {};
		\node [style=none] (139) at (23, 0) {};
		\node [style=none] (140) at (22.5, -0.5) {$\bvdots$};
		\node [style=none] (141) at (23, 1) {};
		\node [style=none] (142) at (22, 1) {};
		\node [style=none] (143) at (20, 1) {};
		\node [style=none] (144) at (19, 1) {};
		\node [style=none] (146) at (1, 0.5) {};
		\node [style=none] (147) at (1, -0.5) {};
		\node [style=none] (148) at (2, 0.5) {};
		\node [style=none] (149) at (2, -0.5) {};
		\node [style=none] (150) at (3, 0) {$=$};
		\node [style=none] (151) at (0, 0.5) {};
		\node [style=none] (152) at (0, -0.5) {};
		\node [style=none] (153) at (1, 0.5) {};
		\node [style=none] (154) at (1, -0.5) {};
		\node [style=none] (155) at (4, -0.5) {};
		\node [style=none] (156) at (4, 0.5) {};
		\node [style=none] (157) at (6, 0.5) {};
		\node [style=none] (158) at (6, -0.5) {};
	\end{pgfonlayer}
	\begin{pgfonlayer}{edgelayer}
		\draw (58.center) to (59.center);
		\draw (59.center) to (61.center);
		\draw (61.center) to (60.center);
		\draw (60.center) to (58.center);
		\draw (67.center) to (68.center);
		\draw (69.center) to (70.center);
		\draw (86.center) to (87.center);
		\draw (88.center) to (89.center);
		\draw (90.center) to (92.center);
		\draw (107.center) to (108.center);
		\draw (109.center) to (110.center);
		\draw (112.center) to (113.center);
		\draw (117.center) to (118.center);
		\draw [in=-180, out=0, looseness=1.25] (69.center) to (86.center);
		\draw [in=0, out=180, looseness=1.25] (89.center) to (68.center);
		\draw [in=0, out=-180] (90.center) to (118.center);
		\draw (117.center) to (113.center);
		\draw (119.center) to (120.center);
		\draw (120.center) to (122.center);
		\draw (122.center) to (121.center);
		\draw (121.center) to (119.center);
		\draw (124.center) to (125.center);
		\draw (126.center) to (127.center);
		\draw (128.center) to (129.center);
		\draw (130.center) to (131.center);
		\draw (132.center) to (133.center);
		\draw (136.center) to (137.center);
		\draw (138.center) to (139.center);
		\draw (141.center) to (142.center);
		\draw (143.center) to (144.center);
		\draw [in=0, out=-180, looseness=1.25] (126.center) to (128.center);
		\draw [in=-180, out=0, looseness=1.25] (131.center) to (125.center);
		\draw [in=-180, out=0] (132.center) to (144.center);
		\draw (143.center) to (142.center);
		\draw [in=0, out=-180] (148.center) to (147.center);
		\draw [in=180, out=0] (146.center) to (149.center);
		\draw [in=0, out=-180] (153.center) to (152.center);
		\draw [in=180, out=0] (151.center) to (154.center);
		\draw [in=0, out=-180] (157.center) to (156.center);
		\draw [in=180, out=0] (155.center) to (158.center);
	\end{pgfonlayer}
\end{tikzpicture}

\end{equation}
Note that the last equation is required to hold for any diagram $D:n\to m$. This
property states that our diagrams form a symmetric monoidal category.
Furthermore, we want this category to be self-dual compact-closed, hence the cup
and cap must satisfy:
\begin{equation}
\begin{tikzpicture}
	\begin{pgfonlayer}{nodelayer}
		\node [style=none] (0) at (0, 0) {};
		\node [style=none] (1) at (0.5, 0.5) {};
		\node [style=none] (2) at (0.5, -0.5) {};
		\node [style=none] (3) at (1.5, 0.5) {};
		\node [style=none] (4) at (1.5, -0.5) {};
		\node [style=none] (10) at (6.5, 1) {};
		\node [style=none] (11) at (7, 0.5) {};
		\node [style=none] (12) at (6.5, 0) {};
		\node [style=none] (13) at (6, -0.5) {};
		\node [style=none] (14) at (6.5, -1) {};
		\node [style=none] (20) at (3.5, 0) {};
		\node [style=none] (21) at (4, 0.5) {};
		\node [style=none] (22) at (4, -0.5) {};
		\node [style=none] (23) at (2.5, 0) {$=$};
		\node [style=none] (24) at (20, 0) {};
		\node [style=none] (25) at (19.5, 0.5) {};
		\node [style=none] (26) at (19.5, -0.5) {};
		\node [style=none] (27) at (18.5, 0.5) {};
		\node [style=none] (28) at (18.5, -0.5) {};
		\node [style=none] (29) at (16.5, 0) {};
		\node [style=none] (30) at (16, 0.5) {};
		\node [style=none] (31) at (16, -0.5) {};
		\node [style=none] (32) at (17.5, 0) {$=$};
		\node [style=none] (33) at (13.5, 1) {};
		\node [style=none] (34) at (13, 0.5) {};
		\node [style=none] (35) at (13.5, 0) {};
		\node [style=none] (36) at (14, -0.5) {};
		\node [style=none] (37) at (13.5, -1) {};
		\node [style=none] (38) at (8, 0) {$=$};
		\node [style=none] (39) at (12, 0) {$=$};
		\node [style=none] (40) at (9, 0) {};
		\node [style=none] (41) at (11, 0) {};
	\end{pgfonlayer}
	\begin{pgfonlayer}{edgelayer}
		\draw [in=0, out=-180] (3.center) to (2.center);
		\draw [bend left=45] (2.center) to (0.center);
		\draw [bend left=45] (0.center) to (1.center);
		\draw [in=180, out=0] (1.center) to (4.center);
		\draw [bend left=45] (10.center) to (11.center);
		\draw [bend left=45] (11.center) to (12.center);
		\draw [bend right=45] (12.center) to (13.center);
		\draw [bend right=45] (13.center) to (14.center);
		\draw [bend left=45] (22.center) to (20.center);
		\draw [bend left=45] (20.center) to (21.center);
		\draw [in=180, out=0] (27.center) to (26.center);
		\draw [bend right=45] (26.center) to (24.center);
		\draw [bend right=45] (24.center) to (25.center);
		\draw [in=0, out=180] (25.center) to (28.center);
		\draw [bend right=45] (31.center) to (29.center);
		\draw [bend right=45] (29.center) to (30.center);
		\draw [bend right=45] (33.center) to (34.center);
		\draw [bend right=45] (34.center) to (35.center);
		\draw [bend left=45] (35.center) to (36.center);
		\draw [bend left=45] (36.center) to (37.center);
		\draw (40.center) to (41.center);
	\end{pgfonlayer}
\end{tikzpicture}

\end{equation}

Furthermore, as long as the connectivity of the diagram remains the same, vertices can be
freely moved around without changing the standard interpretation of the diagram.
This is a consequence of the fact that we require our generators to be
\emph{flexsymmetric}, as shown in \cite{carette_when_2021,
  carette_wielding_2021}. This amounts to imposing that all generators except
the swap satisfy:
\begin{equation}
    \begin{tikzpicture}
	\begin{pgfonlayer}{nodelayer}
		\node [style=none] (0) at (0, 0.5) {};
		\node [style=none] (1) at (1, 1.5) {};
		\node [style=none] (2) at (1, -0.5) {};
		\node [style=none] (23) at (8, 0) {$=$};
		\node [style=none] (42) at (4, 2) {};
		\node [style=none] (43) at (6, 2) {};
		\node [style=none] (44) at (4, -2) {};
		\node [style=none] (45) at (6, -2) {};
		\node [style=none] (46) at (5, 0) {$\sigma$};
		\node [style=none] (47) at (0, -0.5) {};
		\node [style=none] (48) at (1, 0.5) {};
		\node [style=none] (49) at (1, -1.5) {};
		\node [style=none] (50) at (1, 2) {};
		\node [style=none] (51) at (3, 2) {};
		\node [style=none] (52) at (1, 0) {};
		\node [style=none] (53) at (3, 0) {};
		\node [style=none] (54) at (2, 1) {$g$};
		\node [style=none] (55) at (6, 1.5) {};
		\node [style=none] (56) at (6, 0.5) {};
		\node [style=none] (57) at (7, 1.5) {};
		\node [style=none] (58) at (7, 0.5) {};
		\node [style=none] (59) at (3, 1.5) {};
		\node [style=none] (60) at (4, 1.5) {};
		\node [style=none] (61) at (4, 0.5) {};
		\node [style=none] (62) at (3, 0.5) {};
		\node [style=none] (63) at (1, -0.5) {};
		\node [style=none] (64) at (1, -1.5) {};
		\node [style=none] (65) at (4, -0.5) {};
		\node [style=none] (66) at (4, -1.5) {};
		\node [style=none] (67) at (9, 0.5) {};
		\node [style=none] (68) at (10, 1.5) {};
		\node [style=none] (69) at (10, -0.5) {};
		\node [style=none] (75) at (9, -0.5) {};
		\node [style=none] (76) at (10, 0.5) {};
		\node [style=none] (77) at (10, -1.5) {};
		\node [style=none] (78) at (10, 2) {};
		\node [style=none] (79) at (12, 2) {};
		\node [style=none] (80) at (10, 0) {};
		\node [style=none] (81) at (12, 0) {};
		\node [style=none] (82) at (11, 1) {$g$};
		\node [style=none] (87) at (12, 1.5) {};
		\node [style=none] (88) at (13, 1.5) {};
		\node [style=none] (89) at (13, 0.5) {};
		\node [style=none] (90) at (12, 0.5) {};
		\node [style=none] (91) at (10, -0.5) {};
		\node [style=none] (92) at (10, -1.5) {};
		\node [style=none] (93) at (13, -0.5) {};
		\node [style=none] (94) at (13, -1.5) {};
		\node [style=none] (95) at (6, -0.5) {};
		\node [style=none] (96) at (6, -1.5) {};
		\node [style=none] (97) at (7, -0.5) {};
		\node [style=none] (98) at (7, -1.5) {};
		\node [style=none] (99) at (3.5, 1) {$\bvdots$};
		\node [style=none] (100) at (6.5, 1) {$\bvdots$};
		\node [style=none] (101) at (12.5, 1) {$\bvdots$};
		\node [style=none] (102) at (12.5, -1) {$\bvdots$};
		\node [style=none] (103) at (9.75, 1) {$\bvdots$};
		\node [style=none] (104) at (9.75, -1) {$\bvdots$};
		\node [style=none] (105) at (0.75, 1) {$\bvdots$};
		\node [style=none] (106) at (0.75, -1) {$\bvdots$};
		\node [style=none] (107) at (6.5, -1) {$\bvdots$};
		\node [style=none] (108) at (3.5, -1) {$\bvdots$};
	\end{pgfonlayer}
	\begin{pgfonlayer}{edgelayer}
		\draw [bend left=45] (2.center) to (0.center);
		\draw [bend left=45] (0.center) to (1.center);
		\draw (42.center) to (43.center);
		\draw (43.center) to (45.center);
		\draw (45.center) to (44.center);
		\draw (44.center) to (42.center);
		\draw [bend left=45] (49.center) to (47.center);
		\draw [bend left=45] (47.center) to (48.center);
		\draw (50.center) to (51.center);
		\draw (51.center) to (53.center);
		\draw (53.center) to (52.center);
		\draw (52.center) to (50.center);
		\draw (59.center) to (60.center);
		\draw (61.center) to (62.center);
		\draw (63.center) to (65.center);
		\draw (66.center) to (64.center);
		\draw (55.center) to (57.center);
		\draw (58.center) to (56.center);
		\draw [bend left=45] (69.center) to (67.center);
		\draw [bend left=45] (67.center) to (68.center);
		\draw [bend left=45] (77.center) to (75.center);
		\draw [bend left=45] (75.center) to (76.center);
		\draw (78.center) to (79.center);
		\draw (79.center) to (81.center);
		\draw (81.center) to (80.center);
		\draw (80.center) to (78.center);
		\draw (87.center) to (88.center);
		\draw (89.center) to (90.center);
		\draw (91.center) to (93.center);
		\draw (94.center) to (92.center);
		\draw (95.center) to (97.center);
		\draw (98.center) to (96.center);
	\end{pgfonlayer}
\end{tikzpicture}

\end{equation}
where $\sigma:n+m\to n+m$ stands for any permutation of the wires involving swap
maps. We will consider all the previous rules as being purely structural and
will not explicitly state their use. Using these rules, we can in fact deduce
that both the green and red spiders (and their labelled varieties) are
themselves flexsymmetric. This means that the language follows the OCM
meta-rule, and we can formally treat any \(\ZXp\)-diagram as a
graph\footnote{Note that the graphs in question must be allowed to have loops
  and parallel edges, so are perhaps better called \emph{pseudographs} or
  \emph{multigraphs}.} whose vertices are the spiders, and whose edges are
labelled by the \(1 \to 1\) generators of the language.\footnote{There is a
  small ambiguity: \(1 \to 1\) spiders can be treated as either edges or
  vertices. When considering diagrams, it matters little which, since any given
  graph is always to be understood as one of the many equivalent
  \(\ZXp\)-diagrams constructed formally out of the generators.
  Any computer implementation of the calculus will have to carefully resolve
  this ambiguity in its internal representation.}


Figure~\ref{fig:axioms} presents the remainder of the equational theory
\(\ZXeq\), which as we shall see axiomatises the stabiliser fragment of quantum
mechanics in the qupit ZX-calculus. Firstly though, we must be sure that all of
these rules are sound for the standard interpretation, i.e. it should not be
possible to derive an equality of diagrams whose quantum mechanical
interpretations are different.

\begin{theorem}[Soundness]
  \label{thm:soundness}
  The equational theory \(\ZXeq\) is sound for \(\interp{-}\), i.e., for any
  \(A,B \in \ZXp\), \(\ZXeq \vdash A = B\) implies \(\interp{A} = \interp{B}\).
  Put formally, \(\interp{-}\) factors through the projection \(\ZXp \to \ZXp /
  \ZXeq\).
\end{theorem}

This set of rewriting rule also turns out to be also complete:
\begin{theorem}[Completeness]
  \label{thm:completeness}
  The equational theory \(\ZXeq\) is complete for \(\interp{-}\), i.e., for any
  \(A,B \in \ZXp\), \(\interp{A} = \interp{B}\) implies \(\ZXeq \vdash A = B\).
\end{theorem}
The proof of Theorem \ref{thm:completeness} will be the object of the following
sections.

\section{Some useful structures in \(\ZXp\)}
\label{sec:structures}
The set of axioms in figure~\ref{fig:axioms} is somewhat minimalistic. In this section we present how it will be manipulated in practice.

\subsection{Elementary derivations}

While we will see that the equations in figure~\ref{fig:axioms} are enough to
derive any equality of quantum operation in the stabiliser fragment, we describe
here some rules which are derivable in \(\ZXeq\) and particularly useful. We start by the interaction between some $1\to 1$ maps:
\begin{equation}
		\input{hadanteul.tikz}
\end{equation}

The following rules, very similar to the one of \cite{bonchi_interacting_2017}, will be central:
\begin{equation}
	\input{glaeq.tikz}
\end{equation}

Here we see that the $1\to 1 $ red spider play the role of the antipode of Hopf algebra, we will then refer to it as the antipode. Finally, the pauli phases can easily move around diagrams:

\begin{equation}
	\input{paulimove.tikz}
\end{equation}

\subsection{Meta-rules}

From the equational theory follows more general patterns that we will often use as meta-rules.

\paragraph{Changing colours} The qubit ZX-calculus admits a particularly elegant
meta-rule: take any valid equation, and swap the colour of every spider while
keeping everything else identical. Then the resulting equation is also derivable
in the calculus. In fact this color swap transformation is equivalent to the functor mapping a diagram $D$ to $H^{\otimes m} D \circ H^{\otimes n}$. One might hope that such a rule would hold also in \(\ZXeq\).
Unfortunately, the picture is a little more intricate and less pretty since now the analog of the Hadamard gate is not an involution. Using
\textsc{(Colour)}, one can see what happen while Hadamrd gates are going through a diagram:

\begin{equation}
  \input{figures/equations/colour_commutation.tikz}
\end{equation}

This allows to formulate a variation of the color swap rules that holds for qudits as well:

\begin{proposition}[Colour change]
  If \(A \in \ZXp[n,m]\) is a diagram, we can consider any spider which is not the antipode as vertices linked by the following four kind of edges:
  
  \begin{center}
      \input{figures/equations/edges.tikz}
  \end{center}
  
  let \(S(A)\) be the \(\ZXp\)-diagram obtained from \(A\) by :
  \begin{enumerate}
  \item swapping the colour of every vertex;
  \item applying the following transformation on the edges: 
  
  \begin{center}
      \input{figures/equations/swapedges.tikz};
  \end{center}
  \item adding a minus to the Pauli part of every green spider.
  \end{enumerate}
  Then $\ZXeq \vdash S(A) = (\begin{tikzpicture}
	\begin{pgfonlayer}{nodelayer}
		\node [style=had] (0) at (0, 0) {};
		\node [style=none] (1) at (-1, 0) {};
		\node [style=none] (2) at (1, 0) {};
	\end{pgfonlayer}
	\begin{pgfonlayer}{edgelayer}
		\draw (0) to (1.center);
		\draw (0) to (2.center);
	\end{pgfonlayer}
\end{tikzpicture}
)^{\otimes m}
  \circ A \circ (\begin{tikzpicture}
	\begin{pgfonlayer}{nodelayer}
		\node [style=had] (0) at (0, 0) {};
		\node [style=none] (1) at (-1, 0) {};
		\node [style=none] (2) at (1, 0) {};
	\end{pgfonlayer}
	\begin{pgfonlayer}{edgelayer}
		\draw (0) to (1.center);
		\draw (0) to (2.center);
	\end{pgfonlayer}
\end{tikzpicture}
)^{\otimes n}$.
\end{proposition}
Application of this rule will simply be referred to as \textsc{(Colour)}, since
the colour change rule presented in figure \ref{fig:axioms} is a subcase. Note
that, although \(S\) is not functorial, it can be easily made into a functor,
simply by post-composing with antipodes.

Note that the original colour change rule cannot hold for \(\ZXp\), since it
does if and only if the antipode is trivial. To show this, it suffices to colour
change the \textsc{(G-Elim)} rule.

\paragraph{Spider wars} Just like in the qubit case, there are rules for fusing
spiders of the same colour, and splitting spiders of opposite colours. The
fusion rules are pretty much the same, up to the use of antipodes for red
spiders:
\begin{equation}
  \input{figures/equations/spider_fusion.tikz}
\end{equation}
They are a straightforward consequence of \textsc{(Fusion)} and
lemma~\ref{lem:loop} which eliminates self-loops. The splitting rules are more
complicated, since unlike the qubit case, we now have two elementary rules for
splitting: the Hopf identity (lemma \ref{lem:hopf}) and \textsc{(Char)}. Using a
sequential application of these rules, it is straightforward to see that, for
any \(x,y \in \N\) such that \(x \geqslant y\), we must have:
\begin{equation}
  \label{eq:spider_wars}
  \input{figures/equations/spider_wars_1.tikz}
  \quad\quad\quad
  \input{figures/equations/spider_wars_2.tikz}
\end{equation}
We dub this collection of meta-rules \textsc{(Spider)}, along with deformations
resulting from flexsymmetry.

\subsection{Syntactic sugar for multi-edges}

Before we move into the proof of completeness, we introduce some syntactic sugar
to the language. Given equation~\eqref{eq:spider_wars}, in \(\ZXp\), unlike the
qubit case, spiders of opposite colours can be connected by more that one edge,
and these multi-edges cannot be simplified. We therefore add some syntactic
sugar to represent such multi-edges. These constructions add no expressiveness
to the language, and are simply used to reduce the size of some recurring
diagrams. They are shamelessly stolen from previous work
\cite{bonchi_interacting_2017, zanasi_interacting_2018,
  carette_szx-calculus:_2019, carette_colored_2020}, and we use them to obtain a
particularly nice representation of qupit graph states \cite{zhou_quantum_2003}.
Graph states play a central role in our proof of completeness, as they have in
previous completeness results of the stabiliser fragments for dimensions \(2\)
and \(3\) \cite{backens_zx-calculus_2014, duncan_graph-theoretic_2020,
  wang_qutrit_2018}. In particular, these constructions permit a nice
presentation of how graph states evolve under local Clifford operations.

Firstly, we extend \(\ZXp\) by \emph{multipliers}, which are defined inductively
by:
\begin{equation}
  \input{figures/generators/multiplier.tikz}\quad.
\end{equation}
We also define inverted multipliers, using the standard notation for graphical
languages based on symmetric monoidal categories:
\begin{equation}
  \input{figures/generators/multiplier_inverted.tikz}\quad.
\end{equation}

Now, it is straightforward using \textsc{(Char)} to see that, for any \(m \in
\N\),
\begin{equation}
  \label{eq:multiplier_modular}
  \input{figures/equations/multiplier_modular.tikz}\quad,
\end{equation}
and we can restrict the labels of multipliers to \(\Z_p\). Explicitly, then, for
\(x \in \Z_p^*\),
\begin{equation}
  \label{eq:multiplier_explicit}
  \input{figures/equations/multiplier_explicit.tikz}\quad.
\end{equation}

\begin{proposition}
  \label{prop:multiplier}
  Multipliers satisfy the following equations under \(\ZXeq\): for any
  \(x,y\in\Z_p\) and \(z\in\Z_p^*\),
  \begin{equation}
    \input{figures/equations/multiplier.tikz}
  \end{equation}
  which amounts to saying that the multipliers form a presentation of the field
  \(\Z_p\). They also satisfy the following useful copy and elimination
  identities:
  \begin{equation}
    \input{figures/equations/multiplier2.tikz}
  \end{equation}
\end{proposition}
It is worth noting at this point, that multipliers are \emph{not} flexsymmetric,
thus OCM is technically lost and the diagrams cannot be treated as decorated
graphs if multipliers are included. This is one reason, beyond simplicity, that
multipliers are treated as syntactic sugar rather than as generators of the
language. Howerer, we can somewhat recover OCM if we are prepared to allow
\emph{directed} graphs with labelled edges instead.

Multipliers also satisfy the following equation:
\begin{equation}
  \label{eq:hadamard_multiplier}
  \input{figures/equations/hadamard_multiplier.tikz}\quad,
\end{equation}
so we can also unambiguously extend the calculus with \emph{H-boxes}:
\begin{equation}
  \label{eq:H_box_definition}
  \input{figures/generators/hadamard_weighted.tikz}\quad.
\end{equation}
H-boxes exactly match the labelled Hadamard boxes introduced for the qutrit case
\cite{townsend-teague_classifying_2021} when \(p = 3\). More explicitly, H-boxes
amount to repeated Hadamard wires:
\begin{equation}
  \input{figures/equations/weighted_hadamard_multiedge.tikz}
\end{equation}
We have that for any \(x\in\Z_p\),
\begin{equation}
  \label{eq:hadamard_OCM}
  \input{figures/equations/hadamard_inverted.tikz}
  \qand
  \input{figures/equations/hadamard_edge.tikz}\quad.
\end{equation}
so that H-boxes are flexsymmetric, unlike multipliers.

\begin{proposition}
  \label{prop:weighted_hadamard}
  \(\ZXeq\) proves the following equations:
  \begin{equation}
    \input{figures/equations/weighted_hadamard.tikz}
  \end{equation}
\end{proposition}
It follows from this that the inverse Hadamard is just the \(-1\)-weighted
Hadamard box:
\begin{equation}
    \input{figures/equations/hadamard_inverse_weight.tikz}\quad.
\end{equation}

\section{Representing Clifford states as graphs}
\label{sec:states}
Graph states are a now familiar tool of quantum information theory. Their
properties are intricately linked with those of the Clifford group and more
generally the stabiliser fragment. The original proof of completeness of the
qubit stabiliser ZX-calculus also made extensive use of graph states. The qubit
graph states have been generalised to the case of arbitrary (finite) dimension,
and we present them here through the lens of \(\ZXp\).

The idea is to associate a state in \(L^2(\Z_p)^{\otimes V}\) to any graph on a
set \(V\) of vertices. Unlike in the qubit case, it is more natural for qupits
to consider \(\Z_p\)-edge-weighted graphs. We make the choice of identifying
such a graph with its adjacency matrix \(G \in \Z_p^{V \times V}\). Then, the
corresponding graph state is defined to be
\begin{equation}
  \ket{G} = \prod_{\substack{(u,v) \in V \\ u \neq v}} E_{u,v}^{G_{u,v}}
  \bigotimes_{u\in{V}} \ket{0:X}_u.
\end{equation}
In other words, to obtain the graph state, first initialise each qupit in the
state \(\ket{0:X}\). Then, for each edge in \(G\) with weight \(G_{uv}\), apply
the entangling operation \(E_{uv}^{G_{uv}}\). Since \(E\) gates commute, they
can be applied in any order.

As a simple example, the graph
\begin{equation}
  \label{eq:simple_graph}
  \begin{tikzpicture}
	\begin{pgfonlayer}{nodelayer}
		\node [style={graph_vertex}] (0) at (-1, 1) {};
		\node [style={graph_vertex}] (1) at (-1, -1) {};
		\node [style={graph_vertex}] (2) at (1, 0) {};
		\node [style={graph_weight}] (3) at (0, 0.5) {3};
	\end{pgfonlayer}
	\begin{pgfonlayer}{edgelayer}
		\draw (2) to (1);
		\draw (1) to (0);
		\draw (0) to (2);
	\end{pgfonlayer}
\end{tikzpicture}
 \quad \qq{has adjacency matrix}
  \begin{pmatrix}
    0 & 3 & 1 \\
    3 & 0 & 1 \\
    1 & 1 & 0
  \end{pmatrix},
\end{equation}
and is associated with the graph state \(E_{1,2} E_{1,3}^3 E_{2,3} (\ket{0:X}
\otimes \ket{0:X} \otimes \ket{0:X})\).

\subsection{Graph states in \(\ZXp\)}

The associated graph state is represented in \(\ZXp\) as:
\begin{equation}
  \label{eq:simple_graph_state}
  \input{figures/simple_graph_state.tikz}
\end{equation}
where the equality is obtained solely by fusing the green spiders. One
immediately recognises the graph from equation~\eqref{eq:simple_graph} in the
RHS of equation~\eqref{eq:simple_graph_state}.

Here is a slightly more involved example:
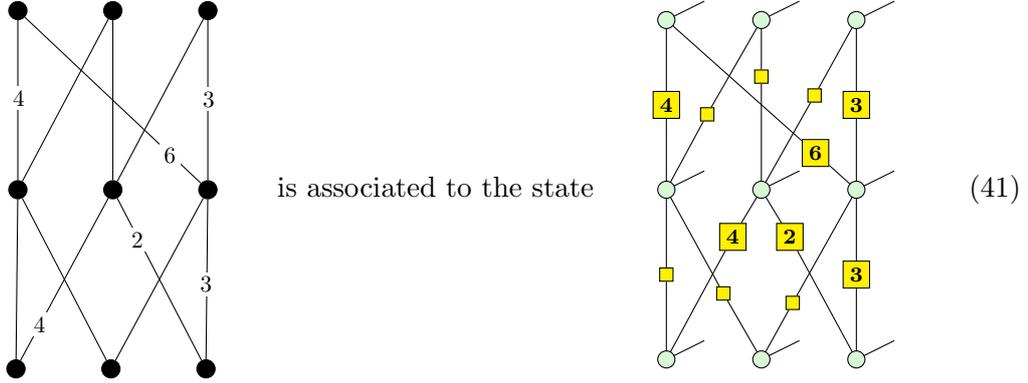
\begin{equation}
  \input{figures/graph.tikz}
  \quad \qq{is associated to the state} \quad
  \input{figures/graph_state.tikz}
\end{equation}
where each spider is connected to an output wire.

Hopefully, it should be clear that from these examples that for any given graph
\(G \in \Z_p^{V \times V}\), one can obtain the \(\ZXp\)-diagram for the corresponding
graph state by identifying each vertex with a green spider, and each edge with a
correspondingly weight H-edge. More formally:
\begin{definition}
  A \(\ZXp\)-diagram is a graph state diagram if
  \begin{enumerate}
  \item it contains only green spiders;
  \item every spider is connected to a single output by a plain wire;
  \item the spiders are connected only by H-edges.
  \end{enumerate}
\end{definition}

Every graph state diagram is uniquely associated to a graph, and vice-versa.
Since we are going to reason about such diagrams in generality and without
referring to a specific graph, we use the following informal notation to
represent the \(\ZXp\)-diagram associated to \(G\):
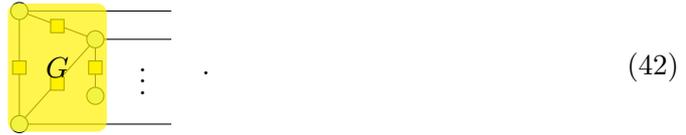
\begin{equation}
  \input{figures/graph_state_notation.tikz}\quad.
\end{equation}
In order to avoid having to track uninteresting scalars, we assume that the
interpretation of this diagram is normalised. This amounts to including, along
with the purely graph-theoretical part, a pair of scalars
\begin{equation}
  \left( \begin{tikzpicture}
	\begin{pgfonlayer}{nodelayer}
		\node [style=gn] (0) at (-0.75, 0) {};
		\node [style=rn] (1) at (0.75, 0) {};
		\node [style=rn] (2) at (0, 0) {};
	\end{pgfonlayer}
	\begin{pgfonlayer}{edgelayer}
		\draw [bend right=60] (0) to (1);
		\draw (2) to (0);
		\draw (2) to (1);
		\draw [bend left=60] (0) to (1);
	\end{pgfonlayer}
\end{tikzpicture}
 \right)^{\otimes\abs{V}}
  \qand
  \left( \begin{tikzpicture}
	\begin{pgfonlayer}{nodelayer}
		\node [style=gn] (0) at (-0.375, 0) {};
		\node [style=rn] (1) at (0.375, 0) {};
	\end{pgfonlayer}
	\begin{pgfonlayer}{edgelayer}
		\draw (0) to (1);
	\end{pgfonlayer}
\end{tikzpicture}
 \right)^{\otimes\abs{E}},
\end{equation}
where \(V\) is the set of vertices in \(G\) and \(E\) the set of edges. These
are simply the normalisation factors for states and entangling gates in G (see
equations~\eqref{eq:C1_generators} and \eqref{eq:stabiliser_states}). It should
now be clear that
\begin{equation}
  \ket{G} = \interp{\input{figures/graph_state_notation.tikz}}\quad.
\end{equation}
Note that this notation for graph states can be formalised using the
``scalable'' construction \cite{carette_colored_2020}, but this would involve
introducing rather more machinery than we really need.

\subsection{Manipulating graph states}

Having shown how to represent graph states within \(\ZXp\), we now gives some
natural rules for manipulating diagrams involving them under \(\ZXeq\).

Two graph states \(\ket{G}\) and \(\ket{H}\) are said to be \emph{local Clifford
equivalent} is there is a sequence of local Clifford operations that maps
\(\ket{G}\) to \(\ket{H}\). It was show in \cite{bahramgiri_efficient_2007} that
any local Clifford equivalence can be decomposed as a sequence of elementary
local Clifford operations called \emph{local scaling} and \emph{local
  complementation}. These are particularly nice operations since their actions
on a graph state \(\ket{G}\) are straightforward to understand at the level of
the graph \(G\) defining the graph state.

\subsubsection*{Pauli stabilisers}

It is well-known that graph states admit simple Pauli stabilisers, namely, for
each \(v \in V\),
\begin{equation}
  X_v^\gamma \prod_{w \in N(v)} Z_w^{\gamma G_{vw}} \ket{G} = \ket{G}.
\end{equation}
We can give a simple formulation of these rules within \(\ZXp\), and they can be
derived in \(\ZXeq\):
\begin{proposition}
  \label{prop:pauli_stabiliser}
  \(\ZXeq\) proves the Pauli stabiliser rules for graph states:
  \begin{equation*}
    \input{figures/equations/pauli_stabiliser.tikz},
  \end{equation*}
  where the red spiders is connected to vertex \(w\) in the graph \(G\), and
  each neighbour \(k\) of \(w\) is connected to a green vertex with phase
  \((\gamma G_{kw},0)\).
\end{proposition}

\subsubsection*{Local scaling}

For any \(\gamma \in \Z_p^*\), the \emph{\(\gamma\)-scaling} about a vertex
\(w\) in a graph \(G\) is given by:
\begin{equation}
  (G \overset{\gamma}{\circ} w)_{uv} \coloneqq
  \begin{cases}
    \gamma G_{uv} \quad \text{if} \quad u = w \quad \text{or} \quad v = w;\\
    G_{uv} \quad \text{otherwise}.
  \end{cases}
\end{equation}

In other words, we apply a multiplicative scaling to all of the edges in the
neighbourhood of \(w\). For example:
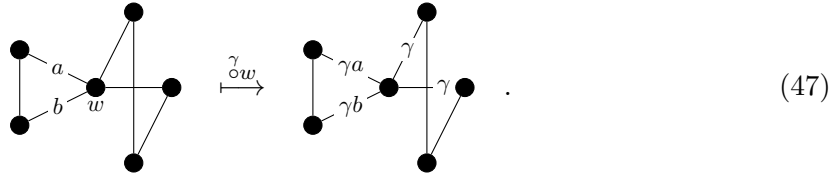
\begin{equation}
  \input{figures/local_scaling_example_lhs.tikz}
  \quad \overset{\overset{\gamma}{\circ} w}{\longmapsto} \quad
  \input{figures/local_scaling_example_rhs.tikz} \quad.
\end{equation}

\begin{proposition}
  \label{prop:local_scaling}
  Local scaling is derivable in \(\ZXeq\): for any graph \(G \in \Z_p^{N
    \times N}\), \(\gamma \in \Z_p\) and \(w \in \{1,2,\cdots,N\}\), \(\ZXeq\)
  proves that
  \begin{equation}
    \input{figures/equations/local_scaling.tikz} \quad,
  \end{equation}
  where in the LHS, the multiplier is connected to vertex \(w\).
\end{proposition}

\subsubsection*{Local complementation}

For any \(\gamma \in \Z_d^*\), the \emph{\(\gamma\)-weighted local
  \(\Z_d\)-complementation} or \(\gamma\)-complementation about a vertex \(w\)
in a graph \(G \i  \Z_d^{V \times V}\) is defined as:
\begin{equation}
  (G \overset{\gamma}{\star} w)_{uv} \coloneqq
  \begin{cases}
    G_{uv} + \gamma G_{uw} G_{wv} \qif u \neq v; \\
    G_{uv} \qq{otherwise.}
  \end{cases}
\end{equation}

This operation is somewhat harder than local scaling to get a good intuition
for. It essentially operates on ``cones'' in \(G\) with summit \(w\). The
simplest example is the following:
\begin{equation}
  \input{figures/local_complementation_triangle_lhs.tikz}
  \quad \overset{\overset{\gamma}{\star} w}{\longmapsto} \quad
  \input{figures/local_complementation_triangle_rhs.tikz} \quad.
\end{equation}
In a more complicated graph, local complementation about \(w\) performs this
simple operation for every such ``cone'' with summit \(w\). For example,
\begin{equation}
  \input{figures/local_complementation_example_lhs.tikz}
  \quad \overset{\overset{\gamma}{\star} w}{\longmapsto} \quad
  \input{figures/local_complementation_example_rhs.tikz} \quad.
\end{equation}

\begin{proposition}
  \label{prop:local_complementation}
  Local complementation is derivable in \(\ZXeq\): for any graph \(G \in \Z_p^{N
    \times N}\), \(\gamma \in \Z_p\) and \(w \in \{1,2,\cdots,N\}\), \(\ZXeq\)
  proves that
  \begin{equation}
    \input{figures/equations/local_complementation.tikz} \quad,
  \end{equation}
  where in the LHS, the red phase is connected to vertex \(w\), and each
  neighbour \(v\) of \(w\) is connected to a green phase \((0,-\gamma G_{vw}^2)\).
\end{proposition}

Combining this with propositions~\ref{prop:local_scaling} and
\ref{prop:pauli_stabiliser}, we get
\begin{corollary}
  \label{cor:GS_red_rotation}
  For any graph \(G \in \Z_p^{N \times N}\), \(w \in \{1,2,\cdots,N\}\) and
  \(x,y \in \Z_p\), \(\ZXeq\) proves that
  \begin{equation}
    \input{figures/equations/GS_red_rotation.tikz} \quad,
  \end{equation}
  where in the LHS, the red phase is connected to vertex \(w\), and each
  neighbour \(v\) of \(w\) is connected to a green phase with labels
  proportional to the edge weight \(G_{vw}\).
\end{corollary}

\section{Completeness}
\label{sec:completeness}
We now have all the necessary tools to show that our equational theory is
complete. The structure of our proof is similar to the one used to show the
completeness in the qubit case \cite{backens_zx-calculus_2014}. However, if the
overall scheme is very similar, each step separately can involve different
approaches more suited to the qupit situation. The plan is as follows. We
identify a family of scalars, called \emph{elementary scalars}, which correspond
to those which appear when applying the rewrites of \(\ZXeq\). We first show the
completeness up to non-zero elementary scalars, which allows us to work with
simpler diagrams without taking care of all the invertible scalars appearing
along the way. Then, we show the completeness for elementary scalars
independently, leading to a general proof of completeness. The proof of
completeness up to non-zero elementary scalars, goes like this:

\begin{itemize}
	\item Take two diagrams with the same interpretation.
	
	\item Put them in GS+LC form.
	
	\item Define the notion of rGS+LC form for a pair of diagrams in which some vertices are marked. Then show that in a rGS+LC pair if a vertex is marked on one side, it must also be marked on the other side, else the two diagrams cannot have the same interpretation.
	
	\item Show that two diagram forming an rGS+LC pair such that their marked vertices matches and having the same interpretation are equal modulo the equational theory. 
	
\end{itemize}

\subsection{Elementary scalars}

The following is standard from categorical quantum mechanics:
\begin{lemma}
  If \(A,B \in \ZXp[0,0]\), then \(\interp{A \otimes B} = \interp{A} \cdot
  \interp{B} = \interp{A \circ B}\), where \(\cdot\) is the usual multiplication
  on \(\C\) restricted to the monoid \(\mathbb{G}\).
\end{lemma}



Now, as when we were defining the group of phases \(\mathbb{P}_p\), the set of
normal forms for phases must depend on the prime \(p\) in question:
\begin{definition}
  An \emph{elementary scalar} is a diagram \(A \in \ZXp[0,0]\) which is a tensor
  product of diagrams from the collection \(O_p \cup P \cup Q\): where
  \begin{itemize}
  \item if \(p \equiv 1 \mod 4\),
    \begin{equation}
      O_p = \left\{
        \begin{tikzpicture}
	\begin{pgfonlayer}{nodelayer}
		\node [style=none] (0) at (-0.75, 0.75) {};
		\node [style=none] (1) at (-0.75, -0.75) {};
		\node [style=none] (2) at (0.75, -0.75) {};
		\node [style=none] (3) at (0.75, 0.75) {};
	\end{pgfonlayer}
	\begin{pgfonlayer}{edgelayer}
		\draw [style={dash_edge}] (3.center)
			 to (0.center)
			 to (1.center)
			 to (2.center)
			 to cycle;
	\end{pgfonlayer}
\end{tikzpicture}
,
        \scalebox{1.4}{$\star$}
      \right\} \quad,
    \end{equation}
  \item if \(p \equiv 3 \mod 4\),
    \begin{equation}
      O_p = \left\{
        \begin{tikzpicture}
	\begin{pgfonlayer}{nodelayer}
		\node [style=none] (0) at (-0.75, 0.75) {};
		\node [style=none] (1) at (-0.75, -0.75) {};
		\node [style=none] (2) at (0.75, -0.75) {};
		\node [style=none] (3) at (0.75, 0.75) {};
	\end{pgfonlayer}
	\begin{pgfonlayer}{edgelayer}
		\draw [style={dash_edge}] (3.center)
			 to (0.center)
			 to (1.center)
			 to (2.center)
			 to cycle;
	\end{pgfonlayer}
\end{tikzpicture}
,
        \begin{tikzpicture}
	\begin{pgfonlayer}{nodelayer}
		\node [style=gn, label={[gphase]below:$0,1$}] (0) at (0, 0) {};
	\end{pgfonlayer}
\end{tikzpicture}
,
        \scalebox{1.4}{$\star$},
        \begin{tikzpicture}
	\begin{pgfonlayer}{nodelayer}
		\node [style=gn, label={[gphase]below:$0,1$}] (0) at (0.375, 0) {};
		\node [style=none] (1) at (-0.375, 0) {\scalebox{1.4}{$\star$}};
	\end{pgfonlayer}
\end{tikzpicture}

      \right\} \quad,
    \end{equation}
  \end{itemize}
  \begin{equation}
    P = \left\{
      \begin{tikzpicture}
	\begin{pgfonlayer}{nodelayer}
		\node [style=none] (0) at (-0.75, 0.75) {};
		\node [style=none] (1) at (-0.75, -0.75) {};
		\node [style=none] (2) at (0.75, -0.75) {};
		\node [style=none] (3) at (0.75, 0.75) {};
	\end{pgfonlayer}
	\begin{pgfonlayer}{edgelayer}
		\draw [style={dash_edge}] (3.center)
			 to (0.center)
			 to (1.center)
			 to (2.center)
			 to cycle;
	\end{pgfonlayer}
\end{tikzpicture}
,
      \begin{tikzpicture}
	\begin{pgfonlayer}{nodelayer}
		\node [style=gn, label={[gphase]below:$s,0$}] (0) at (-0.5, 0) {};
		\node [style=rn, label={[rphase]above:$1,0$}] (1) at (0.25, 0) {};
	\end{pgfonlayer}
	\begin{pgfonlayer}{edgelayer}
		\draw (1) to (0);
	\end{pgfonlayer}
\end{tikzpicture}

      \mid s \in \Z_p^* \right\} \quad,
  \end{equation}
  and
  \begin{equation}
    Q = \left\{
      \begin{tikzpicture}
	\begin{pgfonlayer}{nodelayer}
		\node [style=gn] (2) at (-0.5, 0) {};
		\node [style=rn] (3) at (0.5, 0) {};
		\node [style=none] (4) at (-1, 0) {$\Large($};
		\node [style=none] (5) at (1.375, 0) {$\Large)^{\otimes r}$};
	\end{pgfonlayer}
	\begin{pgfonlayer}{edgelayer}
		\draw (2) to (3);
	\end{pgfonlayer}
\end{tikzpicture}
,
      \begin{tikzpicture}
	\begin{pgfonlayer}{nodelayer}
		\node [style=none] (0) at (-0.75, 0.75) {};
		\node [style=none] (1) at (-0.75, -0.75) {};
		\node [style=none] (2) at (0.75, -0.75) {};
		\node [style=none] (3) at (0.75, 0.75) {};
	\end{pgfonlayer}
	\begin{pgfonlayer}{edgelayer}
		\draw [style={dash_edge}] (3.center)
			 to (0.center)
			 to (1.center)
			 to (2.center)
			 to cycle;
	\end{pgfonlayer}
\end{tikzpicture}
,
      \begin{tikzpicture}
	\begin{pgfonlayer}{nodelayer}
		\node [style=gn] (2) at (-0.75, 0) {};
		\node [style=rn] (3) at (0.75, 0) {};
		\node [style=rn] (4) at (0, 0) {};
		\node [style=none] (5) at (-1.25, 0) {$\Large($};
		\node [style=none] (6) at (1.625, 0) {$\Large)^{\otimes r}$};
	\end{pgfonlayer}
	\begin{pgfonlayer}{edgelayer}
		\draw [bend right=60] (2) to (3);
		\draw (4) to (2);
		\draw (4) to (3);
		\draw [bend left=60] (2) to (3);
	\end{pgfonlayer}
\end{tikzpicture}

    \mid r \in \Z \right\} \quad. 
  \end{equation}
  
  If \(A,B \in \ZXp\), we say that \(A\) and \(B\) are equal \emph{up to an
    elementary scalar} if there is an elementary scalar \(C\) such that \(A = B
  \otimes C\). In that case, we write \(A \simeq B\).
\end{definition}
Comparing with the definition of \(\mathbb{G}\), the interpretation of the
elements of \(P\) correspond to powers of \(\omega\), the elements of \(Q\) to
(possible negative) powers of \(\sqrt{p}\), and the elements of \(O_p\) to
powers of \(-1\) or \(i\) depending on the value of \(p\). This remark will
naturally lead to the normal for for scalars in a few sections.

Now, as written, equality up to an elementary scalar might seem like a relation
that is not symmetric and therefore not an equivalence relation.

\begin{proposition}
  Every elementary scalar \(C \in \ZXp[0,0]\) has a multiplicative inverse, i.e.
  an elementary scalar \(C^{-1} \in \ZXp[0,0]\) such that
  \begin{equation}
    C \otimes C^{-1} = C \circ C^{-1} = \begin{tikzpicture}
	\begin{pgfonlayer}{nodelayer}
		\node [style=none] (0) at (-0.75, 0.75) {};
		\node [style=none] (1) at (-0.75, -0.75) {};
		\node [style=none] (2) at (0.75, -0.75) {};
		\node [style=none] (3) at (0.75, 0.75) {};
	\end{pgfonlayer}
	\begin{pgfonlayer}{edgelayer}
		\draw [style={dash_edge}] (3.center)
			 to (0.center)
			 to (1.center)
			 to (2.center)
			 to cycle;
	\end{pgfonlayer}
\end{tikzpicture}
\quad.
  \end{equation}
\end{proposition}
\begin{lproof}
  This is the content of lemmas~\ref{lem:scalar_elementary},
  \ref{lem:scalar_phase}, \ref{lem:scalar_i_elim}, \ref{lem:scalar_omega_elim}
  and axioms \textsc{(Gauss)} and \textsc{(M-One)}.
\end{lproof}

In light of this fact, if \(A \simeq B\), there is an elementary scalar \(C\)
such that \(A = B \otimes C\), and then \(B = B \otimes C \otimes C^{-1} = A
\otimes C^{-1}\), so that \(B \simeq A\).

\begin{proposition}
  Every equation in \(\ZXeq\) can be loosened to equality up an elementary
  scalar by erasing every part of the LHS and RHS diagrams which is disconnected
  from the inputs and outputs.
\end{proposition}
\begin{lproof}
  This comes from a straightforward structural induction, and noting that every
  scalar diagram which appears in the axiomatisation (figure~\ref{fig:axioms})
  can be straightforwardly rewritten under \(\ZXeq\) to an elementary scalar,
  and thus can be ignored.
\end{lproof}

Probably the most important case of equality up to elementary scalars is the
completeness of the single-qupit Clifford groups, on which the entire proof of
completeness of the calculus rests. The fragment of \(\ZXp\) which corresponds
to \(\mathscr{C}_1\) is that generated by the \(1 \to 1\) diagrams under
sequential composition only:
\begin{equation*}
 \begin{tikzpicture}
	\begin{pgfonlayer}{nodelayer}
		\node [style=rn, label={[rphase]below:$x,y$}] (0) at (0, 0) {};
		\node [style=none] (1) at (1, 0) {};
		\node [style=none] (2) at (-1, 0) {};
	\end{pgfonlayer}
	\begin{pgfonlayer}{edgelayer}
		\draw (0) to (1.center);
		\draw (2.center) to (0);
	\end{pgfonlayer}
\end{tikzpicture}
 \quad\quad
 \begin{tikzpicture}
	\begin{pgfonlayer}{nodelayer}
		\node [style=gn, label={[gphase]below:$x,y$}] (0) at (0, 0) {};
		\node [style=none] (1) at (1, 0) {};
		\node [style=none] (2) at (-1, 0) {};
	\end{pgfonlayer}
	\begin{pgfonlayer}{edgelayer}
		\draw (0) to (1.center);
		\draw (2.center) to (0);
	\end{pgfonlayer}
\end{tikzpicture}
 \quad\quad
 \begin{tikzpicture}
	\begin{pgfonlayer}{nodelayer}
		\node [style=none] (0) at (-1, 0) {};
		\node [style=none] (1) at (1, 0) {};
		\node [style=rmat] (2) at (0, 0) {$x$};
	\end{pgfonlayer}
	\begin{pgfonlayer}{edgelayer}
		\draw (0.center) to (2);
		\draw (2) to (1.center);
	\end{pgfonlayer}
\end{tikzpicture}
 \quad\quad
 \begin{tikzpicture}
	\begin{pgfonlayer}{nodelayer}
		\node [style=none] (0) at (-3, 0) {};
		\node [style=none] (1) at (-1, 0) {};
		\node [style=had] (2) at (-2, 0) {$x$};
	\end{pgfonlayer}
	\begin{pgfonlayer}{edgelayer}
		\draw (2) to (0.center);
		\draw (2) to (1.center);
	\end{pgfonlayer}
\end{tikzpicture}

\end{equation*}
We call any such diagram a \emph{single-qupit Clifford diagram} or
\emph{\(\mathscr{C}_1\)-diagram}.
\begin{proposition}
  \label{prop:c1_completeness}
  If \(A \in \ZXp[1,1]\) is a single-qupit Clifford diagram, then \(\ZXeq\)
  proves that
  \begin{equation}
    \input{figures/equations/c1_normal_form.tikz}\quad,
  \end{equation}
  for some \(s,t,u,v \in \Z_p\) and \(w \in \Z_p^*\). Furthermore,
  this form is unique.
\end{proposition}

\subsection{Relating stabiliser diagrams to graphs}
                     
In order to use the preceding tools in our completeness proof for the whole
stabiliser fragment, we need to understand how an arbitrary stabiliser diagram
can be related to a graph state diagram. Firstly, we rewrite any stabiliser
diagram \(D\) to a state using the Choi-Jamiolkowski isomorphism:
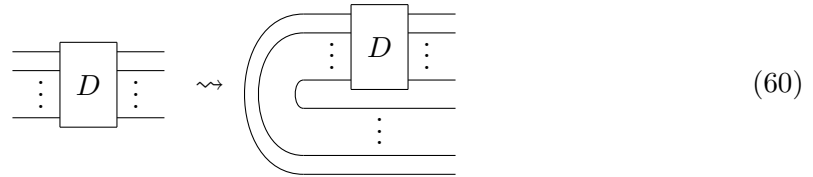
\begin{equation}
  \label{eq:choi_isomorphism}
  \input{figures/equations/choi_stabiliser_graph_LHS.tikz}
  \quad{\Large\rightsquigarrow}
  \input{figures/equations/choi_stabiliser_graph_RHS.tikz}
\end{equation}
Any diagram \(0 \to n\) obtained from the Choi-Jamiolkowski isomorphism
can be rewritten to a graph state diagram, with single qupit Clifford operations
acting on its output wires:
\begin{definition}[\cite{backens_zx-calculus_2014, wang_qutrit_2018}]
  A \emph{GS+LC diagram} is a \(\ZXp\)-diagram which consists of a graph state
  diagram with arbitrary single-qupit Clifford operations applied to each
  output. These associated Clifford operations are called \emph{vertex
    operators}.
\end{definition}

\begin{proposition}
  \label{prop:GS+LC}
  Every stabiliser state \(\ZXp\)-diagram can be rewritten, up to elementary
  scalars, to GS+LC form under \(\ZXeq\).
\end{proposition}

In other words, \(\ZXeq\) proves that, for any stabiliser \(\ZXp\)-diagram \(D : m
\to n\), there is a graph \(G\) on \(m+n\) vertices and a set of vertex
operators \((v_k)_{k=1}^{m+n}\) such that
\begin{equation}
  \input{figures/equations/stabiliser_graph_LHS.tikz}
  \quad\simeq\quad
  \input{figures/equations/stabiliser_graph_RHS.tikz}
\end{equation}
and we need only consider the question of whether \(\ZXeq\) can prove the
equality of two GS+LC diagrams.

\subsection{Completeness modulo elementary scalars}

Now, as a first step, we show that \(\ZXeq\) can normalise any pair of diagrams
with equal interpretations, up to elementary scalars. In particular, as was
shown in the previous section, we can relax \(\ZXeq\) to reason about equality up to
elementary scalars by simply ignoring the scalar part of each rule, and make free
use of the ``scalarless'' equational theory. We will take care of the resulting
scalars in the next section.

This part of the completeness proof follows the general ideas of
\cite{backens_zx-calculus_2014}. The first step on the way to completeness is
to note that, considering a diagram in GS+LC form, where the vertex operators
have been normalised, we can obtain a yet more reduced diagram by absorbing as
much as possible of the vertex operators into local scalings and local
complementations. We then obtain the following form for the vertex operators:
\begin{definition}[\cite{backens_zx-calculus_2014}] 
  \label{def:rGS+LC}
  A \(\ZXp\)-diagram is in \emph{reduced} GS+LC (or rGS+LC) form if it is in
  GS+LC form, and furthermore:
  \begin{enumerate}
  \item All vertex operators belong to the following set:
    \begin{equation}
      R = \left\{ \input{figures/equations/c1_reduced.tikz} \,\middle|\, s,t \in \Z_p \right\}.
    \end{equation}
  \item Two adjacent vertices do not have vertex operators that both include red spiders.
  \end{enumerate}
\end{definition}

\begin{proposition}
  \label{prop:rGS+LC}
  If \(D \in \ZXp\) is a Clifford diagram, then there is a diagram \(G
  \in \ZXp\) in rGS+LC form such that \(\ZXeq \vdash D \simeq G\).
\end{proposition}

Then, given two diagrams with equal interpretations, taking them both to rGS+LC
makes the task of comparing the diagrams considerably easier. In particular, we
can guarantee that the corresponding vertex operators in each diagram always
have matching forms:
\begin{definition}[\cite{backens_zx-calculus_2014}] 
  A pair of rGS+LC diagrams of the same type (i.e. whose graphs have the same vertex set
  \(V\)) is said to be \emph{simplified} if there is no pair of vertices \(q,p
  \in V\) such that \(q\) has a red vertex operator in the first diagram but not
  the second, \(q\) has a red vertex operator in the second diagram but not the
  first, and \(q\) and \(p\) are adjacent in at least one of the diagrams.
\end{definition}

\begin{proposition}
  \label{prop:rGS+LC_simple}
  Any pair \(A,B\) of rGS+LC diagrams of the same type (i.e. on the same vertex
  set) can be simplified.
\end{proposition}

For the sake of clarity, we shall say that the vertex operator (or equivalently,
the vertex itself) is \emph{marked} if it contains a red spider (i.e. it belongs
to the right-hand form of definition~\ref{def:rGS+LC}). Then, two diagrams with
the same interpretation can always be rewriten so that the marked vertices
match:
\begin{proposition}\label{prop:marked}
  Let \(C,D \in \ZXp\) be a simplified pair in rGS+LC form, then \(\interp{C} =
  \interp{D}\) only if the marked vertices in $C$ and $D$ are the same.
\end{proposition}

Finally, we have enough control over the pair of diagrams to finish the
completeness proof: 
\begin{theorem}\label{thm:comp}
  \(\ZXeq\) is complete for \(\interp{-}\), i.e. if for any pair of
  diagrams \(A,B \in \ZXp [0,n]\) with $n\neq 0$ , \(\interp{A} = \interp{B}\), then \(\ZXeq
  \vdash A \simeq B\).
\end{theorem}

\subsection{Completeness of the scalar fragment}


Finally, we are ready to leap-frog off of the previous section into the full
completeness (including scalars). First, we need to find a normal form for
diagrams which evaluate to \(0\). In fact, we need pick one normal form for each
type \(m \to n\):
\begin{proposition}
  \label{prop:zero_normal_forms}
  The ``zero'' scalar ``destroys'' diagrams: for any \(m,n \in \N\) and \(D \in
  \ZXp[m,n]\),
  \begin{equation*}
    \input{figures/equations/zero_normal_forms.tikz}
  \end{equation*}
  We take the RHS diagram to be the ``zero'' diagram of type \(m \to n\).
\end{proposition}

Now, we say that a scalar \(C \in \ZXp[0,0]\) is in \emph{normal form} if it is
either the zero diagram, or it belongs to the set
\begin{equation}
  \label{eq:scalar_normal_form_p1}
  \left\{
    \begin{tikzpicture}
	\begin{pgfonlayer}{nodelayer}
		\node [style=none] (0) at (-0.75, 0.75) {};
		\node [style=none] (1) at (-0.75, -0.75) {};
		\node [style=none] (2) at (0.75, -0.75) {};
		\node [style=none] (3) at (0.75, 0.75) {};
	\end{pgfonlayer}
	\begin{pgfonlayer}{edgelayer}
		\draw [style={dash_edge}] (3.center)
			 to (0.center)
			 to (1.center)
			 to (2.center)
			 to cycle;
	\end{pgfonlayer}
\end{tikzpicture}
,
    \scalebox{1.4}{$\star$}
  \right\}
  \otimes \left\{
    \begin{tikzpicture}
	\begin{pgfonlayer}{nodelayer}
		\node [style=none] (0) at (-0.75, 0.75) {};
		\node [style=none] (1) at (-0.75, -0.75) {};
		\node [style=none] (2) at (0.75, -0.75) {};
		\node [style=none] (3) at (0.75, 0.75) {};
	\end{pgfonlayer}
	\begin{pgfonlayer}{edgelayer}
		\draw [style={dash_edge}] (3.center)
			 to (0.center)
			 to (1.center)
			 to (2.center)
			 to cycle;
	\end{pgfonlayer}
\end{tikzpicture}
,
    \begin{tikzpicture}
	\begin{pgfonlayer}{nodelayer}
		\node [style=gn, label={[gphase]below:$s,0$}] (0) at (-0.5, 0) {};
		\node [style=rn, label={[rphase]above:$1,0$}] (1) at (0.25, 0) {};
	\end{pgfonlayer}
	\begin{pgfonlayer}{edgelayer}
		\draw (1) to (0);
	\end{pgfonlayer}
\end{tikzpicture}

    \mid s \in \Z_p^* \right\}
  \otimes \left\{
    \begin{tikzpicture}
	\begin{pgfonlayer}{nodelayer}
		\node [style=gn] (2) at (-0.5, 0) {};
		\node [style=rn] (3) at (0.5, 0) {};
		\node [style=none] (4) at (-1, 0) {$\Large($};
		\node [style=none] (5) at (1.375, 0) {$\Large)^{\otimes r}$};
	\end{pgfonlayer}
	\begin{pgfonlayer}{edgelayer}
		\draw (2) to (3);
	\end{pgfonlayer}
\end{tikzpicture}
,
    \begin{tikzpicture}
	\begin{pgfonlayer}{nodelayer}
		\node [style=none] (0) at (-0.75, 0.75) {};
		\node [style=none] (1) at (-0.75, -0.75) {};
		\node [style=none] (2) at (0.75, -0.75) {};
		\node [style=none] (3) at (0.75, 0.75) {};
	\end{pgfonlayer}
	\begin{pgfonlayer}{edgelayer}
		\draw [style={dash_edge}] (3.center)
			 to (0.center)
			 to (1.center)
			 to (2.center)
			 to cycle;
	\end{pgfonlayer}
\end{tikzpicture}
,
    \begin{tikzpicture}
	\begin{pgfonlayer}{nodelayer}
		\node [style=gn] (2) at (-0.75, 0) {};
		\node [style=rn] (3) at (0.75, 0) {};
		\node [style=rn] (4) at (0, 0) {};
		\node [style=none] (5) at (-1.25, 0) {$\Large($};
		\node [style=none] (6) at (1.625, 0) {$\Large)^{\otimes r}$};
	\end{pgfonlayer}
	\begin{pgfonlayer}{edgelayer}
		\draw [bend right=60] (2) to (3);
		\draw (4) to (2);
		\draw (4) to (3);
		\draw [bend left=60] (2) to (3);
	\end{pgfonlayer}
\end{tikzpicture}

    \mid r \in \Z \right\},
\end{equation}
if \(p \equiv 1 \mod 4\), or to the set
\begin{equation}
  \label{eq:scalar_normal_form_p3}
  \left\{
    \begin{tikzpicture}
	\begin{pgfonlayer}{nodelayer}
		\node [style=none] (0) at (-0.75, 0.75) {};
		\node [style=none] (1) at (-0.75, -0.75) {};
		\node [style=none] (2) at (0.75, -0.75) {};
		\node [style=none] (3) at (0.75, 0.75) {};
	\end{pgfonlayer}
	\begin{pgfonlayer}{edgelayer}
		\draw [style={dash_edge}] (3.center)
			 to (0.center)
			 to (1.center)
			 to (2.center)
			 to cycle;
	\end{pgfonlayer}
\end{tikzpicture}
,
    \begin{tikzpicture}
	\begin{pgfonlayer}{nodelayer}
		\node [style=gn, label={[gphase]below:$0,1$}] (0) at (0, 0) {};
	\end{pgfonlayer}
\end{tikzpicture}
,
    \scalebox{1.4}{$\star$},
    \begin{tikzpicture}
	\begin{pgfonlayer}{nodelayer}
		\node [style=gn, label={[gphase]below:$0,1$}] (0) at (0.375, 0) {};
		\node [style=none] (1) at (-0.375, 0) {\scalebox{1.4}{$\star$}};
	\end{pgfonlayer}
\end{tikzpicture}

  \right\}
  \otimes \left\{
    \begin{tikzpicture}
	\begin{pgfonlayer}{nodelayer}
		\node [style=none] (0) at (-0.75, 0.75) {};
		\node [style=none] (1) at (-0.75, -0.75) {};
		\node [style=none] (2) at (0.75, -0.75) {};
		\node [style=none] (3) at (0.75, 0.75) {};
	\end{pgfonlayer}
	\begin{pgfonlayer}{edgelayer}
		\draw [style={dash_edge}] (3.center)
			 to (0.center)
			 to (1.center)
			 to (2.center)
			 to cycle;
	\end{pgfonlayer}
\end{tikzpicture}
,
    \begin{tikzpicture}
	\begin{pgfonlayer}{nodelayer}
		\node [style=gn, label={[gphase]below:$s,0$}] (0) at (-0.5, 0) {};
		\node [style=rn, label={[rphase]above:$1,0$}] (1) at (0.25, 0) {};
	\end{pgfonlayer}
	\begin{pgfonlayer}{edgelayer}
		\draw (1) to (0);
	\end{pgfonlayer}
\end{tikzpicture}

    \mid s \in \Z_p^* \right\}
  \otimes \left\{
    \begin{tikzpicture}
	\begin{pgfonlayer}{nodelayer}
		\node [style=gn] (2) at (-0.5, 0) {};
		\node [style=rn] (3) at (0.5, 0) {};
		\node [style=none] (4) at (-1, 0) {$\Large($};
		\node [style=none] (5) at (1.375, 0) {$\Large)^{\otimes r}$};
	\end{pgfonlayer}
	\begin{pgfonlayer}{edgelayer}
		\draw (2) to (3);
	\end{pgfonlayer}
\end{tikzpicture}
,
    \begin{tikzpicture}
	\begin{pgfonlayer}{nodelayer}
		\node [style=none] (0) at (-0.75, 0.75) {};
		\node [style=none] (1) at (-0.75, -0.75) {};
		\node [style=none] (2) at (0.75, -0.75) {};
		\node [style=none] (3) at (0.75, 0.75) {};
	\end{pgfonlayer}
	\begin{pgfonlayer}{edgelayer}
		\draw [style={dash_edge}] (3.center)
			 to (0.center)
			 to (1.center)
			 to (2.center)
			 to cycle;
	\end{pgfonlayer}
\end{tikzpicture}
,
    \begin{tikzpicture}
	\begin{pgfonlayer}{nodelayer}
		\node [style=gn] (2) at (-0.75, 0) {};
		\node [style=rn] (3) at (0.75, 0) {};
		\node [style=rn] (4) at (0, 0) {};
		\node [style=none] (5) at (-1.25, 0) {$\Large($};
		\node [style=none] (6) at (1.625, 0) {$\Large)^{\otimes r}$};
	\end{pgfonlayer}
	\begin{pgfonlayer}{edgelayer}
		\draw [bend right=60] (2) to (3);
		\draw (4) to (2);
		\draw (4) to (3);
		\draw [bend left=60] (2) to (3);
	\end{pgfonlayer}
\end{tikzpicture}

    \mid r \in \Z \right\}.
\end{equation}
when \(p \equiv 3 \mod 4\). It is straightforward to see, by evaluating
\(\interp{-}\) on each element, that the sets in
equations~\eqref{eq:scalar_normal_form_p1} or \eqref{eq:scalar_normal_form_p3}
contain exactly one diagram for each scalar in \(\mathbb{G}_p \setminus \{0\}\)
(and the zero diagram \(\begin{tikzpicture}
	\begin{pgfonlayer}{nodelayer}
		\node [style=gn, label={[gphase]left:$1,0$}] (0) at (0, 0) {};
	\end{pgfonlayer}
\end{tikzpicture}
\) corresponds
to \(0 \in \mathbb{G}_p\)). Comparing with the definition of \(\mathbb{G}_p\) in
equations~\eqref{eq:Pp_1}-\eqref{eq:Pp_3}, the first factor in the tensor
product corresponds to powers of \(-1\) or powers of the imaginary unit \(i\);
the second factor corresponds to powers of \(\omega\); and the third factor
to powers of \(\sqrt{p}\).

\begin{theorem}
  \label{thm:scalar_completeness}
  \(\ZXeq\) proves any scalar diagram equal to a scalar from either
  equations~\eqref{eq:scalar_normal_form_p1} or \eqref{eq:scalar_normal_form_p3}
  (depending on the congruence of \(p\) modulo \(4\)), or to the zero scalar
  \(\begin{tikzpicture}
	\begin{pgfonlayer}{nodelayer}
		\node [style=gn, label={[gphase]left:$1,0$}] (0) at (0, 0) {};
	\end{pgfonlayer}
\end{tikzpicture}
\).
\end{theorem}

The completeness for the whole stabiliser fragment follows immediately:
\begin{theorem}
  The equational theory \(\ZXeq\) is complete for \(\Stab\), i.e. for any
  \(\ZXp\)-diagrams \(A\) and \(B\), if \(\interp{A} = \interp{B}\), then
  \(\ZXeq \vdash A = B\). Put more formally, there is a commutative diagram
  \begin{equation}
    \begin{tikzcd}
      \ZXp & {\ZXp / \ZXeq} \\
      & \Stab
      \arrow["{\interp{-}}"', from=1-1, to=2-2]
      \arrow[from=1-2, to=2-2]
      \arrow["p", twoheadrightarrow, from=1-1, to=1-2]
    \end{tikzcd}
  \end{equation}
  where \(p\) is the projection \(\ZXp \to \ZXp / \ZXeq\) and the vertical arrow
  is an isomorphism of categories.
\end{theorem}

\section{Mixed states and relations}
\label{sec:relations}
In this last section we use the work of \cite{DBLP:conf/icalp/CaretteJPV19} to extend our completeness result to the mixed-state case. We then unravel the connection to the Lagrangian relation investigated in \cite{comfort_graphical_2021}.

\subsection{A complete graphical language for $\textup{CPM}( \mathsf{Stab}_p )$ }

We now extend our completeness result form $ \mathsf{Stab}_p $ to $\textup{CPM}( \mathsf{Stab}_p )$, the category of completely positive maps corresponding to mixed state stabiliser quantum mechanics, see \cite{DBLP:journals/entcs/Selinger07a,coecke_picturing_2017} for a formal definition. We will rely on the discard construction of \cite{DBLP:conf/icalp/CaretteJPV19} to define a graphical language \({(\ZXp)}^{\ground}\). It consists in adding to the equational theory one generator, the discard $\ground :1\to 0$ and equations stating that this generator erases all isometries. In \(\mathsf{Stab}_p \), the isometries are generated by the following diagrams:

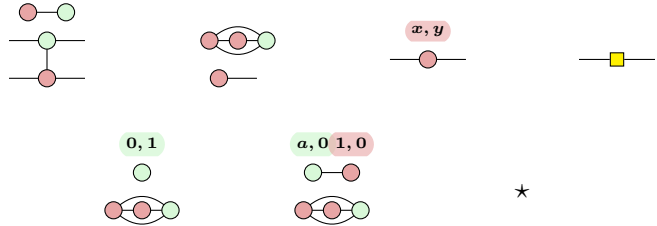
\begin{center}
	\input{figures/isometries.tikz}
\end{center}

The equations to add are then:

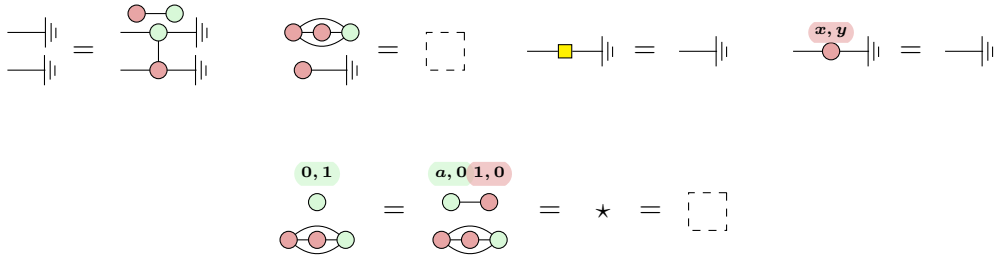
\begin{center}
	\input{figures/isometrieserase.tikz}
\end{center}
	
A new interpretation $\interp{\_}: {\ZXp}^{\ground} \to \textup{CPM}( \mathsf{Stab}_p )$ is defined as $\interp{\input{ground.tikz}}: \rho \mapsto \Tr(\rho)$ for the ground and for all \(\ZXp\)-diagram $D:n\to m$:
	
\begin{center}

$\interp{\input{figures/D.tikz} }^{\ground} : \rho \mapsto \interp{\input{figures/D.tikz} }^{\dagger} \rho \interp{\input{figures/D.tikz} }$.
 	
\end{center}

Corollary 22 of \cite{DBLP:conf/icalp/CaretteJPV19} provides a sufficient condition for the previous construction to extend to a universal and complete graphical language for $\mathsf{Stab}_p $ into a universal complete graphical language for $\textup{CPM}( \mathsf{Stab}_p )$. This condition is for $\mathsf{Stab}_p $ to have \emph{enough isometries}, meaning (we use here a little stronger condition than in \cite{DBLP:conf/icalp/CaretteJPV19}) that for all maps $f:A \to B\otimes X$ and $g:A \to B\otimes Y$ such that:

\begin{center}
	\input{relcpm.tikz}
\end{center}

there is an isometry $V: X\to Y$ in $\mathsf{Stab}_p $ such that:

\begin{center}
	\input{reliso.tikz}
\end{center}

To prove this, we will use arguments similar to the proof that $\mathsf{Stab}_2 $ has enough isometry. Everything relies on the following lemma:

\begin{lemma}\label{lem:bipartition}
	Given any bipartition of the outputs of a $\ZXp$-diagram $D:0\to n+m$. We can find unitaries $A$ and $B$ in $\mathsf{Stab}_p $ such that:

\begin{center}
	\input{figures/bipartiteform.tikz}
\end{center}

\end{lemma}

Using this we can prove:

\begin{lemma}
	$\mathsf{Stab}_p $ has enough isometries.
\end{lemma}

The proof is exactly the same as the qubit case, see the proof of Proposition 18 in \cite{DBLP:conf/icalp/CaretteJPV19}. It then follows directly from \cite{DBLP:conf/icalp/CaretteJPV19} that:

\begin{theorem}
	\({(\ZXp)}^{\ground}\) is universal and complete for $\textup{CPM}( \mathsf{Stab}_p )$.
\end{theorem}

\subsection{Co-isotropic relations}

It has been shown in \cite{comfort_graphical_2021, comfort_symplectic_2021,
  comfort_relational_nodate} that \(\textrm{CPM}(\Stab)\) is equivalent to the
category of affine co-isotropic relations up to scalars. More formally, we endow
\(\Z_p^2\) with the symplectic form
\begin{equation}
  \omega\left(
    \begin{bmatrix}
      a \\ b
    \end{bmatrix},
    \begin{bmatrix}
      c \\ d
    \end{bmatrix}
  \right) = ad - bc,
\end{equation}
and \(\Z_p^{2m} = \bigoplus_m \Z_p^2\) with the direct sum symplectic form.
\begin{definition}
  The symmetric monoidal category \(\AffCoIsoRel\) has as objects \(\N\), and as
  morphisms, relations \(R : \Z_p^{2m} \to \Z_p^{2n}\) such that \(R\) viewed as a
  subset of \(\Z_p^m \times \Z_p^n\) is an affine co-isotropic subspace thereof.
\end{definition}
Since \cite{comfort_graphical_2021} works in the scalarless ZX-calculus, we need
to add one extra axiom, which suffices to eliminate all remaining (non-zero)
scalars in \(\Stab\): we impose the rule \textsc{(Mod)} that \(p = 1\).
Diagrammatically, this amounts to quotienting \((\ZXp)^{\ground}\) by the
following rule:
\begin{equation}
  \input{figures/scalarless_axiom.tikz}
\end{equation}
Then we can give an interpretation \(\left[ - \right]\) of
\({(\ZXp)}^{\ground}\) making it universal and complete for \(\AffCoIsoRel\),
and which is defined uniquely by the commutative diagram
\begin{equation}
  \begin{tikzcd}
    {(\ZXp)^{\ground}} & \AffCoIsoRel \\
    {\mathrm{CPM}(\Stab)} & {\mathrm{CPM}(\Stab) / \textsc{(Mod)}}
    \arrow[dashed,"{\left[ - \right]}", from=1-1, to=1-2]
    \arrow["{\interp{-}}"', from=1-1, to=2-1]
    \arrow[twoheadrightarrow, from=2-1, to=2-2]
    \arrow[leftrightarrow, from=1-2, to=2-2]
  \end{tikzcd}
\end{equation}
Explicitly, it is given by the identity on objects, \([m] = m\), and is defined
on morphisms by: for \(x,y \in \Z_p\) and \(z \in \Z_p^*\)
\begin{align*}
  \left[ \input{figures/generators/r-spider_labelless.tikz} \right]
  &= \left\{ \left(
  \bigoplus_{k=1}^m \begin{bmatrix}
    a \\ b_k
  \end{bmatrix} ,
  \bigoplus_{k=1}^n \begin{bmatrix}
    -a \\ -c_k
  \end{bmatrix}
  \right) \,\middle|\, a,b_k,c_k \in \Z_p \qand \sum_k b_k = \sum_k c_k \right\} \\
  \left[ \input{figures/generators/g-spider_labelless.tikz} \right]
  &= \left\{ \left(
  \bigoplus_{k=1}^m \begin{bmatrix}
    a_k \\ c
  \end{bmatrix} ,
  \bigoplus_{k=1}^n \begin{bmatrix}
    b_k \\ c
  \end{bmatrix}
  \right) \,\middle|\, a,b_k,c_k \in \Z_p \qand \sum_k a_k = \sum_k b_k \right\} \\
  \left[ \begin{tikzpicture}
	\begin{pgfonlayer}{nodelayer}
		\node [style=rn, label={[rphase]below:$x,y$}] (0) at (-0.5, 0) {};
		\node [style=none] (1) at (0.5, 0) {};
	\end{pgfonlayer}
	\begin{pgfonlayer}{edgelayer}
		\draw (0) to (1.center);
	\end{pgfonlayer}
\end{tikzpicture}
 \right]
  &= \left\{ \left(
  \bullet,
    \begin{bmatrix}
      -1 & 0 \\
      -y & -1
    \end{bmatrix} \begin{bmatrix} v \\ 0 \end{bmatrix} + \begin{bmatrix} -x \\ 0 \end{bmatrix}
  \right) \,\middle|\, v \in \Z_p \right\} \\
  \left[ \begin{tikzpicture}
	\begin{pgfonlayer}{nodelayer}
		\node [style=gn, label={[gphase]below:$x,y$}] (0) at (-0.5, 0) {};
		\node [style=none] (1) at (0.5, 0) {};
	\end{pgfonlayer}
	\begin{pgfonlayer}{edgelayer}
		\draw (0) to (1.center);
	\end{pgfonlayer}
\end{tikzpicture}
 \right]
  &= \left\{ \left(
  \bullet,
    \begin{bmatrix}
      1 & -y \\
      0 & 1
    \end{bmatrix} \begin{bmatrix} 0 \\ v \end{bmatrix} + \begin{bmatrix} 0 \\ x \end{bmatrix}
  \right) \,\middle|\, v \in \Z_p \right\} \\
  \left[ \begin{tikzpicture}
	\begin{pgfonlayer}{nodelayer}
		\node [style=had] (0) at (0, 0) {};
		\node [style=none] (1) at (-1, 0) {};
		\node [style=none] (2) at (1, 0) {};
	\end{pgfonlayer}
	\begin{pgfonlayer}{edgelayer}
		\draw (0) to (1.center);
		\draw (0) to (2.center);
	\end{pgfonlayer}
\end{tikzpicture}
 \right]
  &= \left\{ \left(
  \vb{v},
    \begin{bmatrix}
      0 & -1 \\
      1 & 0
    \end{bmatrix} \vb{v}
  \right) \,\middle|\, \vb{v} \in \Z_p^2 \right\} \\
  \left[ \begin{tikzpicture}
	\begin{pgfonlayer}{nodelayer}
		\node [style=none] (0) at (-1, 0) {};
		\node [style=none] (1) at (1, 0) {};
		\node [style=rmat] (2) at (0, 0) {$z$};
	\end{pgfonlayer}
	\begin{pgfonlayer}{edgelayer}
		\draw (0.center) to (2);
		\draw (2) to (1.center);
	\end{pgfonlayer}
\end{tikzpicture}
 \right]
  &= \left\{ \left(
  \vb{v},
    \begin{bmatrix}
      z & 0 \\
      0 & z^{-1}
    \end{bmatrix} \vb{v}
  \right) \,\middle|\, \vb{v} \in \Z_p^2 \right\}
\end{align*}
Note that all of these are actually affine \emph{Lagrangian} relations. The only
generator which has a co-isotropic but not Lagrangian semantics is the discard
map:
\begin{equation}
  \left[ \input{ground.tikz} \right]
  = \left\{ \left( \vb{v}, \bullet \right) \,\middle|\, \vb{v} \in \Z_p^2 \right\}
\end{equation}

As pointed out in \cite{baez_props_2018, baez_compositional_2018,
  comfort_graphical_2021}, the related category of affine Lagrangian relations
over the field \(\R[x,y]/(xy-1)\) can be used to represent a fragment of
electrical circuits. We expect that the axiomatisation of
figure~\ref{fig:axioms} can be adapted to that setting, but leave this for a
future article.

\subsection{Geometrical interpretation}

The previous interpretation corresponds to a geometric intuition. All Clifford gates can be interpreted as affine transformations of the torus $\mathbb{Z}_p \times \mathbb{Z}_p $. This torus can be identified as a phase space, the position and momentum coordinates, $q$ and $p$ corresponding respectively to the first and second wire in the previous section semantics. In this section we will take as an example the case $p=5$.

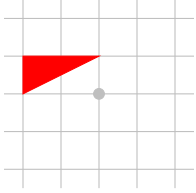
\begin{center}
	\input{figures/lattice0.tikz}
\end{center}

The red triangle here will allow us to illustrate the geometrical action of stab.

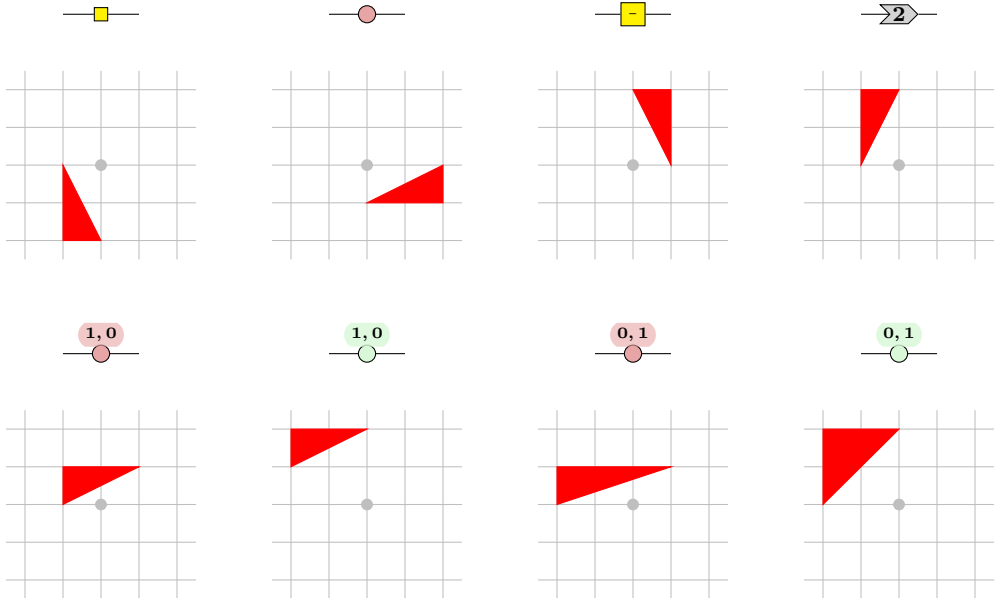
\begin{center}
	\input{figures/lattice.tikz}
\end{center}

The Pauli phase gates, acts as translations along the vertical or horizontal axis. The multiplication gates acts as a scaling. The Hadamard gate corresponds to a $\frac{\pi}{2}$-rotation and the antipode to a $\pi$-rotation. Finally, the pure Clifford gates corresponds to shears.


\section*{Conclusion}

We have constructed a ZX-calculus which captures the stabiliser fragment in odd
prime dimensions, whilst retaining many of the ``nice'' features of the qubit
ZX-calculus. Of course, there are a few obvious questions that we leave for
future work.

First amongst these is the question of whether a fully universal calculus can be
obtained from the ideas we used here. The spiders we have used here labelled
with elements \(a,b \in \Z_p \times \Z_p\) and which can be interpreted as
polynomials \(x \mapsto ax+bx^2\) which parametrise the phases of the spider.
Adding one additional term of degree \(3\) is already sufficient to obtain a
universal calculus in prime dimensions (strictly) greater than \(3\)
\cite{howard_qudit_2012}. In fact, one might as well add all higher order of
polynomials (mod \(p\)) since access to such higher degrees will hopefully
prove useful in finding commutation relations for the resulting spiders.

Secondly, it remains to be seen how to formulate a universal ZX-calculus for
non-prime dimensions, even for just the stabiliser fragment. For this, the
methods in this article are clearly inadequate: for example local scaling is no
longer an invertible operation and thus certainly not in the Clifford group.

Finally, the set of axioms we provide here is probably not minimal. It would be
nice to see if a simplified version can be obtained, as was done in
\cite{backens_simplified_2017} for the qubit case.


{\raggedright\printbibliography}

@unpublished{abramsky_categorical_2008,
  title = {Categorical Quantum Mechanics},
  author = {Abramsky, Samson and Coecke, Bob},
  date = {2008-08-07},
  eprint = {0808.1023},
  eprinttype = {arxiv},
  primaryclass = {quant-ph},
  url = {http://arxiv.org/abs/0808.1023},
  archiveprefix = {arXiv}
}

@article{DBLP:journals/entcs/Selinger07a,
	author    = {Peter Selinger},
	title     = {Dagger Compact Closed Categories and Completely Positive Maps: (Extended
	Abstract)},
	journal   = {Electron. Notes Theor. Comput. Sci.},
	volume    = {170},
	pages     = {139--163},
	year      = {2007},
	url       = {https://doi.org/10.1016/j.entcs.2006.12.018},
	doi       = {10.1016/j.entcs.2006.12.018},
	bibsource = {dblp computer science bibliography, https://dblp.org}
}

@unpublished{appleby_properties_2009,
  title = {Properties of the Extended {{Clifford}} Group with Applications to {{SIC-POVMs}} and {{MUBs}}},
  author = {Appleby, D. M.},
  date = {2009-09-28},
  eprint = {0909.5233},
  eprinttype = {arxiv},
  primaryclass = {quant-ph},
  url = {http://arxiv.org/abs/0909.5233},
  archiveprefix = {arXiv}
}

@article{backens_there_2021,
  title = {There and Back Again: {{A}} Circuit Extraction Tale},
  shorttitle = {There and Back Again},
  author = {Backens, Miriam and Miller-Bakewell, Hector and de Felice, Giovanni and Lobski, Leo and van de Wetering, John},
  date = {2021-03-25},
  journaltitle = {Quantum},
  volume = {5},
  pages = {421},
  publisher = {{Verein zur Förderung des Open Access Publizierens in den Quantenwissenschaften}},
  doi = {10.22331/q-2021-03-25-421},
  url = {https://quantum-journal.org/papers/q-2021-03-25-421/}
}

@article{backens_zx-calculus_2014,
  title = {The {{ZX-calculus}} Is Complete for Stabilizer Quantum Mechanics},
  author = {Backens, Miriam},
  date = {2014-09-17},
  journaltitle = {New Journal of Physics},
  shortjournal = {New J. Phys.},
  volume = {16},
  number = {9},
  eprint = {1307.7025},
  eprinttype = {arxiv},
  pages = {093021},
  issn = {1367-2630},
  doi = {10.1088/1367-2630/16/9/093021},
  url = {http://arxiv.org/abs/1307.7025},
  archiveprefix = {arXiv}
}

@unpublished{bahramgiri_efficient_2007,
  title = {An {{Efficient Algorithm}} to {{Recognize Locally Equivalent Graphs}} in {{Non-Binary Case}}},
  author = {Bahramgiri, Mohsen and Beigi, Salman},
  date = {2007-07-01},
  eprint = {cs/0702057},
  eprinttype = {arxiv},
  url = {http://arxiv.org/abs/cs/0702057},
  archiveprefix = {arXiv}
}

@article{bonchi_interacting_2017,
  title = {Interacting {{Hopf Algebras}}},
  author = {Bonchi, Filippo and Sobocinski, Pawel and Zanasi, Fabio},
  date = {2017-01},
  journaltitle = {Journal of Pure and Applied Algebra},
  shortjournal = {Journal of Pure and Applied Algebra},
  volume = {221},
  number = {1},
  eprint = {1403.7048},
  eprinttype = {arxiv},
  pages = {144--184},
  issn = {00224049},
  doi = {10.1016/j.jpaa.2016.06.002},
  url = {http://arxiv.org/abs/1403.7048},
  archiveprefix = {arXiv}
}

@inproceedings{DBLP:conf/icalp/CaretteJPV19,
	author    = {Titouan Carette and
	Emmanuel Jeandel and
	Simon Perdrix and
	Renaud Vilmart},
	editor    = {Christel Baier and
	Ioannis Chatzigiannakis and
	Paola Flocchini and
	Stefano Leonardi},
	title     = {Completeness of Graphical Languages for Mixed States Quantum Mechanics},
	booktitle = {46th International Colloquium on Automata, Languages, and Programming,
	{ICALP} 2019, July 9-12, 2019, Patras, Greece},
	series    = {LIPIcs},
	volume    = {132},
	pages     = {108:1--108:15},
	publisher = {Schloss Dagstuhl - Leibniz-Zentrum f{\"{u}}r Informatik},
	year      = {2019},
	url       = {https://doi.org/10.4230/LIPIcs.ICALP.2019.108},
	doi       = {10.4230/LIPIcs.ICALP.2019.108},
	bibsource = {dblp computer science bibliography, https://dblp.org}
}

@unpublished{carette_colored_2020,
  title = {Colored Props for Large Scale Graphical Reasoning},
  author = {Carette, Titouan and Perdrix, Simon},
  date = {2020-07-07},
  eprint = {2007.03564},
  eprinttype = {arxiv},
  primaryclass = {quant-ph},
  url = {http://arxiv.org/abs/2007.03564},
  archiveprefix = {arXiv}
}

@unpublished{carette_szx-calculus:_2019,
  title = {{{SZX-calculus}}: {{Scalable Graphical Quantum Reasoning}}},
  shorttitle = {{{SZX-calculus}}},
  author = {Carette, Titouan and Horsman, Dominic and Perdrix, Simon},
  date = {2019-04-30},
  eprint = {1905.00041},
  eprinttype = {arxiv},
  primaryclass = {quant-ph},
  url = {http://arxiv.org/abs/1905.00041},
  archiveprefix = {arXiv}
}

@unpublished{carette_when_2021,
  title = {When {{Only Topology Matters}}},
  author = {Carette, Titouan},
  date = {2021-02-04},
  eprint = {2102.03178},
  eprinttype = {arxiv},
  primaryclass = {quant-ph},
  url = {http://arxiv.org/abs/2102.03178},
  archiveprefix = {arXiv},
  version = {1}
}

@unpublished{chancellor_graphical_2018,
  title = {Graphical {{Structures}} for {{Design}} and {{Verification}} of {{Quantum Error Correction}}},
  author = {Chancellor, Nicholas and Kissinger, Aleks and Roffe, Joschka and Zohren, Stefan and Horsman, Dominic},
  date = {2018-01-12},
  eprint = {1611.08012},
  eprinttype = {arxiv},
  primaryclass = {quant-ph},
  url = {http://arxiv.org/abs/1611.08012},
  archiveprefix = {arXiv}
}

@article{clark_valence_2006,
  title = {Valence Bond Solid Formalism for D-Level One-Way Quantum Computation},
  author = {Clark, Sean},
  date = {2006-03-17},
  journaltitle = {Journal of Physics A: Mathematical and General},
  shortjournal = {J. Phys. A: Math. Gen.},
  volume = {39},
  number = {11},
  eprint = {quant-ph/0512155},
  eprinttype = {arxiv},
  pages = {2701--2721},
  issn = {0305-4470, 1361-6447},
  doi = {10.1088/0305-4470/39/11/010},
  url = {http://arxiv.org/abs/quant-ph/0512155},
  archiveprefix = {arXiv}
}

@article{coecke_interacting_2011,
  title = {Interacting {{Quantum Observables}}: {{Categorical Algebra}} and {{Diagrammatics}}},
  shorttitle = {Interacting {{Quantum Observables}}},
  author = {Coecke, Bob and Duncan, Ross},
  date = {2011-04-14},
  journaltitle = {New Journal of Physics},
  shortjournal = {New J. Phys.},
  volume = {13},
  number = {4},
  eprint = {0906.4725},
  eprinttype = {arxiv},
  pages = {043016},
  issn = {1367-2630},
  doi = {10.1088/1367-2630/13/4/043016},
  url = {http://arxiv.org/abs/0906.4725},
  archiveprefix = {arXiv}
}

@book{coecke_picturing_2017,
  title = {Picturing {{Quantum Processes}}: {{A First Course}} in {{Quantum Theory}} and {{Diagrammatic Reasoning}}},
  shorttitle = {Picturing {{Quantum Processes}}},
  author = {Coecke, Bob and Kissinger, Aleks},
  date = {2017},
  publisher = {{Cambridge University Press}},
  location = {{Cambridge}},
  doi = {10.1017/9781316219317},
  url = {https://www.cambridge.org/core/books/picturing-quantum-processes/1119568B3101F3A685BE832FEEC53E52},
  isbn = {978-1-107-10422-8}
}

@unpublished{comfort_graphical_2021,
  title = {A {{Graphical Calculus}} for {{Lagrangian Relations}}},
  author = {Comfort, Cole and Kissinger, Aleks},
  date = {2021-05-13},
  eprint = {2105.06244},
  eprinttype = {arxiv},
  primaryclass = {math-ph, physics:quant-ph},
  url = {http://arxiv.org/abs/2105.06244},
  archiveprefix = {arXiv}
}

@unpublished{de_beaudrap_fast_2020,
  title = {Fast and Effective Techniques for {{T-count}} Reduction via Spider Nest Identities},
  author = {de Beaudrap, Niel and Bian, Xiaoning and Wang, Quanlong},
  options = {useprefix=true},
  date = {2020-04-10},
  eprint = {2004.05164},
  eprinttype = {arxiv},
  primaryclass = {quant-ph},
  url = {http://arxiv.org/abs/2004.05164},
  archiveprefix = {arXiv}
}

@unpublished{de_beaudrap_linearized_2012,
  title = {A Linearized Stabilizer Formalism for Systems of Finite Dimension},
  author = {de Beaudrap, Niel},
  options = {useprefix=true},
  date = {2012-09-10},
  eprint = {1102.3354},
  eprinttype = {arxiv},
  primaryclass = {quant-ph},
  url = {http://arxiv.org/abs/1102.3354},
  archiveprefix = {arXiv}
}

@unpublished{de_beaudrap_zx_2017,
  title = {The {{ZX}} Calculus Is a Language for Surface Code Lattice Surgery},
  author = {de Beaudrap, Niel and Horsman, Dominic},
  options = {useprefix=true},
  date = {2017-04-27},
  eprint = {1704.08670},
  eprinttype = {arxiv},
  primaryclass = {quant-ph},
  url = {http://arxiv.org/abs/1704.08670},
  archiveprefix = {arXiv}
}

@article{duncan_graph-theoretic_2020,
  title = {Graph-Theoretic {{Simplification}} of {{Quantum Circuits}} with the {{ZX-calculus}}},
  author = {Duncan, Ross and Kissinger, Aleks and Perdrix, Simon and van de Wetering, John},
  options = {useprefix=true},
  date = {2020-06-04},
  journaltitle = {Quantum},
  shortjournal = {Quantum},
  volume = {4},
  eprint = {1902.03178},
  eprinttype = {arxiv},
  pages = {279},
  issn = {2521-327X},
  doi = {10.22331/q-2020-06-04-279},
  url = {http://arxiv.org/abs/1902.03178},
  archiveprefix = {arXiv}
}

@article{duncan_verifying_2014,
  title = {Verifying the {{Steane}} Code with {{Quantomatic}}},
  author = {Duncan, Ross and Lucas, Maxime},
  date = {2014-12-27},
  journaltitle = {Electronic Proceedings in Theoretical Computer Science},
  shortjournal = {Electron. Proc. Theor. Comput. Sci.},
  volume = {171},
  eprint = {1306.4532},
  eprinttype = {arxiv},
  pages = {33--49},
  issn = {2075-2180},
  doi = {10.4204/EPTCS.171.4},
  url = {http://arxiv.org/abs/1306.4532},
  archiveprefix = {arXiv},
  version = {2}
}

@article{garvie_verifying_2018,
  title = {Verifying the {{Smallest Interesting Colour Code}} with {{Quantomatic}}},
  author = {Garvie, Liam and Duncan, Ross},
  date = {2018-02-27},
  journaltitle = {Electronic Proceedings in Theoretical Computer Science},
  shortjournal = {Electron. Proc. Theor. Comput. Sci.},
  volume = {266},
  eprint = {1706.02717},
  eprinttype = {arxiv},
  pages = {147--163},
  issn = {2075-2180},
  doi = {10.4204/EPTCS.266.10},
  url = {http://arxiv.org/abs/1706.02717},
  archiveprefix = {arXiv},
  version = {2}
}

@article{gottesman_fault-tolerant_1999,
  title = {Fault-{{Tolerant Quantum Computation}} with {{Higher-Dimensional Systems}}},
  author = {Gottesman, Daniel},
  date = {1999-09},
  journaltitle = {Chaos, Solitons \& Fractals},
  shortjournal = {Chaos, Solitons \& Fractals},
  volume = {10},
  number = {10},
  eprint = {quant-ph/9802007},
  eprinttype = {arxiv},
  pages = {1749--1758},
  issn = {09600779},
  doi = {10.1016/S0960-0779(98)00218-5},
  url = {http://arxiv.org/abs/quant-ph/9802007},
  archiveprefix = {arXiv}
}

@article{howard_qudit_2012,
  title = {Qudit Versions of the Qubit "Pi-over-Eight" Gate},
  author = {Howard, Mark and Vala, Jiri},
  date = {2012-08-15},
  journaltitle = {Physical Review A},
  shortjournal = {Phys. Rev. A},
  volume = {86},
  number = {2},
  eprint = {1206.1598},
  eprinttype = {arxiv},
  pages = {022316},
  issn = {1050-2947, 1094-1622},
  doi = {10.1103/PhysRevA.86.022316},
  url = {http://arxiv.org/abs/1206.1598},
  archiveprefix = {arXiv}
}

@incollection{hutchison_rewriting_2010,
  ids = {hutchison_rewriting_2010-1},
  title = {Rewriting {{Measurement-Based Quantum Computations}} with {{Generalised Flow}}},
  booktitle = {Automata, {{Languages}} and {{Programming}}},
  author = {Duncan, Ross and Perdrix, Simon},
  editor = {Abramsky, Samson and Gavoille, Cyril and Kirchner, Claude and Meyer auf der Heide, Friedhelm and Spirakis, Paul G.},
  date = {2010},
  volume = {6199},
  pages = {285--296},
  publisher = {{Springer Berlin Heidelberg}},
  location = {{Berlin, Heidelberg}},
  doi = {10.1007/978-3-642-14162-1_24},
  url = {http://link.springer.com/10.1007/978-3-642-14162-1_24},
  editorb = {Hutchison, David and Kanade, Takeo and Kittler, Josef and Kleinberg, Jon M. and Mattern, Friedemann and Mitchell, John C. and Naor, Moni and Nierstrasz, Oscar and Pandu Rangan, C. and Steffen, Bernhard and Sudan, Madhu and Terzopoulos, Demetri and Tygar, Doug and Vardi, Moshe Y. and Weikum, Gerhard},
  editorbtype = {redactor},
  isbn = {978-3-642-14161-4 978-3-642-14162-1}
}

@unpublished{jeandel_complete_2017,
  title = {A {{Complete Axiomatisation}} of the {{ZX-Calculus}} for {{Clifford}}+{{T Quantum Mechanics}}},
  author = {Jeandel, Emmanuel and Perdrix, Simon and Vilmart, Renaud},
  date = {2017-05-31},
  eprint = {1705.11151},
  eprinttype = {arxiv},
  primaryclass = {quant-ph},
  url = {http://arxiv.org/abs/1705.11151},
  archiveprefix = {arXiv}
}

@unpublished{kissinger_classical_2022,
  title = {Classical Simulation of Quantum Circuits with Partial and Graphical Stabiliser Decompositions},
  author = {Kissinger, Aleks and van de Wetering, John and Vilmart, Renaud},
  options = {useprefix=true},
  date = {2022-02-18},
  eprint = {2202.09202},
  eprinttype = {arxiv},
  primaryclass = {quant-ph},
  url = {http://arxiv.org/abs/2202.09202},
  archiveprefix = {arXiv}
}

@unpublished{kissinger_reducing_2020,
  title = {Reducing {{T-count}} with the {{ZX-calculus}}},
  author = {Kissinger, Aleks and van de Wetering, John},
  options = {useprefix=true},
  date = {2020-01-17},
  eprint = {1903.10477},
  eprinttype = {arxiv},
  primaryclass = {quant-ph},
  url = {http://arxiv.org/abs/1903.10477},
  archiveprefix = {arXiv}
}

@article{nenhauser_explicit_2002,
  title = {An {{Explicit Construction}} of the {{Metaplectic Representation}} over a {{Finite Field}}},
  author = {Nenhauser, Markus},
  date = {2002},
  journaltitle = {Journal of Lie Theory},
  volume = {12},
  number = {15}
}

@unpublished{ng_universal_2017,
  title = {A Universal Completion of the {{ZX-calculus}}},
  author = {Ng, Kang Feng and Wang, Quanlong},
  date = {2017-06-29},
  eprint = {1706.09877},
  eprinttype = {arxiv},
  primaryclass = {quant-ph},
  url = {http://arxiv.org/abs/1706.09877},
  archiveprefix = {arXiv}
}

@unpublished{townsend-teague_classifying_2021,
  title = {Classifying {{Complexity}} with the {{ZX-Calculus}}: {{Jones Polynomials}} and {{Potts Partition Functions}}},
  shorttitle = {Classifying {{Complexity}} with the {{ZX-Calculus}}},
  author = {Townsend-Teague, Alex and Meichanetzidis, Konstantinos},
  date = {2021-03-11},
  eprint = {2103.06914},
  eprinttype = {arxiv},
  primaryclass = {quant-ph},
  url = {http://arxiv.org/abs/2103.06914},
  archiveprefix = {arXiv}
}

@unpublished{vilmart_near-optimal_2018,
  title = {A {{Near-Optimal Axiomatisation}} of {{ZX-Calculus}} for {{Pure Qubit Quantum Mechanics}}},
  author = {Vilmart, Renaud},
  date = {2018-12-21},
  eprint = {1812.09114},
  eprinttype = {arxiv},
  primaryclass = {quant-ph},
  url = {http://arxiv.org/abs/1812.09114},
  archiveprefix = {arXiv}
}

@unpublished{wang_non-anyonic_2021,
  title = {A Non-Anyonic Qudit {{ZW-calculus}}},
  author = {Wang, Quanlong},
  date = {2021-10-12},
  eprint = {2109.11285},
  eprinttype = {arxiv},
  primaryclass = {quant-ph},
  url = {http://arxiv.org/abs/2109.11285},
  archiveprefix = {arXiv}
}

@article{wang_qudits_2020,
  title = {Qudits and {{High-Dimensional Quantum Computing}}},
  author = {Wang, Yuchen and Hu, Zixuan and Sanders, Barry C. and Kais, Sabre},
  date = {2020},
  journaltitle = {Frontiers in Physics},
  shortjournal = {Front. Phys.},
  volume = {8},
  publisher = {{Frontiers}},
  issn = {2296-424X},
  doi = {10.3389/fphy.2020.589504},
  url = {https://www.frontiersin.org/articles/10.3389/fphy.2020.589504/full}
}

@unpublished{wang_qufinite_2021,
  title = {Qufinite {{ZX-calculus}}: A Unified Framework of Qudit {{ZX-calculi}}},
  shorttitle = {Qufinite {{ZX-calculus}}},
  author = {Wang, Quanlong},
  date = {2021-05-24},
  eprint = {2104.06429},
  eprinttype = {arxiv},
  primaryclass = {quant-ph},
  url = {http://arxiv.org/abs/2104.06429},
  archiveprefix = {arXiv}
}

@article{wang_qutrit_2014,
  title = {Qutrit {{Dichromatic Calculus}} and {{Its Universality}}},
  author = {Wang, Quanlong and Bian, Xiaoning},
  date = {2014-12-28},
  journaltitle = {Electronic Proceedings in Theoretical Computer Science},
  shortjournal = {Electron. Proc. Theor. Comput. Sci.},
  volume = {172},
  pages = {92--101},
  issn = {2075-2180},
  doi = {10.4204/EPTCS.172.7},
  url = {http://arxiv.org/abs/1406.3056v3}
}

@article{wang_qutrit_2018,
  title = {Qutrit {{ZX-calculus}} Is {{Complete}} for {{Stabilizer Quantum Mechanics}}},
  author = {Wang, Quanlong},
  date = {2018-02-27},
  journaltitle = {Electronic Proceedings in Theoretical Computer Science},
  shortjournal = {Electron. Proc. Theor. Comput. Sci.},
  volume = {266},
  eprint = {1803.00696},
  eprinttype = {arxiv},
  pages = {58--70},
  issn = {2075-2180},
  doi = {10.4204/EPTCS.266.3},
  url = {http://arxiv.org/abs/1803.00696},
  archiveprefix = {arXiv}
}

@thesis{zanasi_interacting_2018,
  type = {phdthesis},
  title = {Interacting {{Hopf Algebras}}: The Theory of Linear Systems},
  shorttitle = {Interacting {{Hopf Algebras}}},
  author = {Zanasi, Fabio},
  date = {2018-05-04},
  eprint = {1805.03032},
  eprinttype = {arxiv},
  institution = {{Ecole Normale Superieure de Lyon}},
  url = {http://arxiv.org/abs/1805.03032},
  archiveprefix = {arXiv}
}

@article{zhou_quantum_2003,
  title = {Quantum Computation Based on D-Level Cluster States},
  author = {Zhou, D. L. and Zeng, B. and Xu, Z. and Sun, C. P.},
  date = {2003-12-02},
  journaltitle = {Physical Review A},
  shortjournal = {Phys. Rev. A},
  volume = {68},
  number = {6},
  eprint = {quant-ph/0304054},
  eprinttype = {arxiv},
  pages = {062303},
  issn = {1050-2947, 1094-1622},
  doi = {10.1103/PhysRevA.68.062303},
  url = {http://arxiv.org/abs/quant-ph/0304054},
  archiveprefix = {arXiv}
}

@thesis{comfort_relational_nodate,
  type = {DPhil},
  title = {Relational Semantics for Quantum Algorithms},
  author = {Comfort, Cole},
  institution = {{University of Oxford}},
  location = {{Oxford}},
  date = {To Appear}}

@unpublished{comfort_symplectic_2021,
  title = {A Symplectic Setting for Mixed Stabiliser Circuits},
  author = {Comfort, Cole},
  date = {2021-12-09},
  eventtitle = {Weekly Online {{ZX}} Meeting},
  venue = {{Online}}
}

@article{baez_compositional_2018,
  title = {A {{Compositional Framework}} for {{Passive Linear Networks}}},
  author = {Baez, John C. and Fong, Brendan},
  date = {2018-11-16},
  journaltitle = {Theory and Application of Categories},
  volume = {33},
  number = {38},
  eprint = {1504.05625},
  eprinttype = {arxiv},
  pages = {1158--1222},
  url = {http://arxiv.org/abs/1504.05625},
  archiveprefix = {arXiv}
}

@article{baez_props_2018,
  title = {Props in {{Network Theory}}},
  author = {Baez, John C. and Coya, Brandon and Rebro, Franciscus},
  date = {2018-06-01},
  journaltitle = {Theory and Application of Categories},
  volume = {33},
  number = {25},
  eprint = {1707.08321},
  eprinttype = {arxiv},
  pages = {727--783},
  url = {http://arxiv.org/abs/1707.08321},
  archiveprefix = {arXiv}
}

@unpublished{van_de_wetering_zx-calculus_2020,
  title = {{{ZX-calculus}} for the Working Quantum Computer Scientist},
  author = {van de Wetering, John},
  options = {useprefix=true},
  date = {2020-12-27},
  eprint = {2012.13966},
  eprinttype = {arxiv},
  primaryclass = {quant-ph},
  url = {http://arxiv.org/abs/2012.13966},
  archiveprefix = {arXiv}
}

@article{ranchin_depicting_2014,
  title = {Depicting Qudit Quantum Mechanics and Mutually Unbiased Qudit Theories},
  author = {Ranchin, André},
  date = {2014-12-28},
  journaltitle = {Electronic Proceedings in Theoretical Computer Science},
  shortjournal = {Electron. Proc. Theor. Comput. Sci.},
  volume = {172},
  pages = {68--91},
  issn = {2075-2180},
  doi = {10.4204/EPTCS.172.6},
  url = {http://arxiv.org/abs/1404.1288v3}
}

@article{gheorghiu_standard_2014,
  title = {Standard Form of Qudit Stabilizer Groups},
  author = {Gheorghiu, Vlad},
  date = {2014},
  journaltitle = {Physics Letters A},
  pages = {5}
}

@unpublished{hausmann_consolidating_2021,
  title = {A Consolidating Review of {{Spekkens}}' Toy Theory},
  author = {Hausmann, Ladina and Nurgalieva, Nuriya and del Rio, Lídia},
  options = {useprefix=true},
  date = {2021-05-07},
  eprint = {2105.03277},
  eprinttype = {arxiv},
  primaryclass = {quant-ph},
  url = {http://arxiv.org/abs/2105.03277},
  archiveprefix = {arXiv}
}

@article{backens_making_2015,
  title = {Making the Stabilizer {{ZX-calculus}} Complete for Scalars},
  author = {Backens, Miriam},
  date = {2015-11-04},
  journaltitle = {Electronic Proceedings in Theoretical Computer Science},
  shortjournal = {Electron. Proc. Theor. Comput. Sci.},
  volume = {195},
  eprint = {1507.03854},
  eprinttype = {arxiv},
  pages = {17--32},
  issn = {2075-2180},
  doi = {10.4204/EPTCS.195.2},
  url = {http://arxiv.org/abs/1507.03854},
  archiveprefix = {arXiv},
  version = {2}
}

@article{backens_simplified_2017,
  title = {A {{Simplified Stabilizer ZX-calculus}}},
  author = {Backens, Miriam and Perdrix, Simon and Wang, Quanlong},
  date = {2017-01-01},
  journaltitle = {Electronic Proceedings in Theoretical Computer Science},
  shortjournal = {Electron. Proc. Theor. Comput. Sci.},
  volume = {236},
  eprint = {1602.04744},
  eprinttype = {arxiv},
  pages = {1--20},
  issn = {2075-2180},
  doi = {10.4204/EPTCS.236.1},
  url = {http://arxiv.org/abs/1602.04744},
  archiveprefix = {arXiv}
}

@unpublished{backens_completeness_2021,
  title = {Completeness of the {{ZH-calculus}}},
  author = {Backens, Miriam and Kissinger, Aleks and Miller-Bakewell, Hector and van de Wetering, John and Wolffs, Sal},
  options = {useprefix=true},
  date = {2021-03-11},
  eprint = {2103.06610},
  eprinttype = {arxiv},
  primaryclass = {quant-ph},
  url = {http://arxiv.org/abs/2103.06610},
  archiveprefix = {arXiv}
}

@article{backens_zh_2019,
  title = {{{ZH}}: {{A Complete Graphical Calculus}} for {{Quantum Computations Involving Classical Non-linearity}}},
  shorttitle = {{{ZH}}},
  author = {Backens, Miriam and Kissinger, Aleks},
  date = {2019-01-31},
  journaltitle = {Electronic Proceedings in Theoretical Computer Science},
  shortjournal = {Electron. Proc. Theor. Comput. Sci.},
  volume = {287},
  eprint = {1805.02175},
  eprinttype = {arxiv},
  pages = {23--42},
  issn = {2075-2180},
  doi = {10.4204/EPTCS.287.2},
  url = {http://arxiv.org/abs/1805.02175},
  archiveprefix = {arXiv}
}

@unpublished{bonchi_string_2021,
  title = {String {{Diagram Rewrite Theory II}}: {{Rewriting}} with {{Symmetric Monoidal Structure}}},
  shorttitle = {String {{Diagram Rewrite Theory II}}},
  author = {Bonchi, Filippo and Gadducci, Fabio and Kissinger, Aleks and Sobocinski, Pawel and Zanasi, Fabio},
  date = {2021-04-29},
  number = {arXiv:2104.14686},
  eprint = {2104.14686},
  eprinttype = {arxiv},
  primaryclass = {cs, math},
  publisher = {{arXiv}},
  doi = {10.48550/arXiv.2104.14686},
  url = {http://arxiv.org/abs/2104.14686},
  archiveprefix = {arXiv}
}

@article{bonchi_string_2022,
  title = {String {{Diagram Rewrite Theory I}}: {{Rewriting}} with {{Frobenius Structure}}},
  shorttitle = {String {{Diagram Rewrite Theory I}}},
  author = {Bonchi, Filippo and Gadducci, Fabio and Kissinger, Aleks and Sobocinski, Paweł and Zanasi, Fabio},
  date = {2022-03-10},
  journaltitle = {Journal of the ACM (JACM)},
  publisher = {{ACM}},
  doi = {10.1145/3502719},
  url = {https://dl.acm.org/doi/full/10.1145/3502719}
}

@unpublished{bonchi_string_2022-1,
  title = {String {{Diagram Rewrite Theory III}}: {{Confluence}} with and without {{Frobenius}}},
  shorttitle = {String {{Diagram Rewrite Theory III}}},
  author = {Bonchi, Filippo and Gadducci, Fabio and Kissinger, Aleks and Sobociński, Paweł and Zanasi, Fabio},
  date = {2022-04-18},
  number = {arXiv:2109.06049},
  eprint = {2109.06049},
  eprinttype = {arxiv},
  primaryclass = {cs},
  publisher = {{arXiv}},
  doi = {10.48550/arXiv.2109.06049},
  url = {http://arxiv.org/abs/2109.06049},
  archiveprefix = {arXiv}
}

@article{bian_graphical_2015,
  title = {Graphical {{Calculus}} for {{Qutrit Systems}}},
  author = {Bian, Xiaoning and Wang, Quanlong},
  date = {2015},
  journaltitle = {IEICE Transactions on Fundamentals of Electronics, Communications and Computer Sciences},
  volume = {E98.A},
  number = {1},
  pages = {391--399},
  doi = {10.1587/transfun.E98.A.391}
}

@thesis{carette_wielding_2021,
  type = {Theses},
  title = {Wielding the {{ZX-calculus}}, {{Flexsymmetry}}, {{Mixed States}}, and {{Scalable Notations}}},
  author = {Carette, Titouan},
  date = {2021-11},
  institution = {{Université de Lorraine}},
  url = {https://hal.archives-ouvertes.fr/tel-03468027}
}

\appendix
\section{Some lemmas from elementary number theory}
\label{app:number_theory}

\begin{proposition}[First supplement to Quadratic Reciprocity]
  \label{prop:first_supplement}
  \(-1\) is a square mod \(p\) if and only if \(p\) is congruent to \(1 \mod
  4\).
\end{proposition}

\begin{lemma}
  For any odd prime \(p\) and \(x,y \in \Z_p\), if neither \(x\) nor \(y\) is
  square, then \(xy\) is a square.
\end{lemma}
\begin{proof}
  The map \(x \mapsto x^2\) is an endomorphism of \(\Z_p\) with kernel \(\{-1,
  1\}\). In other words, the image of this map, the subgroup \(Q_p\) of squares
  in \(\Z_p^*\), must have index \(2\), and then \(\Z_p^* / Q_p\) is a two
  element group. If neither \(x\) nor \(y\) is square, \(xQ_p = yQ_p\) whence
  \((xy)Q_p = (xQ_p)(yQ_p) = (xQ_p)^2 = Q_p\), so that \(xy \in Q_p\).
\end{proof}

\begin{corollary}
  \label{cor:at_least_one_square}
  For any prime \(p\) and \(x \in \Z_p\), at least one of \(-1\), \(x\) or
  \(-x\) is a square.
\end{corollary}

\section{Proof of universality}
\label{app:universality}

\begin{lproof}[of theorem~\ref{thm:universality}]
	By proposition~\ref{proposition:Cn_generators}, in odd prime dimensions the
	generalised Clifford groups are generated by
	\begin{equation}
    \label{eq:C1_generators}
		\interp{\begin{tikzpicture}
	\begin{pgfonlayer}{nodelayer}
		\node [style=gn, label={[gphase]below:$1,1$}] (0) at (0, -0.25) {};
		\node [style=none] (1) at (-1, -0.25) {};
		\node [style=none] (2) at (1, -0.25) {};
	\end{pgfonlayer}
	\begin{pgfonlayer}{edgelayer}
		\draw (0) to (1.center);
		\draw (0) to (2.center);
	\end{pgfonlayer}
\end{tikzpicture}
} = S\qc
		\interp{\begin{tikzpicture}
	\begin{pgfonlayer}{nodelayer}
		\node [style=had] (0) at (0, 0) {};
		\node [style=none] (1) at (-1, 0) {};
		\node [style=none] (2) at (1, 0) {};
	\end{pgfonlayer}
	\begin{pgfonlayer}{edgelayer}
		\draw (0) to (1.center);
		\draw (0) to (2.center);
	\end{pgfonlayer}
\end{tikzpicture}
} = H \qand
		\interp{\begin{tikzpicture}
	\begin{pgfonlayer}{nodelayer}
		\node [style=gn] (0) at (0, 0.325) {};
		\node [style=none] (1) at (-1, 0.325) {};
		\node [style=none] (2) at (1, 0.325) {};
		\node [style=gn] (3) at (0, -0.825) {};
		\node [style=none] (4) at (-1, -0.825) {};
		\node [style=none] (5) at (1, -0.825) {};
		\node [style=had] (6) at (0, -0.25) {};
		\node [style=rn] (7) at (-0.5, 1) {};
		\node [style=gn] (8) at (0.5, 1) {};
	\end{pgfonlayer}
	\begin{pgfonlayer}{edgelayer}
		\draw (0) to (1.center);
		\draw (0) to (2.center);
		\draw (3) to (4.center);
		\draw (3) to (5.center);
		\draw (0) to (6);
		\draw (6) to (3);
		\draw (7) to (8);
	\end{pgfonlayer}
\end{tikzpicture}
} = E.
	\end{equation}
	Since computational basis states and effects are easily represented as:
	\begin{equation}
    \label{eq:stabiliser_states}
		\interp{\begin{tikzpicture}
	\begin{pgfonlayer}{nodelayer}
		\node [style=rn, label={[rphase]below:$2x,0$}] (1) at (-0.875, 0) {};
		\node [style=rn] (2) at (0.125, 0) {};
		\node [style=none] (3) at (1.25, 0) {};
		\node [style=gn] (4) at (-0.75, -0.75) {};
		\node [style=rn] (5) at (0.75, -0.75) {};
		\node [style=rn] (6) at (0, -0.75) {};
	\end{pgfonlayer}
	\begin{pgfonlayer}{edgelayer}
		\draw (2) to (1);
		\draw (3.center) to (2);
		\draw [bend right=60] (4) to (5);
		\draw (6) to (4);
		\draw (6) to (5);
		\draw [bend left=60] (4) to (5);
	\end{pgfonlayer}
\end{tikzpicture}
} = \ket{x} \qand \interp{\begin{tikzpicture}
	\begin{pgfonlayer}{nodelayer}
		\node [style=rn, label={[rphase]below:$2x,0$}] (1) at (0.875, 0) {};
		\node [style=none] (2) at (-0.625, 0) {};
		\node [style=gn] (3) at (-0.125, -0.75) {};
		\node [style=rn] (4) at (1.375, -0.75) {};
		\node [style=rn] (5) at (0.625, -0.75) {};
	\end{pgfonlayer}
	\begin{pgfonlayer}{edgelayer}
		\draw (2.center) to (1);
		\draw [bend right=60] (3) to (4);
		\draw (5) to (3);
		\draw (5) to (4);
		\draw [bend left=60] (3) to (4);
	\end{pgfonlayer}
\end{tikzpicture}
} = \bra{x},
	\end{equation}
	the standard interpretation is clearly universal for the stabiliser fragment
	by the functoriality of \(\interp{-}\).
\end{lproof}

\section{Proof of soundness}
\label{app:soundness}
\subsection{Proof of soundness}
\label{app:proofs:soundness}

\begin{lproof}[of theorem~\ref{thm:soundness}]
  We prove that each of the equations in figure~\ref{fig:axioms} is sound, and
  this extends to the entire equational theory by functoriality of
  \(\interp{-}\).

  Firstly, straightforward calculation shows that, for any \(x,y \in \Z_p\),
  \begin{equation}
    \label{eq:scalars}
    \interp{\begin{tikzpicture}
	\begin{pgfonlayer}{nodelayer}
		\node [style=gn] (0) at (-0.75, 0) {};
		\node [style=rn] (1) at (0.75, 0) {};
		\node [style=rn] (2) at (0, 0) {};
	\end{pgfonlayer}
	\begin{pgfonlayer}{edgelayer}
		\draw [bend right=60] (0) to (1);
		\draw (2) to (0);
		\draw (2) to (1);
		\draw [bend left=60] (0) to (1);
	\end{pgfonlayer}
\end{tikzpicture}
} = \frac{1}{\sqrt{p}}
    \qc
    \interp{\begin{tikzpicture}
	\begin{pgfonlayer}{nodelayer}
		\node [style=gn, label={[gphase]below:$x,0$}] (0) at (-0.375, 0) {};
		\node [style=rn, label={[rphase]above:$z,0$}] (1) at (0.375, 0) {};
	\end{pgfonlayer}
	\begin{pgfonlayer}{edgelayer}
		\draw (0) to (1);
	\end{pgfonlayer}
\end{tikzpicture}
} = \sqrt{p} \omega^{-2^{-2}xz},
    \qand
    \interp{\begin{tikzpicture}
	\begin{pgfonlayer}{nodelayer}
		\node [style=gn, label={[gphase]below:$x,y$}] (0) at (-0.375, -0.25) {};
		\node [style=rn] (1) at (0.375, -0.25) {};
	\end{pgfonlayer}
	\begin{pgfonlayer}{edgelayer}
		\draw (0) to (1);
	\end{pgfonlayer}
\end{tikzpicture}
} = \sqrt{p}.
  \end{equation}
  We also have
  \begin{equation}
    \interp{\begin{tikzpicture}
	\begin{pgfonlayer}{nodelayer}
		\node [style=gn, label={[gphase]below:$x,y$}] (0) at (0, -0.25) {};
	\end{pgfonlayer}
\end{tikzpicture}
} = \sum_{k\in\Z_p} \omega^{2^{-1}(xk+yk^2)},
  \end{equation}
  and using the harmonic series for the Kronecker delta,
  \(\interp{\begin{tikzpicture}
	\begin{pgfonlayer}{nodelayer}
		\node [style=gn, label={[gphase]right:$x,0$}] (0) at (0, 0) {};
	\end{pgfonlayer}
\end{tikzpicture}
} = p\delta_{x,0}\).

  \textsc{(Zero)} and \textsc{(One)} then follow immediately from
  equation~\eqref{eq:scalars}.

  \textsc{(Fusion)} For any \(a,b,c,d \in \Z_p\),
  \begin{align}
    \interp{\,\input{figures/equations/soundness/Fusion_LHS.tikz}\,}
    &=\left( \sum_{k\in\Z_p} \omega^{2^{-1}(ak+bk^2)} \ket{k}^{\otimes n}
        \bra{k}^{\otimes m} \otimes \mathrm{id}_{m'-1} \right) \\ 
    &\quad\quad\quad \circ \left( \mathrm{id}_{n-1}
        \otimes \sum_{\ell\in\Z_p} \omega^{2^{-1}(c\ell+d\ell^2)}
        \ket{\ell}^{\otimes n'} \bra{\ell}^{\otimes m'} \right) \nonumber \\
    &= \sum_{k,\ell\in\Z_p} \omega^{2^{-1}(ak+bk^2)} \omega^{2^{-1}(c\ell+d\ell^2)} \ket{k}^{\otimes n}
        \bra{k}^{\otimes m-1} \ip{k}{\ell} \ket{\ell}^{\otimes n'-1} \bra{\ell}^{\otimes m'} \\
    &= \sum_{k,\ell\in\Z_p} \delta_{k,\ell} \omega^{2^{-1}(ak+bk^2)} \omega^{2^{-1}(c\ell+d\ell^2)} \ket{k}^{\otimes n}
        \bra{k}^{\otimes m-1} \ket{\ell}^{\otimes n'-1} \bra{\ell}^{\otimes m'} \\
    &= \sum_{k\in\Z_p} \omega^{2^{-1}((a+c)k+(b+d)k^2)} \ket{k}^{\otimes n+n'-1} \bra{k}^{\otimes m+m'-1} \\
    &= \interp{\,\input{figures/equations/soundness/Fusion_RHS.tikz}}
  \end{align}

  \textsc{(Colour)} For any \(a,b\in\Z_p\),
  \begin{align}
    \interp{\input{figures/equations/soundness/Colour_LHS.tikz}}
    &=\left( \sum_{j\in\Z_p} \dyad{j:X}{j:Z} \otimes \cdots \otimes
          \sum_{j\in\Z_p} \dyad{j:X}{j:Z} \right) \\
        &\quad\quad\quad \circ \sum_{k\in\Z_p} \omega^{2^{-1}(ak+bk^2)} \dyad{k:Z}{k:Z} \nonumber \\
        &\quad\quad\quad \circ \left( \sum_{\ell\in\Z_p} \dyad{\ell:Z}{-\ell:X} \otimes \cdots \otimes
          \sum_{\ell\in\Z_p} \dyad{\ell:Z}{-\ell:X} \right) \nonumber \\
    &= \sum_{k\in\Z_p} \omega^{2^{-1}(ak+bk^2)} \dyad{k:X}{-k:X} \\
    &= \interp{\input{figures/equations/soundness/Colour_RHS.tikz}}
  \end{align}

  \textsc{(Shear)} For any \(a,c,d\in\Z_p\),
  \begin{align}
    \interp{\input{figures/equations/soundness/Shear_LHS.tikz}}
    &= \sqrt{p}\omega^{-2^{-2}ac} \sum_{j,k,\ell\in\Z_p} \omega^{2^{-1}aj} \omega^{2^{-1}(ck+dk^2)} \omega^{2^{-1}a\ell} \dyad{j}{j} \dyad{-k:X}{k:X} \dyad{\ell}{\ell} \\
    &= \sqrt{p}\omega^{-2^{-2}ac} \sum_{j,k,\ell\in\Z_p} \omega^{2^{-1}aj} \omega^{2^{-1}(ck+dk^2)} \omega^{2^{-1}a\ell} \omega^{-jk} \omega^{-k\ell} \dyad{j}{\ell} \\
    &= \sqrt{p}\omega^{-2^{-2}ac} \sum_{j,k,\ell\in\Z_p} \omega^{2^{-1}aj} \omega^{2^{-1}a\ell} \omega^{2^{-1}k(dk+c-2j-2\ell)} \dyad{j}{\ell} \\
    &= \sqrt{p}\omega^{-2^{-2}ac} \sum_{j,k,\ell\in\Z_p} \omega^{2^{-1}aj} \omega^{2^{-1}a\ell} \omega^{2^{-1}(k+2^{-1}a)(d(k+2^{-1}a)+c-2j-2\ell)} \dyad{j}{\ell} \\
    &= \sqrt{p}\omega^{2^{-3}da^2} \sum_{j,k,\ell\in\Z_p} \omega^{2^{-1}k(dk+ad+c-2j-2\ell)} \dyad{j}{\ell} \\
    &= \sqrt{p}\omega^{2^{-3}da^2} \sum_{j,k,\ell\in\Z_p} \omega^{2^{-1}(ck+adk+dk^2)} \omega^{-jk} \omega^{-k\ell}\dyad{j}{\ell} \\
    &= \sqrt{p}\omega^{2^{-3}da^2} \sum_{k\in\Z_p} \omega^{2^{-1}(ck+adk+dk^2)} \dyad{-k:X}{k:X} \\
    &= \interp{\input{figures/equations/soundness/Shear_RHS.tikz}}
  \end{align}
  
  \textsc{(Char)}
  \begin{align}
    \interp{\input{figures/equations/soundness/Char_LHS.tikz}}
    &= \sqrt{p^p} \sum_{j,k\in\Z_p} \dyad{-k:X}{k:X}^{\otimes p} \ket{j}^{\otimes p} \bra{j} \\
    &= \sum_{j,k\in\Z_p} \omega^{-pjk} \ket{-k:X} \bra{j} \\
    &= \sum_{j,k\in\Z_p} \ket{-k:X} \bra{j} \\
    &= \interp{\begin{tikzpicture}
	\begin{pgfonlayer}{nodelayer}
		\node [style=none] (0) at (-1.25, 0) {};
		\node [style=gn] (1) at (-0.25, 0) {};
		\node [style=rn] (2) at (0.5, 0) {};
		\node [style=none] (3) at (1.5, 0) {};
	\end{pgfonlayer}
	\begin{pgfonlayer}{edgelayer}
		\draw (1) to (0.center);
		\draw (2) to (3.center);
	\end{pgfonlayer}
\end{tikzpicture}
}
  \end{align}

  \textsc{(Bigebra)}
  \begin{align}
    \interp{\input{figures/equations/soundness/Bigebra_LHS.tikz}}
    &= p \sum_{i,j,k,\ell\in\Z_p} \ket{-k:X} \ket{-\ell:X} \ip{k:X}{i} \ip{\ell:X}{i} \ip{k:X}{j} \ip{\ell:X}{j} \bra{i} \bra{j} \\
    &= \frac{1}{p} \sum_{i,j,k,\ell\in\Z_p} \omega^{-ik} \omega^{-i\ell} \omega^{-jk} \omega^{-j\ell} \ket{-k:X} \dyad{-\ell:X}{i} \bra{j} \\
    &= \frac{1}{p^3} \sum_{\substack{ i,j,k,\ell \\ m,n,o,p }\in\Z_p} \omega^{-ik} \omega^{-i\ell} \omega^{-jk} \omega^{-j\ell}
    \omega^{-km}\omega^{-n\ell}\omega^{io}\omega^{pj} \ket{m} \dyad{n}{o:X} \bra{p:X} \\
    &= \frac{1}{p^3} \sum_{\substack{ i,j,k,\ell \\ m,n,o,p }\in\Z_p} \omega^{i(o-k-\ell)} \omega^{j(p-k-\ell)}
    \omega^{-km}\omega^{-n\ell} \ket{m} \dyad{n}{o:X} \bra{p:X} \\
    &= \frac{1}{p} \sum_{\substack{ k,\ell \\ m,n,o,p }\in\Z_p} \delta_{o,k+\ell} \delta_{p,k+\ell}
    \omega^{-km}\omega^{-n\ell} \ket{m} \dyad{n}{o:X} \bra{p:X} \\
    &= \frac{1}{p} \sum_{k,\ell,m,n\in\Z_p} \omega^{-km}\omega^{-n\ell} \ket{m} \dyad{n}{k+\ell:X} \bra{k+\ell:X} \\
    &= \frac{1}{\sqrt{p}} \sum_{k,\ell,m,n\in\Z_p} \omega^{k(n-m)} \ket{m} \ket{n} \ip{n}{-k-\ell:X} \bra{k+\ell:X} \bra{k+\ell:X} \\
    &= \sqrt{p} \sum_{k,\ell,m,n\in\Z_p} \delta_{n,m} \ket{m} \ket{n} \ip{n}{-\ell:X} \bra{\ell:X} \bra{\ell:X} \\
    &= \sqrt{p} \sum_{m,\ell\in\Z_p} \ket{m} \dyad{m}{m} \dyad{-\ell:X}{\ell:X} \bra{\ell:X} \\
    &= \interp{\input{figures/equations/soundness/Bigebra_RHS.tikz}}
  \end{align}

  \textsc{(Copy)}
  \begin{align}
    \interp{\input{figures/equations/soundness/Copy_LHS.tikz}}
    &= \sqrt{p} \sum_{j,k\in\Z_p} \omega^{2^{-1}xj} \ip{k}{-j:X} \ket{k}\ket{k} \\
    &= \sum_{j,k\in\Z_p} \omega^{2^{-1}xj} \omega^{-jk} \ket{k}\ket{k} \\
    &= p \sum_{k\in\Z_p} \delta_{k,2^{-1}x} \ket{k}\ket{k} \\
    &= p \ket{2^{-1}x}\ket{2^{-1}x} \\
    &= \sum_{j\in\Z_p} \omega^{-2^{-1}xj} \ket{j:X} \sum_{k\in\Z_p} \omega^{-2^{-1}xk} \ket{k:X} \\
    &= \sum_{j\in\Z_p} \omega^{2^{-1}xj} \ket{-j:X} \sum_{k\in\Z_p} \omega^{2^{-1}xk} \ket{-k:X} \\
    &= \interp{\input{figures/equations/soundness/Copy_RHS.tikz}}
  \end{align}

  \textsc{(G-Elim)} 
  \begin{equation}
    \interp{\begin{tikzpicture}
	\begin{pgfonlayer}{nodelayer}
		\node [style=none] (10) at (-1, 0) {};
		\node [style=none] (11) at (1, 0) {};
		\node [style=gn] (12) at (0, 0) {};
	\end{pgfonlayer}
	\begin{pgfonlayer}{edgelayer}
		\draw (12) to (10.center);
		\draw (12) to (11.center);
	\end{pgfonlayer}
\end{tikzpicture}
}
    =
    \sum_{k\in\Z_p} \dyad{k}{k}
    =
    \interp{\begin{tikzpicture}
	\begin{pgfonlayer}{nodelayer}
		\node [style=none] (13) at (-0.75, 0) {};
		\node [style=none] (14) at (0.75, 0) {};
	\end{pgfonlayer}
	\begin{pgfonlayer}{edgelayer}
		\draw (14.center) to (13.center);
	\end{pgfonlayer}
\end{tikzpicture}
}
  \end{equation}

  \textsc{(R-Elim)} 
  \begin{align}
    \interp{\begin{tikzpicture}
	\begin{pgfonlayer}{nodelayer}
		\node [style=none] (15) at (1, 0) {};
		\node [style=none] (16) at (-1, 0) {};
		\node [style=rn] (17) at (0.375, 0) {};
		\node [style=rn] (18) at (-0.375, 0) {};
	\end{pgfonlayer}
	\begin{pgfonlayer}{edgelayer}
		\draw (18) to (16.center);
		\draw (18) to (17);
		\draw (17) to (15.center);
	\end{pgfonlayer}
\end{tikzpicture}
}
    &= \sum_{k\in\Z_p} \dyad{-k:X}{k:X} \dyad{-j:X}{j:X}
    = \sum_{k\in\Z_p} \delta_{k,-j} \dyad{-k:X}{j:X} \\
    &= \sum_{k\in\Z_p} \dyad{j:X}{j:X}
    = \interp{\begin{tikzpicture}
	\begin{pgfonlayer}{nodelayer}
		\node [style=none] (13) at (-0.75, 0) {};
		\node [style=none] (14) at (0.75, 0) {};
	\end{pgfonlayer}
	\begin{pgfonlayer}{edgelayer}
		\draw (14.center) to (13.center);
	\end{pgfonlayer}
\end{tikzpicture}
}
  \end{align}

  \textsc{(Mult)} For any \(z \in \Z_p^*\),
  \begin{align}
    \interp{\input{figures/equations/soundness/Mult_RHS.tikz}}
    &= \sum_{a,b,c,d\in\Z_p} \omega^{2^{-1}z^{-1}a^2} \omega^{2^{-1}zb^2} \omega^{2^{-1}z^{-1}c^2} \omega^{-ab} \omega^{-bc} \delta_{c,d} \dyad{a:Z}{d:X} \\
    &= \sum_{a,b,c\in\Z_p} \omega^{2^{-1}z^{-1}a^2} \omega^{2^{-1}zb^2} \omega^{2^{-1}z^{-1}c^2} \omega^{-ab} \omega^{-bc} \dyad{a:Z}{c:X} \\
    &= \sum_{a,b,c\in\Z_p} \omega^{2^{-1}z^{-1}(a-zb+c)^2} \omega^{-z^{-1}ac} \dyad{a:Z}{c:X} \\
    &= \left( \sum_{b\in\Z_p} \omega^{2^{-1}z^{-1}b^2} \right) \sum_{a,c\in\Z_p} \omega^{-z^{-1}ac} \dyad{a:Z}{c:X} \\
    &= \left( \sum_{b\in\Z_p} \omega^{2^{-1}z^{-1}b^2} \right) \sqrt{p} \sum_{a\in\Z_p} \dyad{a:Z}{z^{-1}a:Z} \\
    &= \left( \sum_{b\in\Z_p} \omega^{2^{-1}z^{-1}b^2} \right) \sqrt{p} \sum_{a\in\Z_p} \dyad{za:Z}{a:Z} \\
    &= \left( \sum_{b\in\Z_p} \omega^{2^{-1}z^{-1}b^2} \right) \sum_{a,c\in\Z_p} \omega^{-zac} \dyad{c:X}{a:Z} \\
    &= \left( \sum_{b\in\Z_p} \omega^{2^{-1}z^{-1}b^2} \right) \sqrt{p^z} \sum_{a,c\in\Z_p} \dyad{c:X}{c:X}^{\otimes z} \ket{a:Z}^{\otimes z}\bra{a:Z} \\
    &= \interp{\input{figures/equations/soundness/Mult_LHS.tikz}}.
  \end{align}

  \textsc{(Gauss)}
  follows using equation~(73) of \cite{appleby_properties_2009}:
  \begin{equation}
    \interp{\begin{tikzpicture}
	\begin{pgfonlayer}{nodelayer}
		\node [style=gn, label={[gphase]above:$0,z$}] (0) at (0, 0) {};
	\end{pgfonlayer}
\end{tikzpicture}
}
    = \sum_{j\in\Z_p} \omega^{2^{-1}zj^2} 
    = -i^{\frac{p+3}{2}} \ell_p(z) 
    = -i^{\frac{p+3}{2}} (-1)^{\chi_p(z)} 
    = \interp{\begin{tikzpicture}
	\begin{pgfonlayer}{nodelayer}
		\node [style=gn, label={[gphase]left:$0,1$}] (0) at (0, -0.5) {};
		\node [style=none] (1) at (-0.5, 0.5) {$\Large($};
		\node [style=none] (2) at (1.45, 0.55) {$\Large)^{\otimes \chi_p(z)}$};
		\node [style=none] (3) at (0, 0.5) {\scalebox{1.4}{$\star$}};
	\end{pgfonlayer}
\end{tikzpicture}
}.
  \end{equation}
  where we have used the fact that \(\ell_p(1) = 1\) (the Legendre symbol mod
  \(p\) of 1), and also that \(\ell_p(z) = (-1)^{\chi_p(z)}\) whenever \(z \neq
  0\).

  \textsc{(M-One)}
  \begin{equation}
    \interp{\begin{tikzpicture}
	\begin{pgfonlayer}{nodelayer}
		\node [style=none] (0) at (0, 0.5) {\scalebox{1.4}{$\star$}};
		\node [style=none] (1) at (0, -0.25) {\scalebox{1.4}{$\star$}};
	\end{pgfonlayer}
\end{tikzpicture}
}
    = (-1) (-1)
    = 1
    = \interp{\input{figures/equations/soundness/M-One_RHS.tikz}}.
  \end{equation}

  \textsc{(M-Elim)}
  \begin{align}
    \interp{\input{figures/equations/soundness/M-Elim_LHS.tikz}}
    &= \sum_{j,k\in\Z_p} \omega^{2^{-1}(ak+bk^2)} \bra{k:X}\ket{-zj}\bra{j} \\
    &= \sum_{j,k\in\Z_p} \omega^{2^{-1}(ak+bk^2)} \frac{\omega^{2^{-1}k zj}}{\sqrt{p}} \bra{j} \\
    &= \frac{1}{\sqrt{p}} \sum_{j,k\in\Z_p} \omega^{2^{-1}(ak+bk^2)} \omega^{2^{-1}k zj} \bra{j} \\
    &= \sum_{k\in\Z_p} \omega^{2^{-1}(ak+bk^2)} \bra{-zk:X} \\
    &= \sum_{k\in\Z_p} \omega^{2^{-1}(-az^{-1}k+bz^{-2}k^2)} \bra{k:X} \\
    &= \interp{\begin{tikzpicture}
	\begin{pgfonlayer}{nodelayer}
		\node [style=none] (0) at (-1, 0) {};
		\node [style=rn, label={[rphase]above:$\minu z^{\minu 1}a,z^{\minu 2} b$}] (1) at (1, 0) {};
	\end{pgfonlayer}
	\begin{pgfonlayer}{edgelayer}
		\draw (1) to (0.center);
	\end{pgfonlayer}
\end{tikzpicture}
}.
  \end{align}
  
\end{lproof}

\section{Proof of completeness}
\label{app:completeness}
\subsection{Elementary derivations}

\begin{lemma}
  \label{lem:hadamard_product}
  Products of Hadamards are antipodes:
  \begin{equation*}
    \input{figures/equations/hadamard_product.tikz}
  \end{equation*}
\end{lemma}
\begin{lproof}
  \begin{equation*}
    \input{figures/equations/hadamard_product_proof.tikz}
  \end{equation*}
\end{lproof}

\begin{lemma}
  \label{lem:hadamard_antipode}
  Hadamards and antipodes commute:
  \begin{equation*}
    \input{figures/equations/hadamard_antipode.tikz}
  \end{equation*}
\end{lemma}
\begin{lproof}
  \begin{equation*}
    \input{figures/equations/hadamard_antipode_proof.tikz}
  \end{equation*}
\end{lproof}

\begin{lemma}
  \label{lem:hadamard_inverse}
  The inverse Hadamard is a product of Hadamards:
  \begin{equation*}
    \input{figures/equations/hadamard_inverse.tikz}
  \end{equation*}
\end{lemma}
\begin{lproof}
  \begin{equation*}
    \input{figures/equations/hadamard_inverse_proof.tikz}
  \end{equation*}
\end{lproof}

\begin{lemma}
  \label{lem:antipode_unit}
  Units absorb antipodes:
  \begin{equation*}
    \input{figures/equations/antipode_unit.tikz}
  \end{equation*}
\end{lemma}
\begin{lproof}
  The red equation is basically a subcase of \textsc{(M-Elim)}:
  \begin{equation*}
    \input{figures/equations/antipode_unit_proof_red.tikz}
  \end{equation*}
  and the green rule is obtained using \textsc{(Colour)}:
  \begin{equation*}
    \input{figures/equations/antipode_unit_proof_green.tikz}
  \end{equation*}
\end{lproof}
 
\begin{lemma}
  \label{lem:hadamard_euler}
  The Hadamards admit simple Euler decompositions:
  \begin{equation*}
    \input{figures/equations/hadamard_euler.tikz}
  \end{equation*}
\end{lemma}
\begin{lproof}
  \begin{equation*}
    \input{figures/equations/hadamard_euler_proof.tikz}
  \end{equation*}
  \begin{equation*}
    \input{figures/equations/hadamard_euler_proof_inverse.tikz}
  \end{equation*}
\end{lproof}

\begin{lemma}
  \label{lem:hopf}
  The Hopf identity is derivable in \(\ZXeq\):
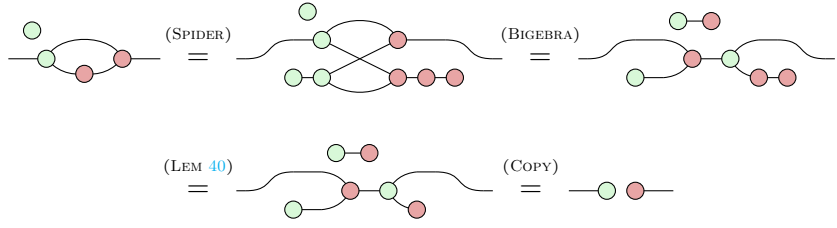
  \begin{equation*}
    \input{figures/equations/hopf.tikz}
  \end{equation*}
\end{lemma}
\begin{lproof}
  \begin{equation*}
    \input{figures/equations/hopf_proof.tikz}
  \end{equation*}
  Then, the second rule follows using the colour change meta-rule.
\end{lproof}

\begin{lemma}
  \label{lem:scalar_elementary}
  The following rules hold between the ``elementary'' scalars:
  \begin{equation*}
    \input{figures/equations/scalar_elementary.tikz}
  \end{equation*}
\end{lemma}
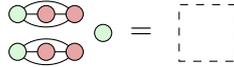
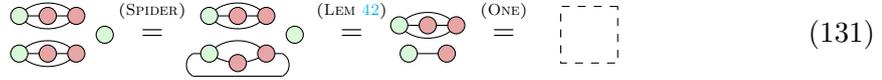
\begin{lproof}
  \begin{equation}
    \input{figures/equations/scalar_elementary_proof.tikz}
  \end{equation}
\end{lproof}

\begin{lemma}
  \label{lem:loop}
  Self-loops on green spiders can be eliminated:
  \begin{equation*}
    \input{figures/equations/loop.tikz}
  \end{equation*}
  We include the colour-swapped version of this rule for completeness, even
  though it no longer includes a genuine self-loop.
\end{lemma}
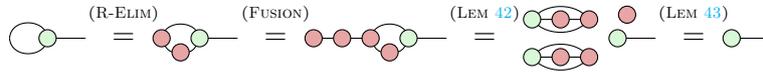
\begin{lproof}
  \begin{equation*}
    \input{figures/equations/loop_proof.tikz}
  \end{equation*}
  The red version is obtained using \textsc{(Colour)}.
\end{lproof}

\begin{lemma}
  \label{lem:scalar_phase}
  The following rules hold between the ``phase'' scalars: for any \(a,b \in
  \Z_p\),
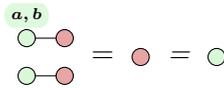
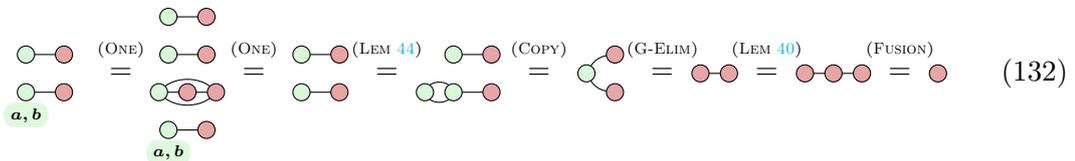
  \begin{equation*}
    \input{figures/equations/scalar_phase.tikz}
  \end{equation*}
\end{lemma}
\begin{lproof}
  \begin{equation}
    \input{figures/equations/scalar_phase_proof.tikz}
  \end{equation}
  The rule \(\begin{tikzpicture}
	\begin{pgfonlayer}{nodelayer}
		\node [style=rn] (0) at (-1, 0) {};
		\node [style=none] (1) at (0, 0) {=};
		\node [style=gn] (2) at (1, 0) {};
	\end{pgfonlayer}
\end{tikzpicture}
\) is immediate using
  \textsc{(Colour)}.
\end{lproof}

\begin{lemma}
  \label{lem:unit_rotation_elim}
  Green units absorb red rotations and vice-versa:
  \begin{equation*}
    \input{figures/equations/unit_rotation_elim.tikz}
  \end{equation*}
\end{lemma}
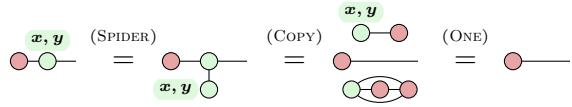
\begin{lproof}
  \begin{equation*}
    \input{figures/equations/unit_rotation_elim_proof_green.tikz}
  \end{equation*}
\end{lproof}

\begin{lemma}
  \label{lem:bigebra_simplified}
  The scalars in the bigebra law can be simplified to:
  \begin{equation*}
    \input{figures/equations/bigebra_simplified.tikz}
  \end{equation*}
\end{lemma}
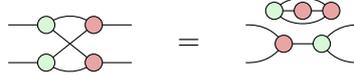
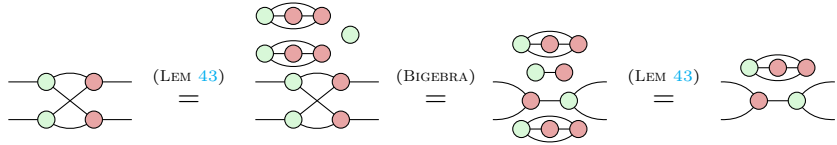
\begin{lproof}
  \begin{equation*}
    \input{figures/equations/bigebra_simplified_proof.tikz}
  \end{equation*}
\end{lproof}

\begin{lemma}
  \label{lem:antipode_copy}
  The green co-multiplication copies antipodes:
  \begin{equation*}
    \input{figures/equations/antipode_copy.tikz}
  \end{equation*}
\end{lemma}
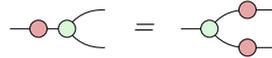
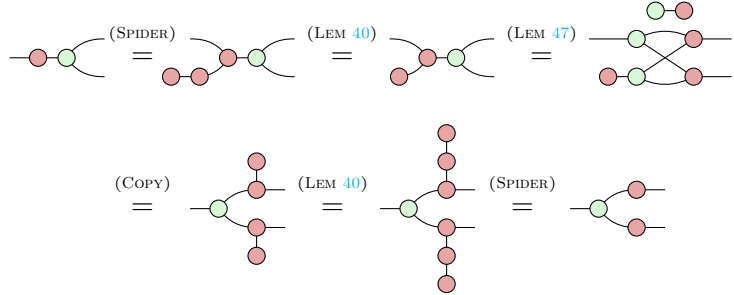
\begin{lproof}
  \begin{equation*}
    \input{figures/equations/antipode_copy_proof.tikz}
  \end{equation*}
\end{lproof}

\begin{lemma}
  \label{lem:bigebra_mn}
  The bigebra law holds for arbitrary arities: for any \(m,n \in \N\),
  \begin{equation*}
    \input{figures/equations/bigebra_arbitrary.tikz}\quad,
  \end{equation*}
  where in the diagram on the LHS, there are \(m\) green and \(n\) red spiders,
  and each green spider is connected to each red spider by a single wire.
\end{lemma}
\begin{lproof}
  The cases \(m=0\) or \(n=0\) correspond to the copy rules.
  The case \(n=1\) is the antipode copy rule (lemma~\ref{lem:antipode_copy}) and
  the case \(m=1\) is trivial by the green identity rule.
  The case \(n=2,m=2\) is lemma~\ref{lem:bigebra_simplified}.
  The general case follows from a straightforward induction (which furthermore
  is analogous to the qubit case).
\end{lproof}

\begin{lemma}
  \label{lem:antipode_multiplier}
  The antipode can be rewritten as a multiplication:
  \begin{equation*}
    \input{figures/equations/antipode_multiplier.tikz}
  \end{equation*}
\end{lemma}
\begin{lproof}
  \begin{equation*}
    \input{figures/equations/antipode_multiplier_proof.tikz}
  \end{equation*}
\end{lproof}

\begin{lemma}
  \label{lem:antipode_phase}
  For any \(x,y \in \Z_p\),
  \begin{equation*}
    \input{figures/equations/antipode_phase.tikz}
  \end{equation*}
\end{lemma}
\begin{lproof}
  \begin{equation*}
    \input{figures/equations/antipode_phase_proof_green.tikz}
  \end{equation*}
  The rule for red spiders is obtained from the green rule:
  \begin{equation*}
    \input{figures/equations/antipode_phase_proof_red.tikz}\quad.
  \end{equation*}
\end{lproof}

\begin{lemma}
  \label{lem:antipode_spider}
  For any \(x,y \in \Z_p\) and \(m,n \in \N\),
  \begin{equation*}
    \input{figures/equations/antipode_spider.tikz}
  \end{equation*}
\end{lemma}
\begin{lproof}
  \begin{equation*}
    \input{figures/equations/antipode_spider_proof.tikz}
  \end{equation*}
  The rule for red spiders is obtained using \textsc{(Colour)} like in the proof
  of lemma~\ref{lem:antipode_phase}.
\end{lproof}

\begin{lemma}
  \label{lem:pauli_copy_phase}
  Green spiders copy red Pauli phases, and vice-versa: for any \(x \in
  \Z_p\),
  \begin{equation*}
    \input{figures/equations/pauli_copy_phase.tikz}
  \end{equation*}
\end{lemma}
\begin{lproof}
  \begin{equation*}
    \input{figures/equations/pauli_copy_phase_proof.tikz}
  \end{equation*}
  The other rule is obtained using \textsc{(Colour)}.
\end{lproof}

\subsection{Multipliers}

\begin{lemma}
  \label{lem:multiplier_sum}
  Parallel multipliers sum: for any \(x,y \in \Z_p\):
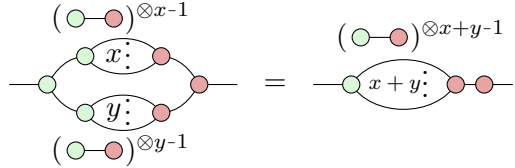
  \begin{equation*}
    \input{figures/equations/multiplier_sum.tikz}
  \end{equation*}
\end{lemma}
\begin{lproof}
  This is a straightforward consequence of \textsc{(Spider)}.
\end{lproof}

\begin{lemma}
  \label{lem:multiplier_elim}
  For any \(z \in \Z_p^*\),
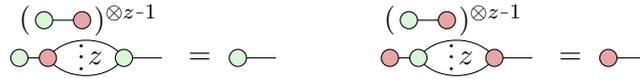
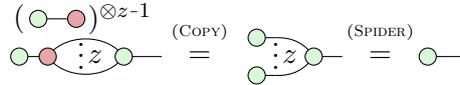
  \begin{equation*}
    \input{figures/equations/multiplier_elim.tikz}
  \end{equation*}
\end{lemma}
\begin{lproof}
  \begin{equation*}
    \input{figures/equations/multiplier_elim_proof_green.tikz}
  \end{equation*}
  The right rule is obtained using \textsc{(Colour)} as previously.
\end{lproof}

\begin{lemma}
  \label{lem:multiplier_product}
  For any \(x,y \in \N\),
  \begin{equation*}
    \input{figures/equations/multiplier_product.tikz}
  \end{equation*}
\end{lemma}
\begin{lproof}
  \begin{equation*}
    \input{figures/equations/multiplier_product_proof.tikz}
  \end{equation*}
\end{lproof}

\begin{lemma}
  \label{lem:multiplier_inverse}
  For any \(x \in \Z_p^*\),
  \begin{equation*}
    \input{figures/equations/multiplier_inverse.tikz}
  \end{equation*}
\end{lemma}
\begin{lproof}
  \begin{equation*}
    \input{figures/equations/multiplier_inverse_proof.tikz}
  \end{equation*}

  The second equality follows from the \textsc{(Colour)} meta-rule.
\end{lproof}

Lemmas \ref{lem:multiplier_sum}-\ref{lem:multiplier_inverse} suffice to prove
proposition~\ref{prop:multiplier} and \ref{prop:weighted_hadamard}, so we
consider those rules proved from this point on, and adopt the multiplier
notation.

\begin{lemma}
  \label{lem:multiplier_copy}
  Spiders copy invertible multipliers: for any \(x \in \Z_p^*\),
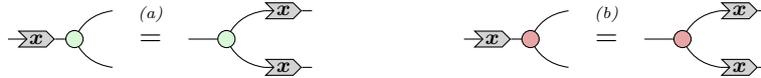
  \begin{equation*}
    \input{figures/equations/multiplier_copy.tikz}
  \end{equation*}
\end{lemma}
\begin{lproof}
  (a) follows from direct calculation using the definition of multipliers:
  \begin{equation*}
    \input{figures/equations/multiplier_copy_proof_green.tikz}
  \end{equation*}
  Then (b) is obtained from (a) using already established properties of the
  multipliers:
  \begin{equation*}
    \input{figures/equations/multiplier_copy_proof_red.tikz}
  \end{equation*}
\end{lproof}

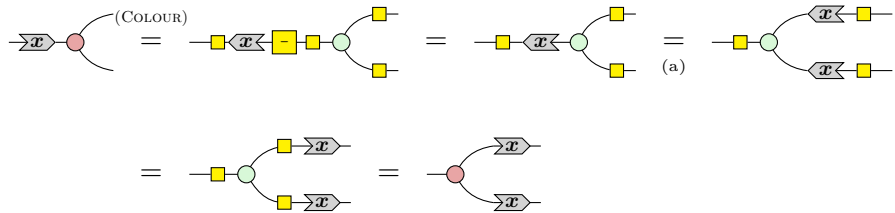
\begin{lproof}[of proposition~\ref{prop:local_scaling}]
  This proposition now follows straightforwardly using
  lemma~\ref{lem:multiplier_copy}, the definition of H-boxes (equation
  \eqref{eq:H_box_definition}) and proposition~\ref{prop:multiplier}.
\end{lproof}

\begin{lemma}
  \label{lem:multiplier_spider}
  The action of multipliers on spiders is given by, for any \(x \in \Z_p^*\),
  \begin{equation*}
    \input{figures/equations/multiplier_spider.tikz}
  \end{equation*}
\end{lemma}
\begin{lproof}
  This follows straightforwardly using lemma~\ref{lem:multiplier_copy} and
  \textsc{(M-Elim)}.
\end{lproof}

\begin{lemma}
  \label{lem:clifford_states}
  Any pure-Clifford states can be represented in both the red and green
  fragment: for any \(x \in \Z_p^*\),
  \begin{equation*}
    \input{figures/equations/clifford_states.tikz}
  \end{equation*}
\end{lemma}
\begin{lproof}
  Firstly, we prove the subcase \(x=1\) of (a):
  \begin{equation*}
    \input{figures/equations/clifford_states_proof_green.tikz}
  \end{equation*}
  Then the general case for any invertible \(x\) follows using
  lemma~\ref{lem:multiplier_spider}.

  (b) follows once again using \textsc{(Colour)}.
\end{lproof}

\begin{lemma}
  \label{lem:scalar_i_elim}
  \begin{equation*}
    \begin{tikzpicture}
	\begin{pgfonlayer}{nodelayer}
		\node [style=gn, label={[gphase]right:$0,1$}] (0) at (-1, 0.375) {};
		\node [style=gn, label={[gphase]right:$0,\minu 1$}] (1) at (-1, -0.375) {};
		\node [style=none] (2) at (0, 0) {$=$};
		\node [style=gn] (3) at (1, 0) {};
	\end{pgfonlayer}
\end{tikzpicture}

  \end{equation*}
\end{lemma}
\begin{lproof}
  Since \(1\) is always a square:
  \begin{equation*}
    \input{figures/equations/scalar_i_elim_proof.tikz}
  \end{equation*}
\end{lproof}

\begin{lemma}
  \label{lem:scalar_gauss_elim}
  For any \(z\in\Z_p^*\),
  \begin{equation*}
    \begin{tikzpicture}
	\begin{pgfonlayer}{nodelayer}
		\node [style=gn, label={[gphase]below:$0,z$}] (0) at (-2, 0) {};
		\node [style=gn, label={[gphase]above:$0,\minu z$}] (1) at (-1, 0) {};
		\node [style=none] (2) at (0, 0) {$=$};
		\node [style=gn] (3) at (1, 0) {};
	\end{pgfonlayer}
\end{tikzpicture}

  \end{equation*}
\end{lemma}
\begin{lproof}
  If \(z\) is a square, then there is some \(\alpha \in \Z_p\) such that
  \begin{equation*}
    \input{figures/equations/scalar_gauss_elim_proof.tikz}
  \end{equation*}
  If \(-z\) is a square, then there is again some \(\alpha \in \Z_p\) such that
  \begin{equation*}
    \input{figures/equations/scalar_gauss_elim_proof_2.tikz}
  \end{equation*}
  Now, if neither \(z\) nor \(-z\) is a square, then by
  corollary~\ref{cor:at_least_one_square}, \(-1\) must be a square. We then
  have
  \begin{equation*}
    \input{figures/equations/scalar_gauss_elim_proof_3.tikz}
  \end{equation*}
\end{lproof}

\begin{lemma}
  \label{lem:scalar_omega_elim}
  For any \(z\in\Z_p\),
  \begin{equation*}
    \input{figures/equations/scalar_omega_elim.tikz}
  \end{equation*}
\end{lemma}
\begin{lproof}
  \begin{equation*}
    \input{figures/equations/scalar_omega_elim_proof.tikz}
  \end{equation*}
  where we have freely used lemma~\ref{lem:antipode_phase} to commute antipodes
  and spiders throughout.
\end{lproof}

\begin{lemma}
  \label{lem:scalar_imaginary_elim}
  If \(p \equiv 3 \mod 4\),
  \begin{equation*}
    \begin{tikzpicture}
	\begin{pgfonlayer}{nodelayer}
		\node [style=gn, label={[gphase]above:$0,1$}] (0) at (-1.25, -0.35) {};
		\node [style=gn, label={[gphase]below:$0,1$}] (1) at (-1.25, 0.4) {};
		\node [style=gn] (2) at (1.25, -0.375) {};
		\node [style=none] (3) at (1.25, 0.375) {\scalebox{1.4}{$\star$}};
		\node [style=none] (4) at (0, 0) {$=$};
	\end{pgfonlayer}
\end{tikzpicture}

  \end{equation*}
\end{lemma}
\begin{proof}
  If \(p \equiv 3 \mod 4\), \(-1\) is \emph{not} a square, so that
  \begin{equation*}
    \input{figures/equations/scalar_imaginary_elim_proof.tikz}
  \end{equation*}
\end{proof}

\begin{lemma}
  \label{lem:zero_elementary}
  All the elementary ``zero'' diagrams are equal:
  \begin{equation*}
    \begin{tikzpicture}
	\begin{pgfonlayer}{nodelayer}
		\node [style=gn, label={[gphase]above:$z,0$}] (0) at (-1, 0) {};
		\node [style=none] (1) at (0, 0) {$=$};
		\node [style=gn, label={[gphase]above:$1,0$}] (2) at (1, 0) {};
	\end{pgfonlayer}
\end{tikzpicture}

  \end{equation*}
\end{lemma}
\begin{lproof}
  \begin{equation*}
    \input{figures/equations/zero_elementary_proof.tikz}
  \end{equation*}
\end{lproof}

\begin{lemma}
  \label{lem:zero_amplitudes}
  The ``zero'' diagram absorbs unlabelled elementary scalars:
  \begin{equation*}
    \input{figures/equations/zero_amplitudes.tikz}
  \end{equation*}
\end{lemma}
\begin{lproof}
  First note that
  \begin{equation*}
    \input{figures/equations/zero_amplitudes_proof_1.tikz}
  \end{equation*}
  so that
  \begin{equation*}
    \input{figures/equations/zero_amplitudes_proof_2.tikz}
  \end{equation*}
  and
  \begin{equation*}
    \input{figures/equations/zero_amplitudes_proof_3.tikz}
  \end{equation*}
\end{lproof}

\begin{lemma}
  \label{lem:zero_phases}
  The zero diagram absorbs the phase scalars: for any \(x \in \Z_p\) and \(z \in
  \Z_p^*\),
  \begin{equation*}
    \input{figures/equations/zero_phases.tikz}
  \end{equation*}
\end{lemma}
\begin{lproof}
  \begin{equation*}
    \input{figures/equations/zero_phases_proof_1.tikz}
  \end{equation*}
  \begin{equation*}
    \input{figures/equations/zero_phases_proof_2.tikz}
  \end{equation*}
\end{lproof}

\begin{lemma}
  \label{lem:scalar_gauss_multiplication}
  Quadratic phases satisfy the following multiplication law: for any \(x,y \in
  \Z_p^*\)
  \begin{equation*}
    \input{figures/equations/scalar_gauss_multiplication.tikz}
  \end{equation*}
\end{lemma}
\begin{lproof}
  This follows from applications of \textsc{(Gauss)} and
  lemma~\ref{lem:scalar_i_elim}.
\end{lproof}

\begin{lemma}
  \label{lem:scalar_omega_multiplication}
  For any \(x,y\in\Z_p\),
  \begin{equation*}
    \input{figures/equations/scalar_omega_multiplication.tikz}
  \end{equation*}
\end{lemma}
\begin{lproof}
  \begin{equation*}
    \input{figures/equations/scalar_omega_multiplication_proof.tikz}
  \end{equation*}
  where we have freely used lemma~\ref{lem:antipode_phase} to commute antipodes
  and spiders throughout.
\end{lproof}

\begin{lemma}
  \label{lem:hadamard_euler_better}
  The Euler decomposition of the Hadamards can be ``improved'' to:
  \begin{equation*}
    \input{figures/equations/hadamard_euler_better.tikz}
  \end{equation*}
\end{lemma}
\begin{lproof}
  This follows from applying lemma~\ref{lem:clifford_states} to the
  decomposition of lemma~\ref{lem:hadamard_euler}, then using
  lemma~\ref{lem:scalar_i_elim} to simplify the scalar.
\end{lproof}

\begin{lemma}
  \label{lem:H_loop}
  Hadamard loops correspond to pure-Clifford operations: for any \(x \in \Z_p\)
  and \(z \in \Z_p^*\),
  \begin{equation*}
    \input{figures/equations/H_loop.tikz}
  \end{equation*}
\end{lemma}
\begin{lproof}
  The case \(x=0\) is clear by decomposing the H-box. We begin by proving the
  case \(x=1\):
  \begin{equation*}
    \input{figures/equations/H_loop_proof_green.tikz}
  \end{equation*}
  The general case can be obtained by decomposing the weighted H-box into \(x\)
  H-loops using the sum rule from proposition~\ref{prop:weighted_hadamard}.
  Then, under the assumption that the weight is invertible, the red version once
  again follows using \textsc{(Colour)} and the equations of
  proposition~\ref{prop:weighted_hadamard}:
  \begin{equation*}
    \input{figures/equations/H_loop_proof_red.tikz}
  \end{equation*}
\end{lproof}

\begin{lproof}[of proposition~\ref{prop:pauli_stabiliser}]
  \begin{equation*}
    \input{figures/equations/pauli_stabiliser_proof.tikz}
  \end{equation*}
\end{lproof}

\begin{lproof}[of proposition~\ref{prop:local_scaling}]
  This follows straightforwardly using proposition~\ref{prop:multiplier}.
\end{lproof}

\begin{lemma}
  \label{lem:local_complementation_triangle}
  For any \(x,y \in \Z_p^*\),
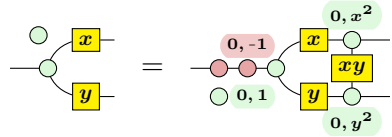
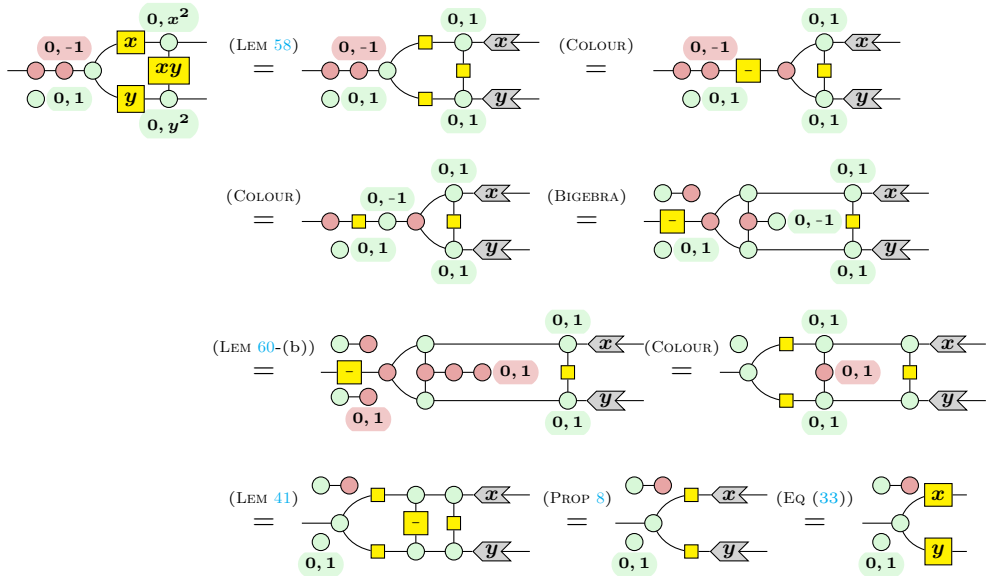
  \begin{equation*}
    \input{figures/equations/local_complementation_triangle.tikz}
  \end{equation*}
\end{lemma}
\begin{lproof}
  \begin{equation*}
    \input{figures/equations/local_complementation_triangle_proof.tikz}
  \end{equation*}
  and we can eliminate the \(\begin{tikzpicture}
	\begin{pgfonlayer}{nodelayer}
		\node [style=gn, label={[gphase]right:$0,1$}] (0) at (0, -0.25) {};
	\end{pgfonlayer}
\end{tikzpicture}
\) scalar using
  lemma~\ref{lem:scalar_i_elim}.
\end{lproof}

\begin{lemma}
  \label{lem:local_complementation_tree}
  Let \(\Sigma\) be a \(\Z_p\)-weighted star graph on \(N \in \N\) vertices,
  i.e.\, it is a tree with \(N-1\) leaves. Order the vertices such that the
  first vertex is the only internal vertex, and it follows that all of the edges
  have weights \(\Sigma_{1w}\) for \(w\) ranging from \(2\) to \(N\). Then,
  \begin{equation*}
    \input{figures/equations/local_complementation_tree.tikz}
  \end{equation*}
  where \(\Sigma \overset{\gamma}{\star} 1\) is the graph obtained by adding an
  edge weighted in \(\gamma \Sigma_{1m}\Sigma{1n}\) between the wires connected
  to each pair of vertices \(m,n \neq 1\).
\end{lemma}
\begin{lproof}
  The proof is by induction on the size \(N\) of the tree. Assume the lemma is
  true for any tree of size \(N-1\).
  First, bend all the wires with green phases into inputs:
  \begin{equation*}
    \input{figures/equations/local_complementation_tree_proof_bent.tikz}
  \end{equation*}
  Then,
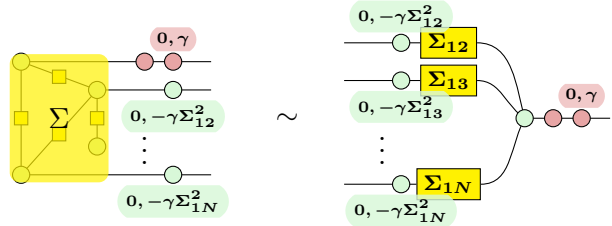
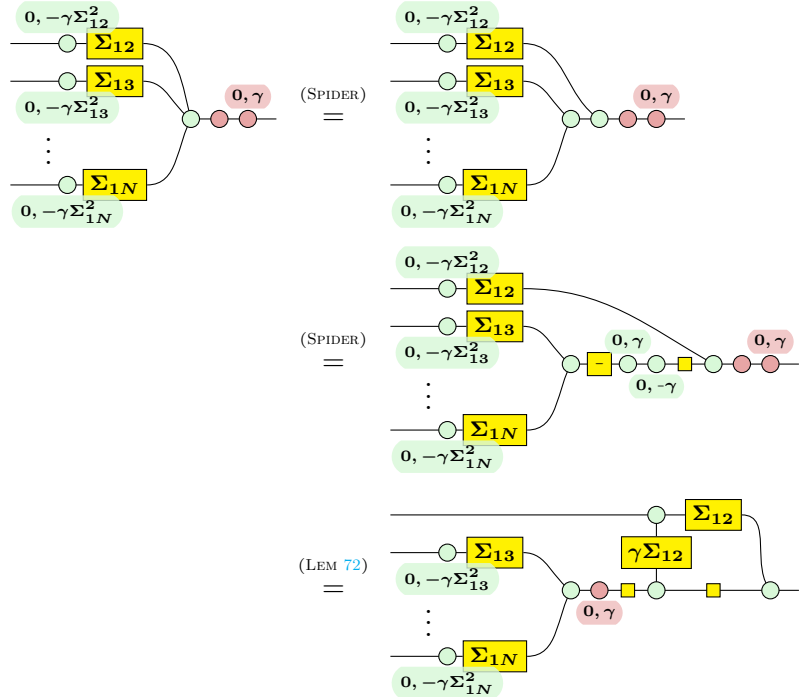
  \begin{equation*}
    \input{figures/equations/local_complementation_tree_proof_1.tikz}
  \end{equation*}
  Then, we recognise the subtree with head \(1\) and leaves \(3\) to \(N\),
  which is thus a tree of size \(N-1\), to which we can apply the inductive
  hypothesis:
  \begin{equation*}
    \input{figures/equations/local_complementation_tree_proof_2.tikz}
  \end{equation*}
  In the last step, we have pulled out the ``head'' of the \(N-1\) tree before
  applying lemma~\ref{lem:bigebra_mn}, since the \(\Sigma_{1w}\)-weighted
  Hadamards on each edge to the head are left unchanged by the local
  complementation.
  
  Then, copy each \(\Sigma_{1w}\)-weighted Hadamard through the corresponding
  red spider on its right (which changes its colour and multiplies the weight by
  \(-1\)), and copy the multiplier and antipode through the green spiders below
  and to the left respectively. Fusing the resulting green spiders we get:
  \begin{equation*}
    \input{figures/equations/local_complementation_tree_proof_3.tikz}
  \end{equation*}
  Bending the inputs back to outputs completes the proof.
\end{lproof}

\begin{lproof}[of proposition~\ref{prop:local_complementation}]
  The proof of this proposition is rather cumbersome to write in standard
  \(\ZXp\). We give a sketch of the proof, and leave writing it out formally to
  future work, since it has a much clearer form in the scalable
  \(\ZXp\)-calculus.
  
  The idea is to split out the subtree of \(G\) with head \(w\) and leaves each
  of the neighbours of \(w\). Then, one applies
  lemma~\ref{lem:local_complementation_tree}, and adds the resulting edges into
  \(G\) using the addivitity of parallel H-boxes
  (proposition~\ref{prop:weighted_hadamard}).
\end{lproof}

\subsection{Main completeness proofs}

\begin{lproof}[of proposition~\ref{prop:c1_completeness}]
  The single-qupit Clifford group is generated by the invertible generators
  under sequential composition. It therefore suffices to show that the
  composition of either of the above diagrams with such a generator can be
  rewritten to the claimed form. Now, every rewrite rule, except \textsc{(Zero)}
  can be interpreted as an equality up to an invertible scalar, if we simply
  ignore the parts of the equations disconnected from both the inputs and the
  outputs. Thus, we can freely use any of these rules in our proof of
  normalisation up to invertible scalars, and any equation derivable without
  \textsc{(Zero)}.

  \emph{(Multipliers)}
  For the multipliers this is very straightforward: for any
  \(x \in \Z_d^*\),
  \begin{equation*}
    \input{figures/equations/c1_multiplier_forms.tikz}
  \end{equation*}

  \emph{(Hadamard)}
  We can ignore the weight since this can be extracted as
  a multiplier and then the previous proof applies. Then, for the first normal
  form,
  \begin{equation*}
    \input{figures/equations/c1_hadamard_forms_10.tikz}
  \end{equation*}
  we need to split subcases based on the value of \(v\). If \(v=0\), then we
  have
  \begin{equation*}
    \input{figures/equations/c1_hadamard_forms_11.tikz}
  \end{equation*}
  where in the last step we have used \textsc{(Shear)} and \textsc{(Spider)}
  several times to commute the Paulis on the right back into the leftmost red
  and green spiders.

  If \(v \neq 0\), then,
  \begin{equation*}
    \input{figures/equations/c1_hadamard_forms_12.tikz}
  \end{equation*}

  The second normal form can be done in one go.
  \begin{equation*}
    \input{figures/equations/c1_hadamard_forms_2.tikz}
  \end{equation*}

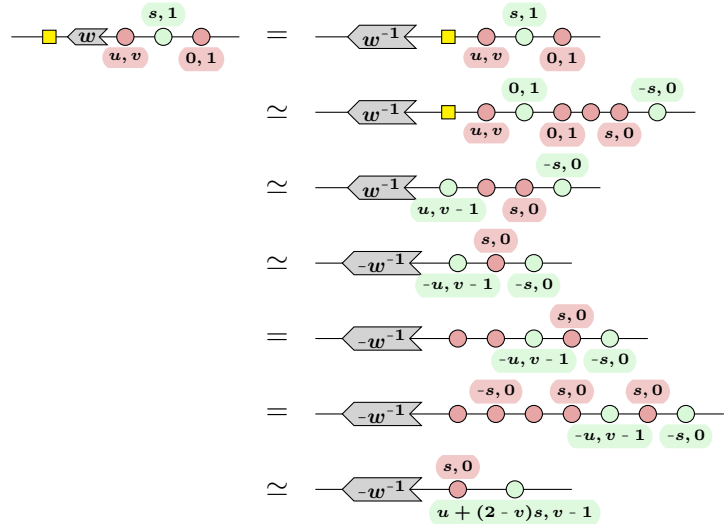
  \emph{(Rotations)}
  The red phase is once again very straightforward: let \(x,y \in \Z_d\), then,
  \begin{equation*}
    \input{figures/equations/c1_phase_forms_red.tikz}
  \end{equation*}

  The green phases are more involved. Firstly, note that we can ignore the
  multiplier, since we have:
  \begin{equation*}
    \input{figures/equations/c1_phase_forms_green_multipliers.tikz}
  \end{equation*}
  and \(x,y\) are arbitrary. Then, the Pauli part can be straightforwardly taken
  care of using \textsc{(Shear)} to commute it through to the green spider on
  the right. Finally, we can view \(\begin{tikzpicture}
	\begin{pgfonlayer}{nodelayer}
		\node [style=gn, label={[gphase]:$0,y$}] (0) at (0, 0) {};
		\node [style=none] (1) at (-0.75, 0) {};
		\node [style=none] (2) at (0.75, 0) {};
	\end{pgfonlayer}
	\begin{pgfonlayer}{edgelayer}
		\draw (2.center) to (1.center);
	\end{pgfonlayer}
\end{tikzpicture}
\)
  as the \(y\)-fold composition of
  \(\begin{tikzpicture}
	\begin{pgfonlayer}{nodelayer}
		\node [style=gn, label={[gphase]:$0,1$}] (0) at (0, 0) {};
		\node [style=none] (1) at (-0.75, 0) {};
		\node [style=none] (2) at (0.75, 0) {};
	\end{pgfonlayer}
	\begin{pgfonlayer}{edgelayer}
		\draw (2.center) to (1.center);
	\end{pgfonlayer}
\end{tikzpicture}
\), so that we have reduced the
  proof to the normalisation of:
  \begin{equation*}
    \input{figures/equations/c1_phase_forms_green.tikz}
  \end{equation*}
  On the right (in dashed lines) we recognise the composition of a Hadamard and
  a normal form, which we have already shown can be normalised. This done, we
  obtain the composition of a red spider and a normal form, which we have also
  already shown to be normalisable. Thus, we are done.

  As for the second normal form, we have:
  \begin{equation*}
    \input{figures/equations/c1_phase_forms_green_2.tikz}
  \end{equation*}
  and we have again reduced to previously solved cases.

  Unicity follows from the fact that none of these forms are equivalent, and
  there are therefore \(p^3(p^2 - 1)\) distinct forms, which matches the
  cardinality of the single qupit Clifford group.
\end{lproof}

\begin{lproof}[of proposition~\ref{prop:GS+LC}]
  Assume w.l.o.g. that the diagram \(D\) has type \(0 \to n\) (i.e.\, we
  consider that the diagram has already been turned into a state via
  equation~\eqref{eq:choi_isomorphism}), and furthermore that it contains no
  multipliers (equivalently, they have been unpacked into the sugarless calculus
  using equation~\eqref{eq:multiplier_explicit}). Then, transform \(D\) as
  follows:
  \begin{enumerate}
  \item Use \textsc{(Colour)} to change every red spider into a green one,
    surrounded by Hadamards. The resulting diagram consists of only green
    spiders and Hadamards.
  \item Use \textsc{(G-Elim)} to add a green spider between any two subsequent
    Hadamards. The diagram now consists of only green spiders, connected either
    by plain edges or H-edges.
  \item Use \textsc{(Spider)} to fuse any two spiders connected by plain edges,
    to eliminate any loops, and use lemma~\ref{lem:H_loop} to eliminate any
    H-edge loops. The resulting diagram contains no loops, and furthermore
    spiders are connected only by H-edges.
  \item Use the following rule from proposition~\ref{prop:weighted_hadamard}:
    \begin{equation}
      \input{figures/equations/hadamard_sum.tikz}
    \end{equation}
    to fuse all H-edges between two spiders into a single weighted H-edge or no
    edge.
  \item Use \textsc{(G-Elim)} and proposition~\ref{prop:weighted_hadamard} to
    obtain
    \begin{equation}
      \input{figures/equations/spider_graph_split.tikz}\quad,
    \end{equation}
    which allows one to split any spider connected to more than one output into
    several green spiders, each connected to exactly one output.
  \end{enumerate}
  The resulting diagram is nearly in GS+LC form, except it contains ``internal''
  vertices which are not connected to an output. Diagrams in this form have been
  called \emph{graph-like} in the literature. However, we can eliminate these
  internal vertices as follows. Let \(u\) be such a vertex,, and let \((a,b)\)
  be  its phase. If \(u\) has no neighbours, then it is an elementary scalar
  diagram which we can ignore.

  Otherwise, assume \(u\) has a neighbour \(v\) which is also connected to an
  output of the diagram. Then, we can eliminate the phase at \(u\) through the
  following manipulation. Throughout, we use local dilations
  (proposition~\ref{prop:local_scaling}) and local complementations
  (corollary~\ref{cor:GS_red_rotation}) on the rest of the graph freely. We can
  ignore the green spiders introduced by corollary~\ref{cor:GS_red_rotation}
  since they are absorbed into the phases of the graph vertices, and track only
  the red spiders introduced which are the only ``problematic'' spiders that
  take us away from GS+LC form.
  Then, we can eliminate vertex \(u\) as follows:
  \begin{equation*}
    \input{figures/equations/eliminate_internal.tikz}
  \end{equation*}
  where in the second step we have used the Euler decomposition of the Hadamard
  (lemma~\ref{lem:hadamard_euler}). The resulting green units are connected via
  a normal wire to a single other green spider, to which they can be fused
  using \textsc{(Spider)}.

  We still need to treat the case where \(u\) has only neighbours which are not
  connected to an output. In that case, letting \(v\) be any neighbour of \(u\),
  we can still eliminate \(u\) using a very similar argument and relying on the
  normal form for single-qupit operators. The first step is analogous to the
  previous case:
  \begin{equation*}
    \input{figures/equations/eliminate_internal_2.tikz}
  \end{equation*}
  Normalising the \(\mathscr{C}_1\)-diagram on \(v\) using
  proposition~\ref{prop:c1_completeness}, we obtain one of two forms which are
  essentially equivalent since we have
  \begin{equation*}
    \input{figures/equations/eliminate_internal_3.tikz}
    \quad \qand \quad
    \input{figures/equations/eliminate_internal_4.tikz}
  \end{equation*}
  for some \(u,v,s,t \in \Z_p\) and \(w \in \Z_p^*\). Then, pulling the
  multiplier and red spider into a local scaling and local complementation
  acting on the rest of the graph, we get:
  \begin{equation*}
    \input{figures/equations/eliminate_internal_5.tikz}
  \end{equation*}

  As a result, by the end of this procedure, we have eliminated \(u\) without
  reintroducing any new internal vertices, so we can repeat this process to
  eliminate each internal vertex in the diagram. The resulting diagram will then
  be in GS+LC form.
\end{lproof}

\begin{lproof}[of proposition~\ref{prop:rGS+LC}]
  By proposition~\ref{prop:GS+LC}, every stabiliser diagram is equivalent to
  some GS+LC diagram, and furthermore, the vertex operators can be brought into
  the normal form of proposition~\ref{prop:c1_completeness}. Using
  proposition~\ref{prop:local_scaling}, the multiplier part of this normal form
  can immediately be absorbed into a local scaling of the graph.

  We now need to eliminate the leftmost red spider of the normal form. Consider
  the vertex operator acting on vertex \(u\) of the graph. If \(u\) has no
  neighbours, we can use lemma~\ref{lem:unit_rotation_elim} to eliminate the red
  spider (up to a scalar in \(\mathbb{G}\)). Otherwise, \(u\) has at least one
  neighbour. In this case, we can use lemma~\ref{cor:GS_red_rotation} to
  ``copy'' the red spider into a green spider on each neighbour of \(u\).

  As in the qubit case, the set \(R\) is not stable under pre-composition by a
  green spider. However, by the proof of proposition~\ref{prop:c1_completeness},
  we see that, after normalisation, this operation applied to any element of
  \(R\) does not increase the number of red spiders. Thus, applying this
  procedure at most \(2\abs{V}\) times maps every vertex operator to one of the
  forms in \(R\).

  In order to prove the second part of the simplification, assume that the
  preceding procedure has been applied to the diagram, so that all vertex
  operators are in \(R\). Let \(u,v\) be neighbouring vertices such that the
  vertex operators of both \(u\) and \(v\) contain a red spider. We are going to
  use a sequence of local complementations at \(u\) and \(v\) to eliminate the
  red spiders on both \(u,v\). Since we have just shown that the green spider
  these operations enact on neighbours can be corrected and the vertex operators
  brought back to \(R\), and that local complementations and scalings preserve
  the GS+LC form, we ignore this action and consider only the part of the
  diagram connected to \(u\) and \(v\): 
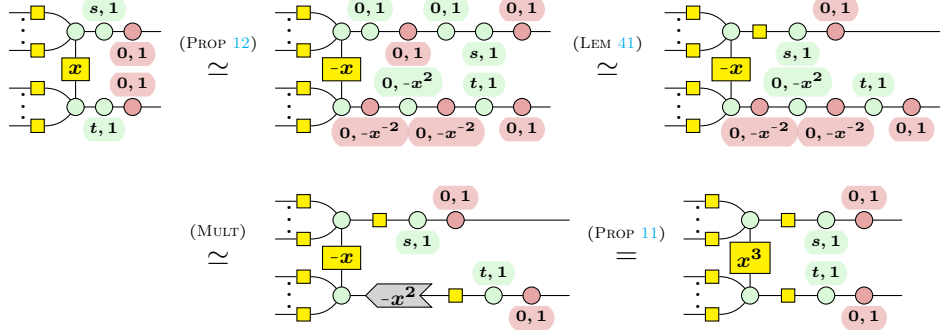
  \begin{equation*}
    \input{figures/equations/rGS+LC_paired_reds.tikz}
  \end{equation*}
  Then, normalising the vertex operators, we obtain
  \begin{equation*}
    \input{figures/equations/rGS+LC_paired_reds_2.tikz}
  \end{equation*}
  The red spiders can be eliminated using the previous strategy for mapping the
  vertex operators to \(R\), and so we have removed the pair of neighbouring red
  spiders. Repeating this process for each such pair renders a diagram in rGS+LC
  form.
\end{lproof}

\begin{lproof}[of proposition~\ref{prop:rGS+LC_simple}]
  Suppose this is not the case, i.e. there are vertices \(u,v\) such that \(u\)
  has a red vertex operator in \(A\) but not \(B\), \(v\) has a red vertex
  operator in \(B\) but not \(A\), and \(u,v\) are neighbours in \(A\). Then we
  perform a manipulation entirely analogous to the proof of
  proposition~\ref{prop:rGS+LC} in the diagram \(A\):
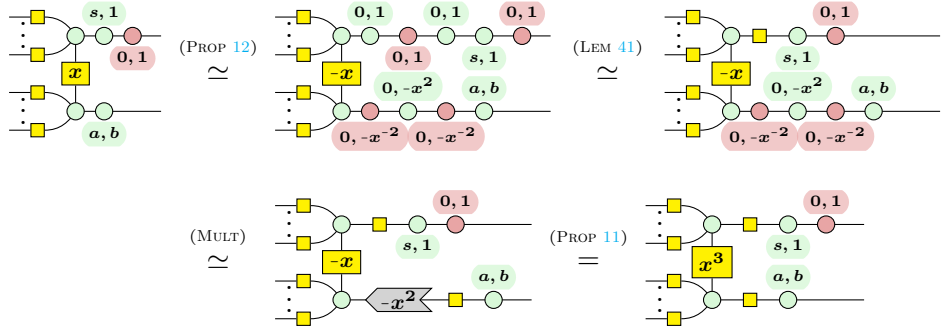
  \begin{equation*}
    \input{figures/equations/rGS+LC_simple.tikz}
  \end{equation*}
  Then, after normalisation (proposition~\ref{prop:c1_completeness}) and using
  proposition~\ref{prop:rGS+LC} to bring the diagram back to rGS+LC form, the
  vertex operator for \(u\) no longer contains a red vertex in \(A\) and the
  vertex operator for \(v\) contains a red vertex. Since we can do this for any
  pair of such red vertices in \(A\) and \(B\), it is clear that we can can
  simplify the pair of diagrams.
\end{lproof}

\begin{lproof}[of theorem~\ref{thm:scalar_completeness}]
  Let \(A \in \ZXp[0,0]\), then \(A\) consists of a tensor product of connected
  scalar diagrams. We first show that \(A\) can be rewritten to a tensor product
  of elementary scalars. If this is not already the case, pick some connected
  scalar diagram \(A'\) contained in \(A\). This diagram \(A'\) contains at
  least one red or green spider, from which we can extract a unit:
  \begin{equation}
    \input{figures/equations/scalar_normalisation_units.tikz}
  \end{equation}
  where \(B \in \ZXp[0,1]\). By theorem~\ref{thm:comp} and the normal form for
  the single-qupit Clifford group (proposition~\ref{prop:c1_completeness}), we
  can rewrite \(B\) to one of the forms
  \begin{equation}
    \input{figures/equations/scalar_normalisation_unit_forms.tikz}
  \end{equation}
  since we know that these forms cover all single qupit Clifford states.
  Simplifying a bit, we therefore only need to consider the following forms:
  \begin{equation}
    \label{eq:scalar_normalisation_all_scalars}
    \input{figures/equations/scalar_normalisation_all_scalars.tikz}
  \end{equation}
  The first of these is already an elementary scalar, and if \(s\) or \(t\) is
  \(0\), the second is also. If \(s=0\), the third diagram is also elementary.
  Thus, we can assume that \(t\neq 0 \neq s\) and treat both middle diagrams in
  one go:
  \begin{equation}
    \input{figures/equations/scalar_normalisation_green.tikz}
  \end{equation}
  and both of these diagrams can be normalised using \textsc{(M-Elim)} and
  \textsc{(Shear)}.
  Finally, the last scalar from
  equation~\eqref{eq:scalar_normalisation_all_scalars} can be rewritten to:
  \begin{equation}
    \input{figures/equations/scalar_normalisation_last.tikz}.
  \end{equation}

  We have shown that \(A\) can be rewritten to a tensor product of elementary
  scalars. If this tensor product contains a zero diagram, then the whole
  diagram equal to the zero diagram by proposition~\ref{prop:zero_normal_forms}.
  Otherwise, \(A\) can further be rewritten to a scalar in normal form using the
  multiplication rules for elementary scalars, which are given by
  \textsc{(M-One)}, lemmas~\ref{lem:scalar_omega_multiplication} and
  \ref{lem:scalar_gauss_multiplication}, up to applying \textsc{(Gauss)} to
  decompose quadratic scalars.
\end{lproof}

\begin{lproof}[of theorem~\ref{lem:bipartition}]
	
	First, $D$ can be put in GS-LC form. This allows to decompose $D$ in the following form:
	
	\begin{center}
		\input{figures/bipartite1.tikz}
	\end{center}
	
	Where $X$ and $Y$ are compositions of E-gates constituting the edges within the same partition together with local Clifford gates, and $G$ is a bipartite graph-state gathering all vertices having a neighbor in a different partition and all corresponding edges. Thus, we see that without loss of generality we can restrict to the case where $D$ is a bipartite graph-state.
	
	Assuming that $D$ is a bipartite graph-state with partitions of size $n$ and $m$ with $n\leq m$ we show how to find $A$ and $B$ such that:
	
	\begin{center}
		\input{figures/bipartite2.tikz}
	\end{center}
	
	Let's start with $n$ Bell's pair. Plugging Generalized Hadamard boxes we get the following graph state:
	
	\begin{center}
		\input{figures/bipartite3.tikz}
	\end{center}
	
	Now, let's say we want to add an edge with weight $w$ between the vertices $x$ and $y$ of this graph, to do so we first add an edge of weight $1$ between $x'$ and $y$ and apply the right local complementation to $x'$ to obtain the desired edge. Finally we clean up all the unwanted edges inside the partition by using the conveniant $E$-gates. The process is sketched in the following diagrams:
	
	\begin{center}
		\input{figures/bipartite4.tikz}
	\end{center}
	
	At the end, when all desired edges have been obtained, we still have to take care of the edges between $x$ and $x'$, but since we only kept vertices which are linked to the other partition in the final $D$, we know there is a neighbor $t$ that can be use to take care of this edge in the same way we used $x'$ before, first we add an edge between $t$ and $x'$, then we apply the needed local complementation on $t$ and finally we clean up the edges internal to the partitions.
	
	This give up the wanted bipartite graph-state $D$ which conclude the proof.
	
\end{lproof}

\begin{lproof}[of \ref{prop:marked}]
	Let's assume that there is a marked vertices $v$ in $C$ which is not marked in $D$. If $v$ is isolated in $C$ and not in $D$ then the two interpretations are obviously different. The same goes if $v$ is isolated in both diagram since then we can clearly differentiate the two corresponding states. So the only case remaining is when $v$ has neighbors in both diagrams.
	
	We then apply to both diagram the same unitary, first a green Clifford map on the chosen vertex and appropriate gates for each edges in the neighborhood of $v$ in $C$. In $C$ we have:
	
	\begin{center}
		\input{figures/testproof.tikz}
	\end{center}
	
	And in $D$:
	
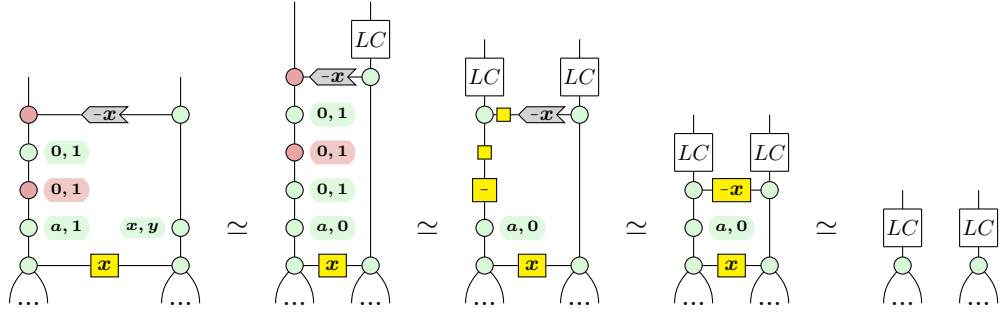
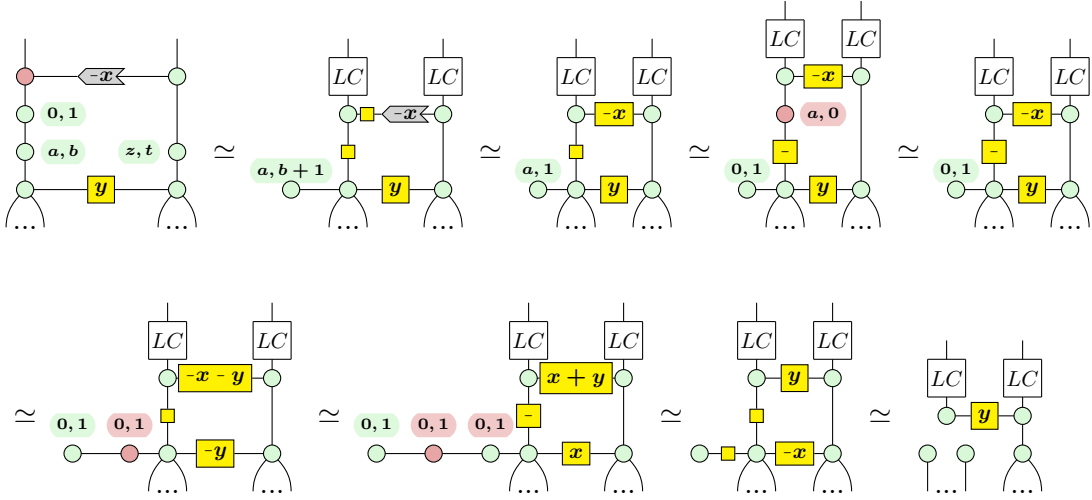
	\begin{center}
		\input{figures/testproof1.tikz}
	\end{center}
	
	Here we compute up to local Cliffords on the wires. Thus, we see that in $C$, $v$ is now disconnected while it isn't in $D$. Then $\interp{D}\neq \interp{C}$.
	
\end{lproof}

\begin{lproof}[of Theorem \ref{thm:comp}]
	We only show it for state by map-state duality.
	
	First we rewrite $A$ and $B$ into a simplified pair of rGS-LC diagrams $A'$ and $B'$. So, the marked vertices are the same. It only remain to show that the graph state part with . But if we apply the gates corresponding to the edges in $A'$ to both diagrams, then we obtain a completely disconected graph in $A'$. And the two interpretation can only be equals if $B'$ is also now completely disconnected. It then follows that the edges in both diagrams were the same.
	
	\begin{center}
		\input{completenessproof.tikz}
	\end{center}
	
	Finally, the only way for the interpretations to be the same is that all phases in the green vertices are equals, and then $A'$ can be rewritten in $B'$. So \(\Stab
	\vdash A = B\).
\end{lproof}

\end{document}